\documentclass[a4paper,11pt]{article}
\usepackage[top=1.25in, bottom=1.25in, left=1.05in, right=1.05in]{geometry}
\usepackage{setspace}		

\usepackage[english]{babel}
\usepackage[utf8]{inputenc}
\usepackage[T1]{fontenc}

\usepackage{natbib}
\usepackage{xr}
\usepackage{xr-hyper} 					% reference to labels in external .tex
\makeatother
\usepackage{enumitem}
\setlist{noitemsep}  % Reduce space between list items (itemize, enumerate, etc.)
\usepackage[section]{placeins}		% floats not to next section
\usepackage{amscd,amsmath,amssymb,amsfonts,amsthm,bbm,bm,latexsym,mathrsfs,mathtools}
\usepackage[subnum]{cases}	% multi-case equations with separate number for each
\usepackage{dsfont}

\RequirePackage[colorlinks=true, citecolor=blue, urlcolor=red]{hyperref}
\usepackage{booktabs}

\usepackage{xcolor,subfigure,epsfig,rotating,graphics,graphicx} 
\usepackage[width=1.2\linewidth, font={footnotesize}]{caption}

\usepackage{rotating,lscape}
\usepackage[flushleft]{threeparttable}
\usepackage{ragged2e,adjustbox}
\usepackage{longtable,float,array,multirow,colortbl}

\RequirePackage{hypernat}

\newtheorem{theorem}{Theorem}

\theoremstyle{remark}

\theoremstyle{plain}

\newcolumntype{C}[1]{>{\centering\arraybackslash}p{#1}}
\DeclareMathOperator{\Tr}{Tr}

\title{\vspace{-50pt} \textbf{Loss-based prior for the degrees of freedom of the Wishart distribution}%\thanks{\footnotesize
}

\author{Luca Rossini\thanks{University of Milan, Italy and Fondazione Eni Enrico Mattei. \color{blue}\texttt{luca.rossini@unimi.it}} \and Cristiano Villa\thanks{Duke-Kunshan University, Suzhou, China. \color{blue}}
\and Sotiris Prevenas\thanks{School of Mathematics, Statistics and Actuarial Science, University of Kent, United Kingdom. \color{blue}}
\and Rachel McCrea\thanks{Department of Mathematics and Statistics, Lancaster University, United Kingdom. \color{blue}}}

\begin{document}

\maketitle

\begin{abstract}
\noindent
{Motivated by the proliferation of extensive macroeconomic and health datasets necessitating accurate forecasts, a novel approach is introduced to address Vector Autoregressive (VAR) models. This approach employs the global-local shrinkage-Wishart prior.  Unlike conventional VAR models, where degrees of freedom are predetermined to be equivalent to the size of the variable plus one or equal to zero, the proposed method integrates a hyperprior for the degrees of freedom to account for the uncertainty about the parameter values.  Specifically, a loss-based prior is derived to leverage information regarding the data-inherent degrees of freedom. The efficacy of the proposed prior is demonstrated in a multivariate setting for forecasting macroeconomic data, as well as Dengue infection data.} \\

%The work of this paper is motivated by the growth of large macroeconomic and health data sets for which forecasts are desired.    
%In this paper, we propose a novel method to deal with Vector Autoregressive (VAR) models, when the {global-local shrinkage}-Wishart prior is considered. 
%In the usual VAR representation, the degrees of freedom are assumed {\em a-priori} to be equal to the size of the variable of interest plus one. Here, we propose to assign a hyperprior for the degrees of freedom in order to account for the uncertainty about the parameter values. In particular, we derive a loss-based prior that exploits the information about the number of degrees of freedom contained in the data. 
%We show evidence of the forecasting capabilities of the proposed prior in a multivariate setting for macroeconomic data and similar improved performance is also demonstrated in a forecasting analysis of the viral Dengue infection.\\

\end{abstract}

\section{Introduction}

{Macroeconomic and health forecasting has always been of pivotal importance for policymakers and researchers. In the twenty-first century, the availability of large data sets and novel modeling techniques have improved the forecasting of macroeconomic and financial variables \citep[see][]{Stock2002, DeMol2008, Banbura2010, Koop2013, Huber2019}. It is beneficial to extract relevant information about the processes of interest from the data, and this paper demonstrates how to leverage the information that is intrinsic to variables of interest, to provide an adaptable representation of prior uncertainty. 

A core model used by researchers and policymakers is the Vector Autoregressive (VAR) model, introduced in \cite{Sims1980}. The VAR model is ideal for representing the relationships between quantities of interest given its flexibility to accommodate multiple time series. Within the Bayesian framework, substantial research has been carried out to avoid over parametrization and overfitting \citep[e.g.][]{Doan1984,Litterman1986,Sims1998}, mainly through the definition of suitable shrinkage prior distributions, or by employing appropriate hierarchical structures \citep[e.g.][]{Banbura2010, Carriero2015, Carriero2019, Huber2019}. 

In this paper, we focus on the global-local shrinkage prior for the matrix of coefficients, and we rely on the Horseshoe prior distribution \citep{Carvalho2010}, which was introduced for VAR models by \cite{Follett2019} and recently, was used on the reduced form of the VAR matrix of coefficients by \cite{Bernardi2024} and \cite{Gruber2023}. Conventional approaches fix the degrees of freedom of the Wishart distribution to be equal to the size of the variable plus one \citep{Koop2010} or equal to zero \citep{Uhlig2005}. As stated in \cite{Koop2010}, this choice balances being informative and allowing uncertainty in the estimation of the covariance matrix, thus it is a trade-off between avoiding overfitting and allowing for flexibility. As far as we are aware, no investigation has been conducted regarding the validity of this assumption, and in this paper, we address this issue by assuming a hyperprior for the degrees of freedom based on the loss-based priors introduced in \cite{Villa2015}. We propose this solution to account for the uncertainty of the parameter values and derive a prior that exploits the information about the number of degrees of freedom contained in the data.

The loss-based hyperprior has its foundation in the measure of the loss that a researcher would incur if the wrong model was chosen, and it depends solely on the probability distribution used to model the data. In particular, it allows the researcher to not set a value for the degrees of freedom of the Wishart distribution and thus it runs the risk of impacting the posterior distribution and the model forecasts. Our proposed hyperprior is tested through different simulation scenarios and compared with alternative approaches, such as the flat prior proposed by \cite{Uhlig2005}.

We study the merits of our approach by considering the Federal Reserve Economic Data (FRED) dataset by \cite{McCracken2016} and \cite{McCracken2020}, and by performing both point and density forecasting measures on different macroeconomic variables. The proposed prior has two important advantages compared to fixing the degrees of freedom to a fixed value. Firstly, in the absence of any prior information about the true variance-covariance matrix, it is possible to integrate the process with suitable information contained in the data, making the approach more robust. Secondly, there is the possibility of evaluating the degrees of freedom across time via a rolling window estimation procedure. Indeed, we find strong evidence of changes in the degrees of freedom after the year 2000 and a fall around the 2009 financial crisis, which is strongly related to the Lehman Brothers failure. 

As a second case study, we analyse the Google Dengue Trends (GDT) dataset for ten different countries, which is a query-based reporting system for infectious disease \citep{Carneiro2009,Strauss2017}. Dengue is a viral infection transmitted by mosquitoes and is present in South American and Asian countries. 
The results from the second case study demonstrate the flexibility of the proposed approach.

The paper is structured as follows. Section \ref{sec_prel} describes the VAR model and derives the proposed loss-based prior for the degrees of freedom of the Wishart distribution. In Section \ref{sec_simu}, using simulation, we compare the proposed hyperprior to the model that assumes fixed degrees of freedom. Section \ref{sec_real} deals with the forecasting of macroeconomic variables, and Section \ref{sec_dengue} shows additional results on the forecasting of the Google Dengue Trends data. Section \ref{sec_concl} concludes the paper.
}

\section{VAR model for forecasting and the novel loss-based prior} \label{sec_prel}

{In this section, we define the VAR model, present the prior assumptions, and derive the novel loss-based hyperprior for the degrees of freedom of the Wishart distribution and its properness property.}
%We first define the time series model used for multivariate analysis, the VAR model, which will form the basis for the remainder of the paper. We then present the usual prior assumptions and derive the novel loss-based hyperprior for the degrees of freedom of the Wishart distribution and its properness property.

%See Huber and Feldkirch (2019,JBES) - Descrizione modello VAR 
%\subsection{Preliminaries on the model}
Let $\bm{y}_t$ be the $m$-dimensional vector of observations, for $t=1,\ldots,T$. {We define a VAR model with $p$ lags as}
\begin{equation}
\bm{y}_t = \sum_{j=1}^p A_j \bm{y}_{t-j} + \bm{\varepsilon}_t,\label{VAR}
\end{equation}
where $A_j$ is a $(m\times m)$ matrix of coefficients, and $\bm{\varepsilon}_t$  {is an $m$-dimensional vector of independent and identically normally distributed error terms centered on $0$, and with covariance matrix $\Sigma$, $\bm{\varepsilon}_t \sim \mathcal{N}(0,\Sigma)$.}
%We assume that $\bm{\varepsilon}_t = (\varepsilon_{1t},\ldots,\varepsilon_{mt})$ are i.i.d. with Gaussian distribution $\mathcal{N}(0,\Sigma)$.

{In a more compact form, Eq.~\eqref{VAR} can be written as}
\begin{equation*}
    Y=XA+E,
\end{equation*}
where Y is a $(T\times m)$ matrix constructed as $Y=(\bm{y}_1,\bm{y}_2,\dots,\bm{y}_T)'$, and $X=(\bm{x}_1,\bm{x}_2,\dots,\bm{x}_T)'$ is a $(T\times k)$ matrix containing the lagged response variables, where $\bm{x}_t=(\bm{y}'_{t-1},\bm{y}'_{t-2},\dots,\bm{y}'_{t-p})$. {Moreover, $A=(A_1,A_2,\dots,A_p)$ is a $(k\times m)$ matrix of coefficients and $E=(\bm{\varepsilon}_1,\dots,\bm{\varepsilon}_T)'$ is a $(T\times m)$ matrix of errors. In a vectorized form, we can define the VAR models of order $p$ as}
\begin{equation*}
\bm{y} = \left(I_m \otimes X\right) \bm{\alpha} + \bm{\varepsilon},
\end{equation*}
where $\bm{y} = \text{vec}(Y)$, $\bm{\alpha} = \text{vec}(A)$ and $\bm{\varepsilon} = \text{vec}(E)$ with distribution $\bm{\varepsilon} \sim \mathcal{N}(0,\Sigma \otimes I_T)$ and $\otimes$ is the Kronecker product.

In this paper, we adopt a {global-local shrinkage}-Wishart prior for the parameters of the model, where we assume a Horseshoe prior distribution \citep{Carvalho2010} for the matrix of coefficients %such as $\bm{\alpha} \sim \mathcal{N}(\underline{\bm{\alpha}},\underline{V})$, 
and a Wishart distribution for the precision matrix $\Sigma^{-1} \sim \mathcal{W}(\underline{\nu},\underline{S}^{-1})$. We set the hyperparameters for the Wishart prior equal to $\underline{\nu} = m+1$ and $\underline{S} = I_m$. As a robustness check, we ran a Normal-Wishart prior with hyperparameters for the Normal prior equal to $\underline{\bm{\alpha}} = \bm{0}$ and $\underline{V} = 10\cdot I_{mk}$ (the results are available upon request).
%of the Normal prior equal to $\underline{\bm{\alpha}} = \bm{0}$ and $\underline{V} = 10\cdot I_{mk}$, while for the Wishart prior, we set $\underline{\nu} = m+1$ and $\underline{S} = I_m$.

%posterior N-W
%{The advantageous property of the {Horseshoe}-Wishart prior is the possibility of having a conditional posterior distribution for the parameters in closed form. 
In particular, the general Horseshoe prior for each element of the vectorized matrix of coefficient $\bm{\alpha}$, takes the form
{
\begin{align*}
     \alpha_j| (\lambda_j^{\alpha})^2, (\tau^{\alpha})^2 &\sim \mathcal{N}(0, (\lambda_j^{\alpha})^2 (\tau^{\alpha})^2),\\
    \lambda_j^{\alpha} &\sim \mathcal{C}^{+}(0,1), \\
    \tau^{\alpha} &\sim \mathcal{C}^{+}(0,1),     
\end{align*}
}where $\mathcal{C}^{+}(\cdot, \cdot)$ denotes the half-Cauchy distribution, $\lambda_j^{\alpha}$ is the local shrinkage parameter, $\tau^\alpha$ is the global shrinkage parameter, and $j=1,\ldots, k\cdot m$. For the posterior distribution of $\alpha$, and of the global and local shrinkage parameters $\lambda_j^\alpha$ and $\tau^\alpha$, refer to \cite{Cross2020} and the algorithm proposed by \cite{Makalic2016}.

%On the other hand, the posterior distribution for the precision matrix is
%\begin{align*}
%\Sigma^{-1}|\bm{\alpha},\bm{y} \sim \mathcal{W}\left(\overline{\nu},\overline{S}^{-1}\right),
%\end{align*} 
%where $\overline{\nu} = \underline{\nu} + T$ and $\overline{S} = \underline{S} + \sum_{t=1}^T (\bm{y}_t- Z_t \bm{\alpha}) (\bm{y}_t - Z_t \bm{\alpha})'$.

\subsection{Loss-based hyperprior} \label{sec_loss}
{In the literature, the usual assumption on the degrees of freedom of the Wishart distribution is to set the parameter to $\nu = m +1$. In this section, we derive} the loss-based prior distribution for $\nu$. Thus, the parameter is assumed discrete, and, to construct the prior, we employ the objective method introduced in \cite{Villa2015}.

Let us consider a Bayesian model with sampling distribution $f(x|\theta)$, characterized by the discrete parameter $\theta \in \Theta$, and prior $\pi(\theta)$. A mass is assigned to each value of the parameter that is proportional to the Kullback--Leibler divergence between the model defined by $\theta$ and the nearest one. In other words, if $f(x|\theta)$ is the true model, and it is not chosen, then the loss in information that one would incur is represented by the Kullback--Leibler divergence between $f(x|\theta)$ and $f(x|\theta^\prime)$, where the latter is the nearest model to the true one.
Therefore, the prior on $\theta$ is
\begin{equation} \label{objpriorgeneral}
\pi(\theta)  \propto \exp \left\{ \min_{\theta^\prime\neq \theta \in \Theta} KL\left(f(x|\theta)\|f(x|\theta^\prime)\right)  \right\} -1,
\end{equation}
where $KL\left(f(x|\theta)\|f(x|\theta^\prime)\right)$ represents the Kullback--Leibler divergence between the two models. {A more detailed derivation of the prior in Eq.~\eqref{objpriorgeneral} is illustrated in Appendix \ref{LBprior}.}

To derive the prior in Eq.~\eqref{objpriorgeneral} for $\nu$, the Kullback--Leibler divergence between two Wishart distributions that share the same scale matrix $V$ and differ in the number of degrees of freedom, say $W_\nu$ and $W_{\nu+c}$ is required. The probability density function of a Wishart distribution with parameters $V$ and $\nu$ is given by: 
\begin{align*}
W(X|V,\nu)=\frac{|X|^{(\nu-m-1)/2}\exp(-\Tr(V^{-1}X)/2)}{2^{\frac{\nu m}{2}}|V|^{\nu/2}\Gamma_m(\nu/2)}, 
\end{align*}
where $\Gamma_m(\cdot)$ is the multivariate Gamma function, and $\Tr(\cdot)$ is the trace function.
Thus the Kullback--Leibler divergence between two Wishart distributions is given by:
\begin{align*}
    KL(W_\nu\|W_{\nu+c}) = \log\left\{\frac{\Gamma_m\left(\frac{\nu+c}{2}\right)}{\Gamma_m\left(\frac{\nu}{2}\right)}\right\} - \frac{c}{2}\psi_m\left(\frac{\nu}{2}\right),
\end{align*}
where $\psi_m(\cdot)$ is the multivariate digamma function defined as $\psi_m(x)=\sum_{i=1}^{m}\psi(x+(1-i)/2)$, $\psi(x)=\Gamma'(x)/\Gamma(x)$ is the digamma function, and $c \in \mathbb{Z}$. As $ KL(W_\nu\|W_{\nu+c})$ is a convex function of $c$, and its global minimum is at $c=0$, the nearest Wishart distribution to $W_\nu$ will be $W_{\nu+c}$ for either $c=-1$ or $c=1$. Theorem~\ref{teo1} shows that the Kullback--Leibler divergence between $W_\nu$ and $W_{\nu+c}$ is minimized for $c=1$.
\begin{theorem}\label{teo1}
Consider two Wishart distributions, $W_\nu$ and $W_{\nu+c}$, with the same scale matrix and $\nu$ and $\nu+c$ degrees of freedom, respectively, and $c\neq0$ is an integer. Then, the Kullback--Leibler divergence between $W_\nu$ and $W_{\nu+c}$  is minimum for $c=1$.
\end{theorem}

\begin{proof}
For $c=1$, the Kullback--Leibler divergence is 
\begin{align*}
   KL(W_\nu\|W_{\nu+1}) = \log \Gamma_m\left(\frac{\nu+1}{2}\right) - \log\Gamma_m\left(\frac{\nu}{2}\right) - \frac{1}{2}\psi_m\left(\frac{\nu}{2}\right),
\end{align*}
while for $c=-1$, we obtain
\begin{align*}
    KL(W_\nu\|W_{\nu-1}) = \log \Gamma_m\left(\frac{\nu-1}{2}\right) - \log\Gamma_m\left(\frac{\nu}{2}\right) +\frac{1}{2}\psi_m\left(\frac{\nu}{2}\right).
\end{align*}
By taking the difference of the two divergences, we yield 
\begin{align}
   KL(W_\nu\|W_{\nu+1})-KL(W_\nu\|W_{\nu-1}) & = \log \Gamma_m\left(\frac{\nu+1}{2}\right)-\log \Gamma_m\left(\frac{\nu-1}{2}\right) -\psi_m\left(\frac{\nu}{2}\right) \\ 
   & = \log\left[\frac{\Gamma(\nu)}{2^m\Gamma(\nu-m)}\right] -\psi_m\left(\frac{\nu}{2}\right).\label{eq:number8}
\end{align}
%& =  \log\frac{\Gamma_m(\frac{\nu+1}{2})}{\Gamma_m(\frac{\nu-1}{2})}-\psi_m\Big(\frac{\nu}{2}\Big)\\
%& = \log \frac{\prod_{j=1}^{m}\Gamma(\frac{\nu+2-j}{2}) }{\prod_{j=1}^{m}\Gamma(\frac{\nu-j}{2})}-\psi_m\Big(\frac{\nu}{2}\Big)\\
%& =  \log\frac{\Gamma(\frac{\nu+1}{2})\Gamma(\frac{\nu}{2})\Gamma(\frac{\nu-1}{2})\dots \Gamma(\frac{\nu+2-m}{2})}  {\Gamma(\frac{\nu-1}{2})\Gamma(\frac{\nu-2}{2})\Gamma(\frac{\nu-3}{2})\dots \Gamma(\frac{\nu-m}{2})} -\psi_m\Big(\frac{\nu}{2}\Big)\\
%& = \log\frac{\Gamma(\frac{\nu+1}{2})\Gamma(\frac{\nu}{2})}  {\Gamma(\frac{\nu+1-m}{2})\Gamma(\frac{\nu-m}{2})} -\psi_m\Big(\frac{\nu}{2}\Big)\\
%The last step is due to the Legendre duplication formula: $\Gamma(z) \ \Gamma\left(z + \frac{1}{2}\right) = 2^{1-2z} \ \sqrt{\pi} \ \Gamma(2z)$ for $z=\frac{\nu}{2}$ in the numerator and $z=\frac{\nu-p}{2}$ for the denominator.
We will prove that Eq.~\eqref{eq:number8} is always negative for any $\nu,m$ such that $\nu>m\geq2$. As $\nu>m$, $\nu=m+k$, for $k=1,2,\dots$, and  $m\geq2$, thus we obtain that the minimum Kullback--Leibler divergence is achieved at $c=1$ if
\begin{equation}\label{ineq1}
\log\left\{\frac{\Gamma(m+k)}{2^m\Gamma(k)}\right\} < \psi_m\left(\frac{m+k}{2}\right).
\end{equation} 
To prove the inequality in Eq.~\eqref{ineq1}, we rely on
\begin{equation}\label{psidef}
\psi_{m+1}\left(\frac{\nu+1}{2}\right) = \psi_{m}\left(\frac{\nu}{2}\right) + \psi \left(\frac{\nu+1}{2}\right),
\end{equation} 
and
\begin{equation}\label{psiineq}
\log\left(x-\tfrac{1}{2}\right) < \psi(x),
\end{equation} 
where Eq.~\eqref{psidef} comes from the definition of the multivariate digamma function, and inequality~\eqref{psiineq} can be deduced from $\log(x + \tfrac{1}{2}) - \frac{1}{x} < \psi(x) < \log(x + e^{-\gamma}) - \frac{1}{x}$, \citep[see][]{elezovic2000best} for $x>\tfrac{1}{2}$ and with $\gamma$ equal to $0.57721$ as the Euler-Mascheroni constant.

Initially, we assume it holds for a particular $m$, and then we prove for $m+1$:
\begin{align*} 
\log\left\{\frac{\Gamma(m+1+k)}{2^{m+1}\Gamma(k)}\right\} &< \psi_{m+1}\left(\frac{m+1+k}{2}\right) \\ 
\log\left\{\frac{\Gamma(m+k)(m+k)}{2^{m}\Gamma(k)2}\right\} &< \psi_{m}\left(\frac{m+k}{2}\right) + \psi\left(\frac{m+1+k}{2}\right)\\
\log\left\{\frac{\Gamma(m+k)}{2^{m}\Gamma(k)}\right\} + \log\left\{\frac{m+k}{2}\right\} &< \psi_{m}\left(\frac{m+k}{2}\right) + \psi\left(\frac{m+1+k}{2}\right),
\end{align*}
where in the last inequality we have $\log\frac{m+k}{2}<\psi\Big(\frac{m+1+k}{2}\Big)$ as a consequence of the result in inequality~\eqref{psiineq}. Thus, if the inequality~\eqref{ineq1} holds for $m$, then it holds for $m+1$, and it holds for any $k$. The smallest possible value $m$ is equal to $m=2$, where
\begin{align*} 
\log\left\{\frac{\Gamma(2+k)}{2^2\Gamma(k)}\right\} &< \psi_2\left(\frac{2+k}{2}\right)\\
\log\left\{\frac{(k+1)k}{2^2}\right\} &< \psi\left(\frac{2+k}{2}\right)+\psi\left(\frac{1+k}{2}\right),
\end{align*}
where we have $\log\left(\frac{k}{2}\right) < \psi\left(\frac{k+1}{2}\right)$ and $\log\left(\frac{k+1}{2}\right) < \psi\left(\frac{k+2}{2}\right)$ due to inequality~\eqref{psiineq}. Thus, inequality~\eqref{ineq1} holds for $m=2$ and, subsequently, it holds for any $m$. 
\end{proof}

We can then define the objective prior distribution for $\nu$ %, following Eq.~\eqref{objpriorgeneral}, 
as 
\begin{equation} \label{objprior}
    \pi(\nu)\propto \frac{\Gamma\left(\frac{\nu+1}{2}\right)}{\Gamma\left(\frac{\nu+1-m}{2}\right)}e^{-\frac{1}{2}\sum_{i=1}^{m}\psi\left(\frac{\nu+1-i}{2}\right)}-1.
\end{equation}

{Figure~\ref{fig:priorplots} shows the loss-based prior for $\nu$ for three different values of $m\in \{3,7,15\}$, which represents the dimensionality of the macroeconomic dataset that we analyse in Section \ref{sec_Data_Macro}. Each line yields similar patterns and shows that the loss-based prior probability distribution decreases as the degrees of freedom $\nu$ increases.}

%{To have a feeling of the loss-based prior for $\nu$, in Figure \ref{fig:priorplots}, we have plotted it for different values of $m$, in particular for $m=3, 7, 15$, which are the dimensions of the three data sets used in Section \ref{sec_real} for forecasting macroeconomic data. As one would expect for a positive parameter, the prior probability decreases as $\nu$ increases.}

\begin{figure}[h!]
    \centering
    \includegraphics[scale=0.5]{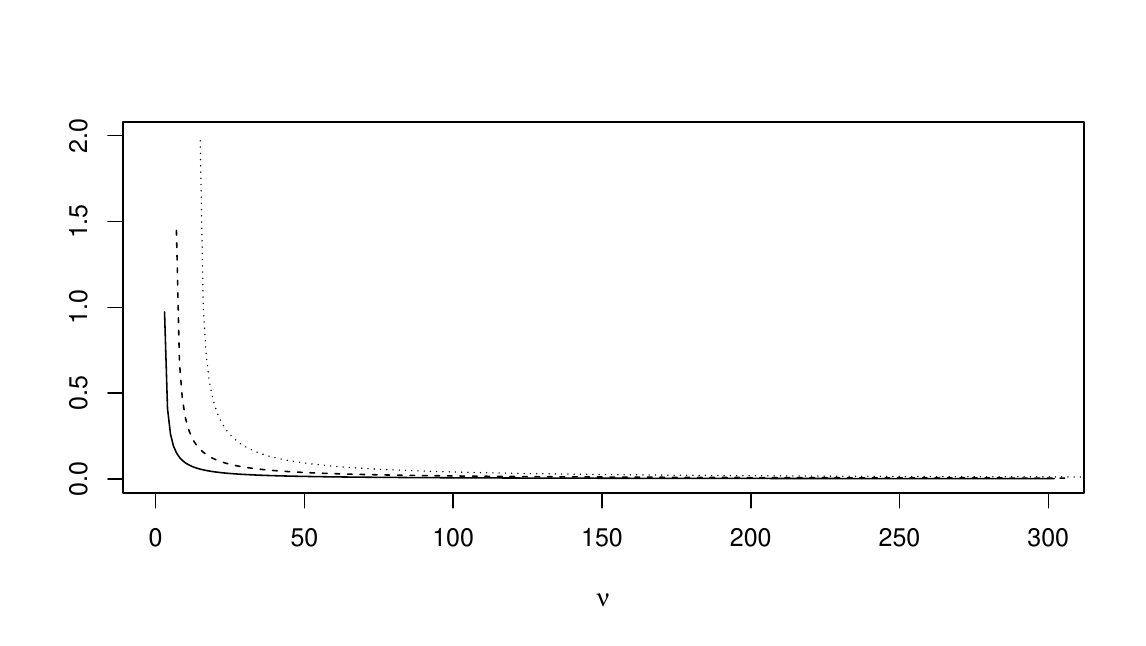}
    \caption{Loss-based prior distribution (unnormalised) for $\nu$ for dimensionality $m=3$ (continuous line), $m=7$ (dashed line), and $m=15$ (dotted line).}
    \label{fig:priorplots}
\end{figure}

\subsection{Properness of the posterior for $\nu$}\label{sc_posterior}
Let us assume that we observe one random matrix $\Sigma^{-1}$ from the Wishart $W(S_0^{-1},\nu)$, then, the likelihood function is given by
\begin{align} \label{loglik}
        p(y|\alpha,\Sigma) &= (2\pi)^{-\frac{mT}{2}}|\Sigma|^{-\frac{T}{2}} \exp\left\{ -\frac{1}{2}\left[y-(I_m \otimes X) \alpha\right]'\left(\Sigma^{-1}\otimes I_T\right)\left[y-(I_m \otimes X) \alpha\right] \right\} \nonumber \\
        &\propto |\Sigma|^{-\frac{T}{2}} \exp\left\{ -\frac{1}{2} \Tr\left[ (Y-XA)'(Y-XA)\Sigma^{-1} \right]  \right\}.
\end{align}
Using the loss-based prior for $\nu$ in Eq.~\eqref{objpriorgeneral}, we obtain the posterior distribution for the number of degrees of freedom as
\begin{align} \label{posterior1}
p(\nu|\Sigma^{-1}) & \propto \pi(\nu) \pi\left(\Sigma^{-1}|\nu,S_0^{-1}\right) \nonumber \\
 & \propto \left\{ \frac{\Gamma\left(\frac{\nu+1}{2}\right)}{\Gamma\left(\frac{\nu+1-m}{2}\right)}e^{-\frac{1}{2}\sum_{i=1}^{m}\psi \left( \frac{\nu+1-i}{2} \right)} - 1 \right\} \left\{\frac{|\Sigma^{-1}|^{(\nu-m-1)/2}e^{\left(-\Tr(S_0\Sigma^{-1})/2\right)}}{2^{\frac{\nu m}{2}}|S_0^{-1}|^{\nu/2}\Gamma_m\left(\frac{\nu}{2}\right)} \right\}.
\end{align}

Theorem~\ref{teo2} shows that the marginal posterior distribution for $\nu$ is proper. This is a necessary step when deriving objective priors.

\begin{theorem}\label{teo2}
The posterior distribution for the number of degrees of freedom $\nu$ in Eq.~\eqref{posterior1} is proper.
\end{theorem}

\begin{proof}
We prove {by using Abel's test of convergence} that $\sum_{\nu=m}^{\infty}p(\nu|\Sigma^{-1})<\infty$. The sequence $\{\pi(\nu)\}$ is bounded {since} $\pi(m)>\pi(\nu)>0$ and is also monotone (decreasing). We show that $\sum_{\nu=m}^{\infty}\pi(\Sigma^{-1}|\nu,S_0^{-1})<\infty$ by using the ratio test
%\[\sum_{n=p}^{\infty}\frac{|\Sigma^{-1}|^{(n-p-1)/2}\exp(-tr(S_0^{-1})/2)}{2^{\frac{np}{2}}|S_0|^{n/2}\Gamma_p(\frac{n}{2})}<\infty% \]
\begin{align*} 
R_\nu & = \frac{|\Sigma^{-1}|^{(\nu-m)/2}}{2^{\frac{(\nu+1)m}{2}}|S_0^{-1}|^{(\nu+1)/2}\Gamma_m\left(\frac{\nu+1}{2}\right)} \times \frac{2^{\frac{\nu m}{2}}|S_0^{-1}|^{\nu/2}\Gamma_m\left(\frac{\nu}{2}\right)}{|\Sigma^{-1}|^{(\nu-m-1)/2}}\\
& = \frac{|\Sigma^{-1}|^{1/2}\Gamma_m(\frac{\nu}{2})}{2^{m/2}|S_0^{-1}|^{1/2}\Gamma_m\left(\frac{\nu+1}{2}\right)}\\
& = \frac{|\Sigma^{-1}|^{1/2}\Gamma\left(\frac{\nu+1-m}{2}\right)}{2^{m/2}|S_0^{-1}|^{1/2}\Gamma\left(\frac{\nu+1}{2}\right)}.
\end{align*}
{Thus, we have}
$$\lim_{\nu\to\infty}\left\{\frac{\Gamma(\frac{\nu+1-m}{2})}{\Gamma(\frac{\nu+1}{2})}\right\} = 0,$$
so $\lim_{\nu\to\infty}R_\nu=0$ \citep{Abramowitz1972}, therefore, the posterior {distribution} for $\nu$ is a proper distribution.
\end{proof}

\subsection{Gibbs sampling algorithm}

Based on the {proposed loss-based} prior distribution, we provide a Gibbs sampler algorithm. In particular, the Monte Carlo Markov Chain (MCMC) algorithm follows the usual steps seen in the multivariate time series literature \citep[see,][]{Koop2010, Follett2019}, {where the new step relates to the full conditional posterior distribution of $\nu$}. We {implement} a Metropolis-Hastings algorithm since the posterior distribution of $\nu$ is not available in closed form.

To summarize, the Gibbs sampler is based on the following steps:
\begin{itemize}
   {\item[(i)] Update the vectorized matrix of coefficients, $\bm{\alpha}$ given the data $\textbf{y}$ and $\Sigma^{-1}$ by using the corrected triangular algorithm of \cite{Carriero2022}.}
    \item[(ii)] Update the precision matrix $\Sigma^{-1}$ given $\bm{\alpha}, \textbf{y}$ and $\nu$ from a Wishart distribution.
    {\item[(iii)] Update the local and global shrinkage parameters $\lambda_j^{\alpha}$ and $\tau^\alpha$ given the vectorized matrix of coefficients, $\bm{\alpha}$, as in \cite{Makalic2016}}.
    \item[(iv)] Update the degree of freedom $\nu$ given the $\Sigma^{-1}$ by using a Metropolis-Hastings algorithm with a symmetric random walk proposal.
 \end{itemize}
{For the model with fixed $\nu$ equal to $0$ or $m+1$, the Gibbs sampler is based only on Steps (i)--(iii).}
% Remember that in the benchmark model, the Gibbs sampler is based only on Steps (i)- (iii) since the $\nu$ parameter is known.

\section{Simulation study} \label{sec_simu}

We compare the performance of our loss-based hyperprior in different simulation studies {to the case of fixed $\nu$ equal to $0$ or $m+1$}. We {generate the data from a VAR model with one lag with different matrix dimensions and time lengths, where the elements of the matrix of coefficients of the VAR model are drawn from a $U(-0.95,0.95)$ distribution and then stationarity conditions are checked. We consider small, medium, and large numbers of response variables, where $m$ is equal to $5$, $10$, and $20$, respectively, and time, $T$, is equal to $30$ or $100$. Moreover, we have considered different combinations for the choice of the degrees of freedom when generating the data: for each dimension $m$, we have chosen $\nu$ equal to $\{5,10,15\}$, for $m=5$, $\{10,15,20\}$, for $m=10$, and $\{20,24,26\}$, for $m=20$. In Appendix~\ref{App_A}, we generate data as in the macroeconomic application with a time dimension equal to $240$ and the number of response variables equal to $3$, $7$, and $15$. The choice of $\nu$ is $\{3,5,7\}$ for $m=3$; $\{7, 10, 13\}$ for $m=7$ and $\{15, 20, 25\}$ for $m=15$.

For the comparative approaches, we use a flat prior as in \cite{Uhlig2005}, or we treat the degrees of freedom as fixed at $\nu = m+1$. We also assume an identity matrix for the prior scale matrix of the Wishart distribution. We assume a Horseshoe prior for the matrix of coefficients since the interest of our simulation experiment is in the evaluation of the covariance matrix. %In formulae, we have that the prior for the matrix of coefficients is a Normal distribution with hyperparameters $\underline{\bm{\alpha}}$ and $\underline{V}$, while the covariance matrix follows an inverse Wishart with scale matrix equal to $\underline{S}$ and degrees of freedom equal to {$\nu = 0$ (in the first scenario)}, $\nu = m+1$ (in the second scenario) and $\nu \sim \pi(\nu)$ (in the third scenario).
}

For each dataset, and each posterior sample, we estimate the posterior means of the {vectorized matrix of coefficients and of the covariance matrix $\Sigma$, and compute the Root Mean Absolute Deviation (RMAD)} between the posterior means and the true parameter values as
\begin{align*}
    RMAD = \bigg[ \frac{1}{N}\sum_{i=1}^{N}|\theta-\hat{\theta}| \bigg]^\frac{1}{2},
\end{align*}
where $N$ is the number of parameters estimated (which is equal to $m^2$ for the covariance matrix and depends on the lags for the matrix of coefficients) and $\theta$ is the matrix of coefficients or the covariance matrix. 

{This process is repeated $250$ times, and the Gibbs sampler is run for $6000$ iterations with a burn-in of $1000$ iterations. For each batch of RMAD, we create boxplots for the cases with fixed $\nu$ equal to $0$ or $m+1$, and the case with loss-based prior for the degrees of freedom. %In particular, since we are interested in the covariance matrix, we have reported the RMADs for the covariance matrices for each dataset. 
Figure~\ref{Fig_Sim_M5} shows the results for the RMAD for the case with $m=5$ and $T=30$, where the left panel explains the results when the data are generated with $\nu = 5$, the center with $\nu = 10$, and the right with $\nu = 15$.}

\begin{figure}[h!]
%	{\includegraphics[width=0.33\linewidth]{Figures_Revision/RootMAD,T=30,M=5,dof=5_Rev}}\hfill%
%	{\includegraphics[width=0.33\linewidth]{Figures_Revision/RootMAD,T=30,M=5,dof=10_Rev}}\hfill%
%	{\includegraphics[width=0.33\linewidth]{Figures_Revision/RootMAD,T=30,M=5,dof=15_Rev}} \\
  	{\includegraphics[width=0.33\linewidth]{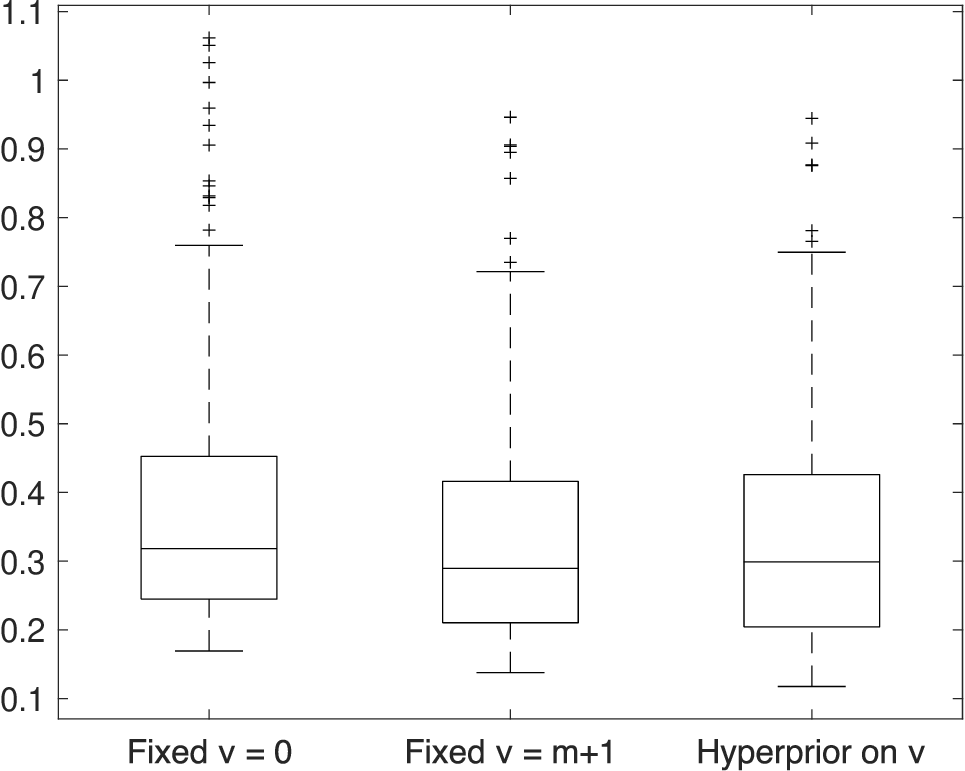}}\hfill%
	{\includegraphics[width=0.33\linewidth]{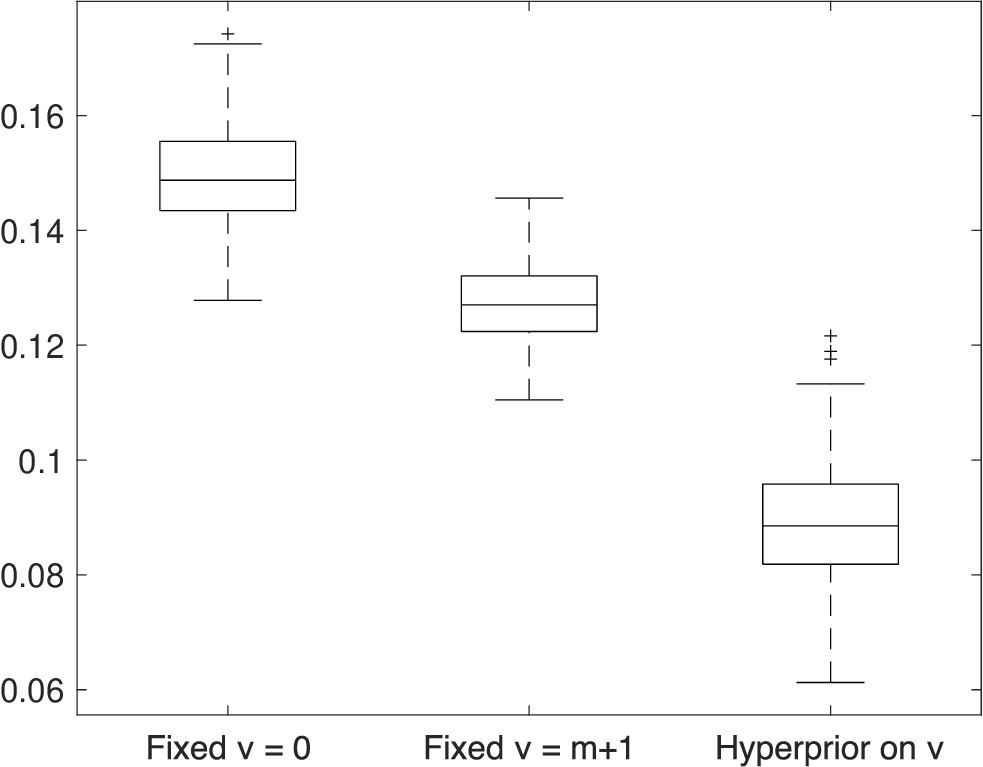}}\hfill%
	{\includegraphics[width=0.33\linewidth]{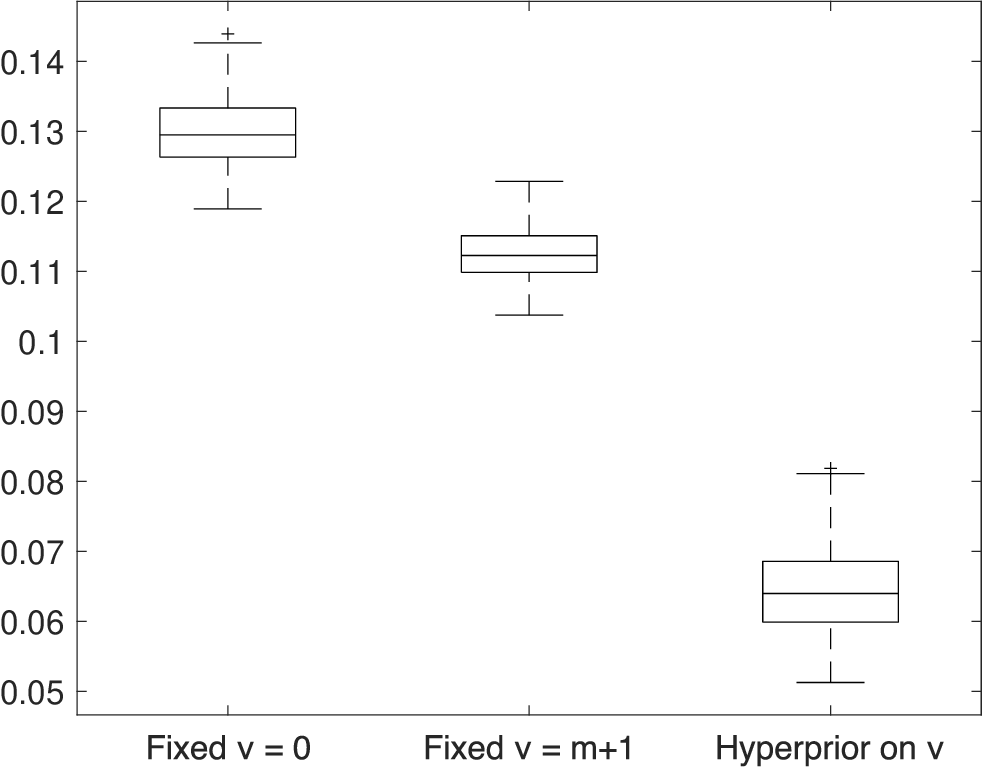}}
%	\bigskip
%	{\includegraphics[width=0.33\linewidth]{Figures/RootMAD,T=100,M=5,dof=5}}\hfill%
%	{\includegraphics[width=0.33\linewidth]{Figures/RootMAD,T=100,M=5,dof=10}}\hfill%
%	{\includegraphics[width=0.33\linewidth]{Figures/RootMAD,T=100,M=5,dof=15}}
	\caption{{Monte Carlo Simulation - Root Mean Absolute Deviations (RMAD) of the covariance matrices. These distributions are obtained by simulating $250$ VAR(1) with dimension $m=5$ and sample size $T=30$. Results are reported for data generated from a Wishart distribution with $\nu = 5$ (left), $\nu=10$ (center), and $\nu = 15$ (right).}}
	\label{Fig_Sim_M5}
\end{figure}

{From Figure~\ref{Fig_Sim_M5}, we observe that differences in RMAD between the alternative priors increase when increasing $\nu$. The left panel shows no difference between the three cases, except for the outliers, which are smaller for our proposed hyperprior. Increasing $\nu$ to $10$ and $15$ leads to smaller RMAD when our hyperprior is used.} %The results are bigger when we evaluate our hyperprior with $\nu$ equal to $15$ and $\nu$ equal to $5$.

{As a second measure of evaluation of the proposed prior, in Table~\ref{tab:IRF_Sim_M5} we provide the RMAD for the orthogonal (Choleski) impulse response function evaluated at four different horizons ($h = 1, 3, 5$ and $7$). In particular, the column called Fixed $\nu=0$ refers to the \cite{Uhlig2005} flat prior and it is considered as the benchmark, and the other two columns (called Fixed $\nu = m+1$ and Hyperprior) provide the ratio with respect to the benchmark. If a value is greater than $1$, it means that the flat prior outperforms the other priors, while if the value is lower than $1$, the flat prior shows poorer performance. In Table~\ref{tab:IRF_Sim_M5}, increasing $\nu$ from $5$ to $15$ leads to strong improvements of around $12\%$ at horizon $3$ and $27\%$ at horizon $5$ for the proposed loss-based prior against the other priors. When $\nu = 5$, the differences between the hyperprior and the flat prior are small except for horizon $5$ and $7$, while for $\nu = 10$, the improvement ranges from 7\% at horizon $3$ to $31\%$ at horizon $7$.}

\begin{table}[h!]
    \centering
    \begin{adjustbox}{width=1\textwidth,center=\textwidth}
    \begin{tabular}{c|c c c| c c c| c c c}
    \hline 
    & \multicolumn{3}{c}{$\nu = 5$} & \multicolumn{3}{c}{$\nu = 10$} & \multicolumn{3}{c}{$\nu = 15$} \\
        Horizon & Fixed & Fixed & Hyperprior & Fixed & Fixed & Hyperprior & Fixed & Fixed & Hyperprior  \\
         & $\nu =0$ & $\nu = m+1$ & & $\nu =0$ & $\nu = m+1$ & & $\nu =0$ & $\nu = m+1$ \\
         \hline
         1 &  0.5700   &  0.9950  &   0.9940 &     0.3370   &  0.9990  &   0.9950 &   0.3030  &  1.0000  &  0.9980\\
    3 & 0.6400   &  1.0070  &   1.0070 &    0.4430   &  0.9780  &   0.9320 &     0.4890 &   0.9780 &   0.8880\\
   5 &  0.7260   &  0.9710  &   0.9610 &    0.9350   &  0.9350  &   0.7910 &    1.3200   & 0.9440  &  0.7300\\
    7 & 1.2780   &  0.9150  &   0.8950 &    2.4220   &  0.8960  &   0.6910 &    4.0800  &  0.9170  &  0.6150\\
         \hline 
    \end{tabular}
    \end{adjustbox}
    \caption{{Monte Carlo Simulation - RMAD of the impulse response functions (IRF) for four horizons $h=1, 3, 5$ and $7$ by simulating $250$ VAR(1) with dimension $m=5$ and sample size $T=30$. Column Fixed $\nu=0$ provides the RMAD of the IRF, while Columns Fixed $\nu = m+1$ and Hyperprior provide the ratio between the referred priors and the flat prior.}}
    \label{tab:IRF_Sim_M5}
\end{table}

These results are also confirmed in high dimensional cases as shown in Figures~\ref{Fig_Sim_M10} and \ref{Fig_Sim_M20} for the ten-dimensional and twenty-dimensional cases, respectively. In Figure~\ref{Fig_Sim_M10}, we compare our loss-based hyperprior with the fixed $\nu$ {equal to $0$ and $m+1$ for the data generated from a Wishart with degrees of freedom equal to $10$ (left panel), $15$ (center), and $20$ (right). In this scenario, the results demonstrate that our loss-based prior is an improvement over fixed $\nu$ for the case of $15$ and $20$ degrees of freedom, while for $\nu$ equal to $10$, we have small differences between the three prior representations.}

\begin{figure}[h!]
% 	{\includegraphics[width=0.33\linewidth]{Figures_Revision/RootMAD,T=30,M=10,dof=10_HS_Rev}}\hfill%
%	{\includegraphics[width=0.33\linewidth]{Figures_Revision/RootMAD,T=30,M=10,dof=15_HS_Rev}}\hfill%
%	{\includegraphics[width=0.33\linewidth]{Figures_Revision/RootMAD,T=30,M=10,dof=20_HS_Rev}} \\
 	{\includegraphics[width=0.33\linewidth]{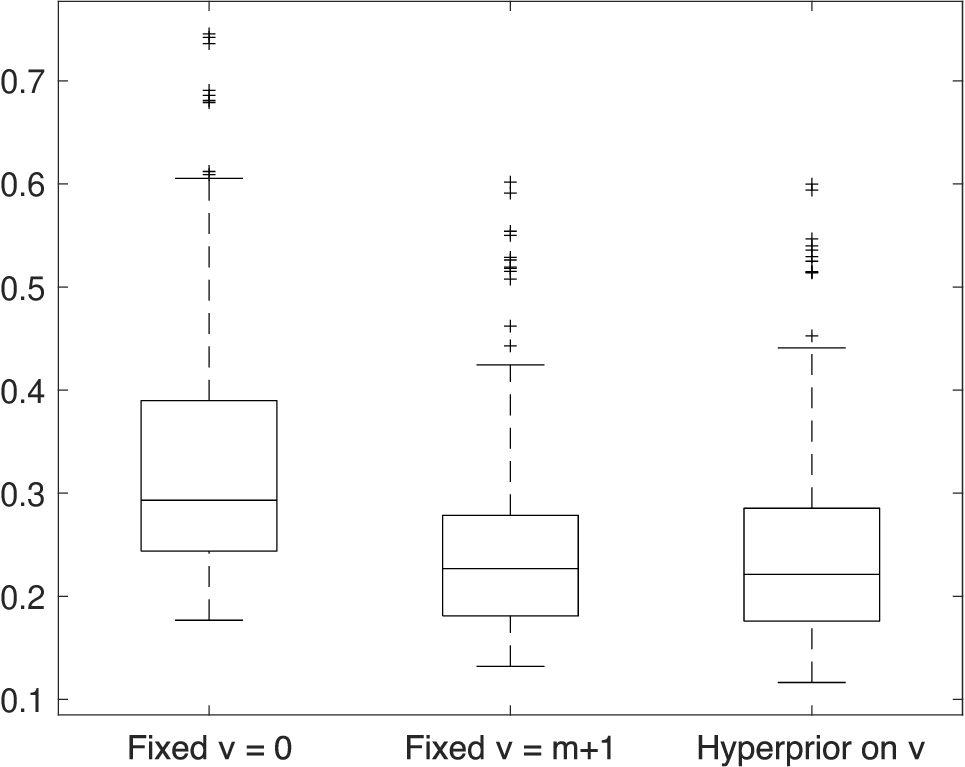}}\hfill%
	{\includegraphics[width=0.33\linewidth]{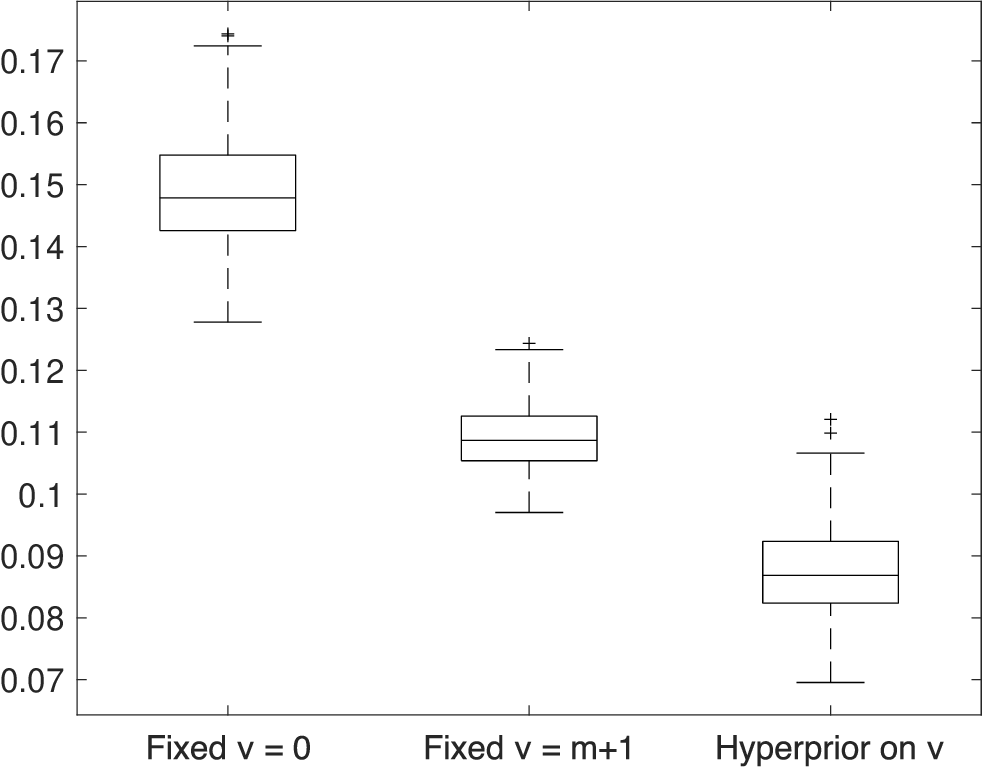}}\hfill%
	{\includegraphics[width=0.33\linewidth]{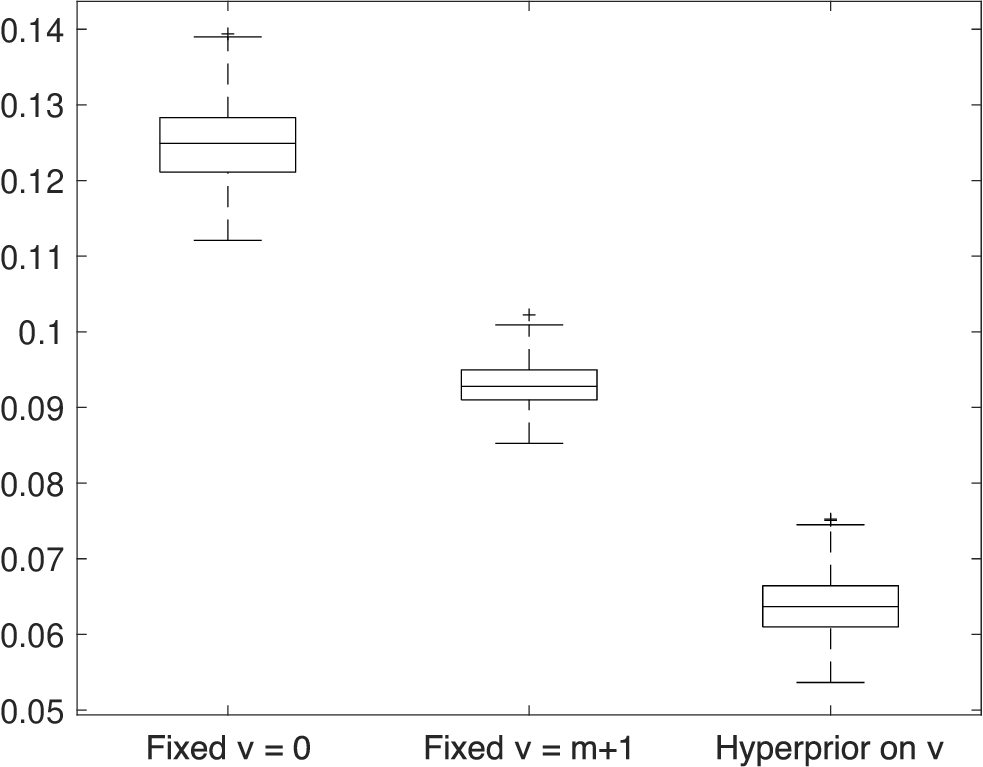}}
	\caption{{Monte Carlo Simulation - RMAD of the covariance matrices. These distributions are obtained by simulating $250$ VAR(1) with dimension $m=10$ and sample size $T=30$. Results are reported for data generated from a Wishart distribution with $\nu = 10$ (left), $\nu=15$ (center), and $\nu = 20$ (right).}}
 %Monte Carlo Simulation - Root Mean Absolute Deviations of the covariance matrices of dimension $m=10$. These empirical distributions are obtained by simulating $250$ VAR(1) of sample size $T=30$. Results are reported separately for data generated from a Wishart with $\nu= 10$ (left), $\nu =15$ (center), and $\nu = 20$ (right).}
	\label{Fig_Sim_M10}
\end{figure}

Table~\ref{tab:IRF_Sim_M10} provides the RMAD of the IRFs for different horizons for the scenario with $m=10$. These results confirm conclusions obtained from Figure~\ref{Fig_Sim_M10}, where our hyperprior outperforms alternatives when degrees of freedom are increased. {In detail, for $\nu$ equal to $10$, the loss-based prior shows strong differences when the horizons increase. For the medium and large cases, the loss-based prior outperforms the other two priors with fixed $\nu$ of around $5-7\%$ at horizon $3$ and $25-30\%$ at horizon 7. 
} %We see improvements of around $5\%$ at horizon $3$, while the improvements are around $25\%$ at horizon $7$ for a higher choice of $\nu$.}

\begin{table}[h!]
    \centering
    \begin{adjustbox}{width=1\textwidth,center=\textwidth}
    \begin{tabular}{c|c c c| c c c| c c c}
    \hline 
    & \multicolumn{3}{c}{$\nu = 10$} & \multicolumn{3}{c}{$\nu = 15$} & \multicolumn{3}{c}{$\nu = 20$} \\
        Horizon & Fixed & Fixed & Hyperprior & Fixed & Fixed & Hyperprior & Fixed & Fixed & Hyperprior  \\
         & $\nu =0$ & $\nu = m+1$ & & $\nu =0$ & $\nu = m+1$ & & $\nu =0$ & $\nu = m+1$ \\
         \hline
    1 & 0.6490 &   0.9960  &  0.9960 &  0.3820 &   0.9990     & 0.9990 &  0.3280 &   1.0000  &  1.0000 \\
    3 & 0.6360  &  1.0280  &  1.0290 &     0.4630  &  0.9720     & 0.9510 &     0.4740  &  0.9700  &  0.9300 \\
    5 & 1.2470  &  0.9180  &  0.9110 &     1.5100  &  0.9030     & 0.8230 &     1.8410  &  0.9140  &  0.7980 \\
    7 & 3.9530  &  0.8630  &  0.8530 &     5.6870  &  0.8550     & 0.7410 &     7.8950  &  0.8630  &  0.6990 \\
         \hline 
    \end{tabular}
    \end{adjustbox}
    \caption{{Monte Carlo Simulation - RMAD of the IRF for four horizons $h=1, 3, 5$ and $7$ by simulating $250$ VAR(1) with dimension $m=10$ and sample size $T=30$. Column Fixed $\nu=0$ provides the RMAD of the IRF, while Columns Fixed $\nu = m+1$ and Hyperprior provide the ratio between the referred priors and the flat prior.}}
    %Monte Carlo Simulation - Root Mean Absolute Deviation of the impulse response functions (IRF)  of dimension $m=10$ for four different horizons $h= 1, 3, 5, 7$ by simulating $250$ VAR(1) of sample size $T=30$. Column Fixed $\nu=0$ provides the RMAD of the IRF, while Columns Fixed $\nu = m+1$ and Hyperprior provide the ratio between the referred priors and the flat prior.}
    \label{tab:IRF_Sim_M10}
\end{table}

{For the twenty-dimensional case,  Figure~\ref{Fig_Sim_M20} reports the results for data generated from a Wishart with $20$ (left panel), $24$ (center), and $26$ (right) degrees of freedom with T equal to $30$. In this case, we observe that our loss-based prior and the fixed $\nu = m+1$ prior behave similarly for the left and center panels, when both outperform the flat prior. When data are generated from a Wishart distribution with $26$ degrees of freedom (right panel), our loss-based prior outperforms the fixed $\nu = m+1$ prior. Table~\ref{tab:IRF_Sim_M20} provides the results for the RMAD for the IRFs for the dimensionality $m=20$ and it confirms the findings from Figure~\ref{Fig_Sim_M20}. Hence, the improvements are strong when degrees of freedom are increased and when horizons are increased.}
%In conclusion, Figure~\ref{Fig_Sim_M20} shows the results for the twenty-dimensional case. Also, in this case, we have reported three different figures for data generated from a Wishart with $20$ degrees of freedom (left panel), $24$ (center), and $26$ (right) with T equal to $30$. In this case, the improvements are less evident than in the previous case for the left and center panels. However, in the case of $24$ degrees of freedom, our loss-based hyperprior outperforms the fixed $\nu$ prior. This result is strong when we use data generated from a Wishart with $26$ degrees of freedom as stated in the right panel of Figure~\ref{Fig_Sim_M20}.

\begin{figure}[h!]
 %	{\includegraphics[width=0.33\linewidth]{Figures_Revision/RootMAD,T=30,M=20,dof=20_HS_Rev}}\hfill%
%	{\includegraphics[width=0.33\linewidth]{Figures_Revision/RootMAD,T=30,M=20,dof=24_HS_Rev}}\hfill%
%	{\includegraphics[width=0.33\linewidth]{Figures_Revision/RootMAD,T=30,M=20,dof=26_HS_Rev}} \\
 	{\includegraphics[width=0.33\linewidth]{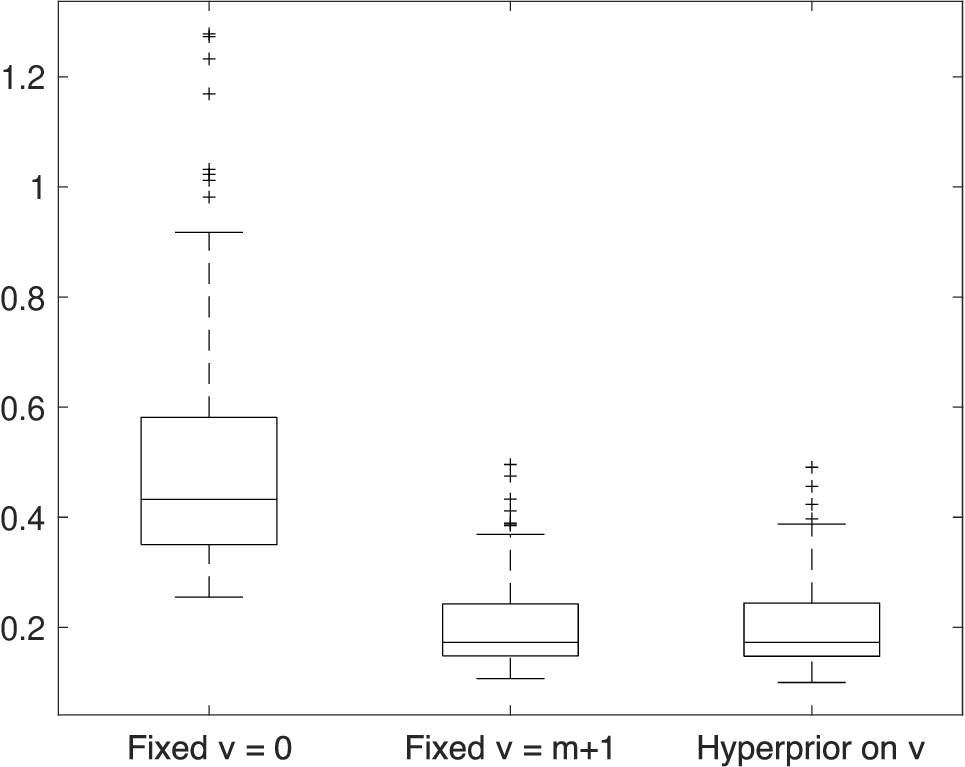}}\hfill%
	{\includegraphics[width=0.33\linewidth]{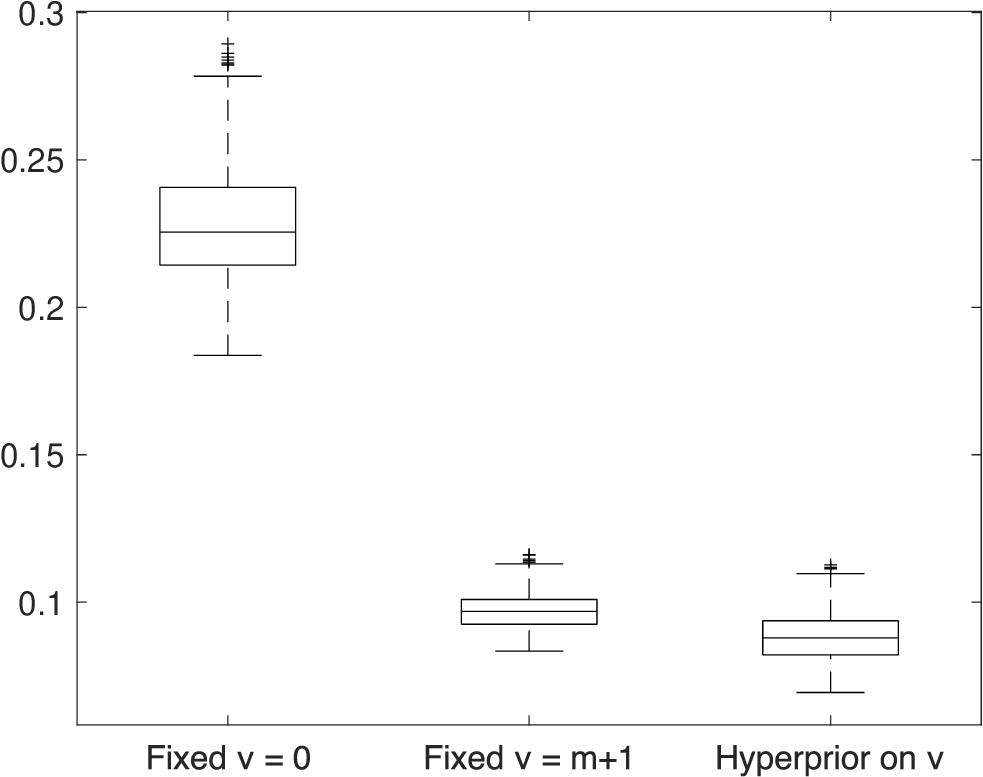}}\hfill%
	{\includegraphics[width=0.33\linewidth]{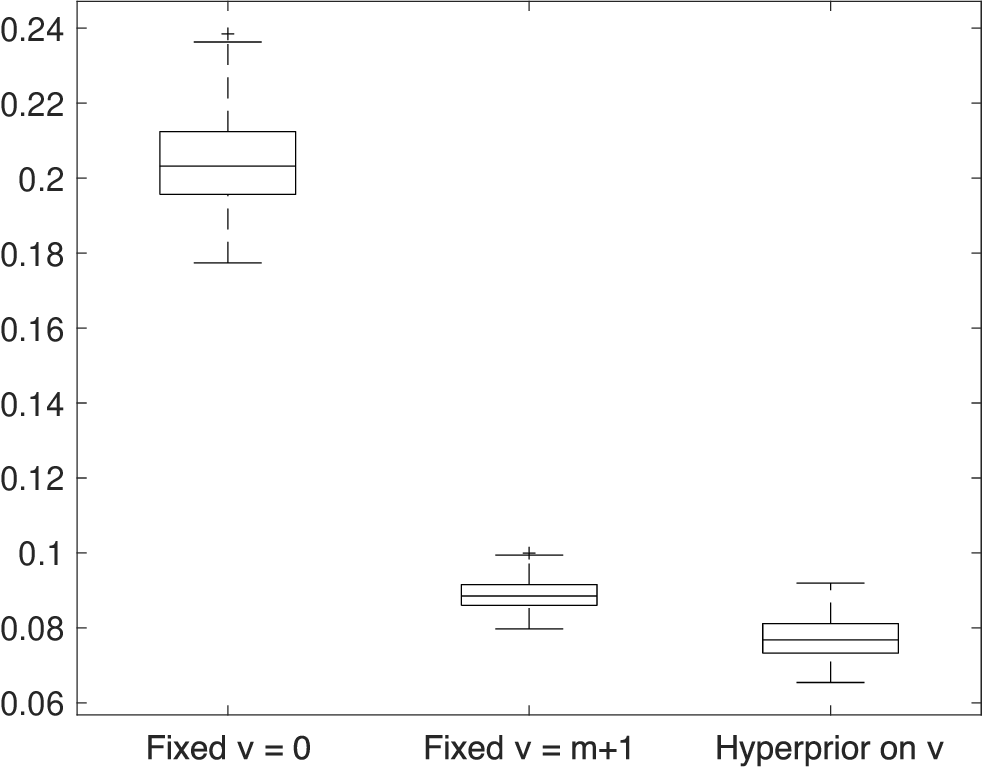}}
	\caption{{Monte Carlo Simulation - RMAD of the covariance matrices. These distributions are obtained by simulating $250$ VAR(1) with dimension $m=20$ and sample size $T=30$. Results are reported for data generated from a Wishart distribution with $\nu = 20$ (left), $\nu=24$ (center), and $\nu = 26$ (right).}}
 %Monte Carlo Simulation - Root Mean Absolute Deviations of the covariance matrices of dimension $m=20$. These empirical distributions are obtained by simulating $250$ VAR(1) of sample size $T=30$. Results are reported separately for data generated from a Wishart with $\nu= 20$ (left),  $\nu=24$ (center), and $\nu = 26$ (right).}
	\label{Fig_Sim_M20}
\end{figure}

%{Table~\ref{tab:IRF_Sim_M20} provides the results for the RMAD for the IRFs for the dimensionality $m=20$. As stated for the other cases, Table~\ref{tab:IRF_Sim_M20} confirms the findings from Figure~\ref{Fig_Sim_M20}. Hence, the improvements are strong when increasing the degrees of freedom and when increasing the horizons.}

\begin{table}[h!]
    \centering
    \begin{adjustbox}{width=1\textwidth,center=\textwidth}
    \begin{tabular}{c|c c c| c c c| c c c}
    \hline 
    & \multicolumn{3}{c}{$\nu = 20$} & \multicolumn{3}{c}{$\nu = 24$} & \multicolumn{3}{c}{$\nu = 26$} \\
        Horizon & Fixed & Fixed & Hyperprior & Fixed & Fixed & Hyperprior & Fixed & Fixed & Hyperprior  \\
         & $\nu =0$ & $\nu = m+1$ & & $\nu =0$ & $\nu = m+1$ & & $\nu =0$ & $\nu = m+1$ \\
         \hline
    1 & 0.6640   &   1.0030   &   1.0030 &   0.4140  &  1.0020    & 1.0020 &  0.3820  &  1.0010  &  1.0010 \\
    3 & 0.6430   &   1.0060   &   1.0070 &     0.5090  &  0.9040    & 0.8950 &     0.5110  &  0.8900  &  0.8780 \\
    5 & 2.4470   &   0.8040   &   0.8050 &     2.6860   & 0.7930    & 0.7660 &     2.8390  &  0.7930  &  0.7590 \\
   7 & 12.9200   &   0.7340   &   0.7350 &    15.5480  &  0.7180    & 0.6840 &    17.1270  &  0.7190  &  0.6710 \\
         \hline 
    \end{tabular}
    \end{adjustbox}
    \caption{{Monte Carlo Simulation - RMAD of the IRF for four horizons $h=1, 3, 5$ and $7$ by simulating $250$ VAR(1) with dimension $m=20$ and sample size $T=30$. Column Fixed $\nu=0$ provides the RMAD of the IRF, while Columns Fixed $\nu = m+1$ and Hyperprior provide the ratio between the referred priors and the flat prior.}}
%    Monte Carlo Simulation - Root Mean Absolute Deviation of the impulse response functions (IRF)  of dimension $m=20$ for four different horizons $h=1, 3, 5, 7$ by simulating $250$ VAR(1) of sample size $T=30$. Column Fixed $\nu=0$ provides the RMAD of the IRF, while Columns Fixed $\nu = m+1$ and Hyperprior provide the ratio between the referred priors and the flat prior.}
    \label{tab:IRF_Sim_M20}
\end{table}

These results are confirmed when T is equal to $100$ and $240$ as shown in Appendix~\ref{App_A}.

\section{Case study 1: Forecasting macroeconomic data}\label{sec_real}

In Section \ref{sec_Data_Macro} we first summarise the macroeconomic dataset \citep[see,][]{McCracken2016,McCracken2020}, while Section \ref{sec_macro} presents the forecasting measures and the main findings of the forecasting exercise by comparing the results of our loss-based hyperprior with respect to the fixed $\nu$\ prior.

\subsection{Data description} \label{sec_Data_Macro}
%Descrizione delle misure di forecasting 
{The FRED data are a set of key US macroeconomic quantities sampled at a quarterly frequency from the second quarter of 1959 to the third quarter of 2019. All variables are transformed to be stationary by following the approach of \cite{McCracken2020} and we use three different sets of variables to estimate a small, medium, and large-scale VAR model.}
%This empirical application aims to forecast a set of key US macroeconomic quantities from the FRED dataset \citep{McCracken2020}, which is a variant of the well-known \cite{McCracken2016} dataset for the US. The data are on a quarterly frequency and the time period spans from the second quarter of 1959 to the third quarter of 2019. Moreover, all the variables are transformed to be stationary by following \cite{McCracken2020} and we use three different sets of variables to estimate a small-scale, medium-scale, and large-scale VAR model.

The small-scale VAR model considers three variables, which represent inflation, the Gross Domestic Product (GDP), and the target interest rate. In particular, the GDP is measured in Billions of Dollars, while inflation is measured by the GDP deflator, which computes the changes in prices for all goods and services produced in the economy and differs from the Consumer Price Index (CPI) because it is not based on a fixed basket of goods. The final variable is the effective Federal Funds Rate {(FEDFUNDS)}, which is the target interest rate set by the Federal Open Market Committee at which commercial banks borrow and lend their excess reserves to each other overnight.

For the medium-scale model, we additionally consider consumption, investment, and production variables. The real personal consumption expenditure (PCECC96) is a measure of consumer or household spending for a period and it is used to construct the PCE Price index, which measures the price changes in consumer goods and services in the US economy. Real gross private domestic investment (GPDIC1) is a component of the GDP and measures the quantity of money invested by private businesses in the domestic economy. The average weekly hours of production and nonsupervisory employees for the manufacturing sector (AWHMAN) relate to the average hours per worker for which pay was received, and it differs from the standard and scheduled hours. %Analogously, one can measure the real Average Hourly Earnings of Production and Nonsupervisory Employees from the construction sector.

The large-scale model additionally uses the macroeconomic variables related to GDP, inflation, production, consumption, and investment jointly with private investment in the residential sector (PRFIx), consumption (PCECTPI), and common stock index based on the S\&P 500 index (SP500). The Industrial Production Index (INDPRO) measures the level of production and capacity in the manufacturing, mining, electric, and gas industries relative to 2012. Capacity Utilization (CUMFNS) captures the manufacturing and production capabilities that are being used by the economy at any given time, relating the output produced with the given resources and the potential output that can be produced if capacity is fully used. %All Employees: Service-Providing Industries (SRVPRD)
%Gross Private Domestic Investment (GPDICTPI)
Lastly, CPI for All Urban Consumers (CPIAUCSL) measures the average change over time in the prices paid by consumers for a market basket of consumer goods and services.

%These important US macroeconomic variables are used in the multivariate time series models to forecast one quarter ahead ($h=1$).
 
\subsection{Forecasting  results}\label{sec_macro}

In this section, we evaluate the performance of our loss-based hyperprior with respect to the fixed $\nu$ prior by forecasting one-quarter ahead ($h=1$). %compare the results of our loss-based hyperprior by using two different real data applications. The first real data application deals with macroeconomic variables from the FRED dataset and we consider three different sizes: a small (3 variables), a medium (7 variables) and a large (15 variables) datasets. On the other hand,  the second data refers to the Dengue Fever data on $10$ different countries across the world. 
%In both applications, we evaluate the performance of our loss-based hyperprior by forecasting one step ahead values.
We compare the predictive ability of the three different priors by using point and density forecasting measures. To evaluate the forecasting capability, we compute the root mean square error (RMSE), given by
\begin{equation} \label{rmse}
    RMSE_i =  \bigg[ \frac{1}{T-R}\sum_{t=R}^{T-1}(\hat{y}_{i,t+1}-y_{i,t+1})^2 \bigg]^\frac{1}{2},
\end{equation}
where $R$ is the length of the rolling window, $y_{i,t+1}$ is the observation for the $i$-th variable, and $\hat{y}_{i,t+1}$ is the one-step ahead prediction for the $i$-th variable.

In addition, we evaluate the density forecasting using the continuous ranked probability score (CRPS) introduced by \cite{Gneiting2007} and \cite{Gneiting2011}. The use of the CRPS has some advantages with respect to the log score since it weights values from the predictive density that are close to the outcome and it is less sensitive to outlier outcomes. The CRPS is defined such that a lower value indicates better performance, and is given by
\begin{equation} \label{crps}
\begin{split}
    CRPS_{t}(y_{t+1}) &= \int_{-\infty}^{+\infty} (F(z)-\mathbbm{1}(y_{t+1}\leq z))^2dz\\
    & = E_{f}|Y_{t+1}-y_{t+1}| - 0.5E_{f}|Y_{t+1}-Y'_{t+1}|,
\end{split}
\end{equation}
where $F(\cdot)$ is the cumulative distribution function associated with the posterior predictive density, $f$, $\mathbbm{1}(y_{t+1}\leq z)$ is an indicator function taking the value 1 if $y_{t+1}\leq z$ and 0 otherwise, and $Y_{t+1}$, $Y'_{t+1}$ are independent random draws from the posterior predictive density. 
%{\color{blue} In the tables, we provide the RMSE and the average CRPS for the flat prior, while for fixed $\nu = m+1$ and the proposed loss-based prior, we provide the ratios with respect to the flat prior. A value lower than $1$ means that the flat prior is outperformed by the other prior representations.}
%For the point and the density forecasting metrics, in the table we show their values and a lower value of the RMSE or of the average CRPS means that the model is forecasting better with respect to the others.

In addition, we apply Diebold-Mariano $t$ tests \citep{Diebold1995} for equality of the average loss (with loss defined as the RMSE or CRPS) to compare the predictions of alternative models with the benchmark. The differences in accuracy that are statistically different from zero are denoted with one, two, or three asterisks, corresponding to significance levels of $10\%$, $5\%$, and $1\%$, respectively. %We report $p$ values based on one-sided tests, taking the VAR model with fixed $\nu$ as the benchmark and our proposed approach as the alternative.

%We now turn focus on real data and compare the predictive ability of the two models via a rolling window analysis. We assess the predictive performance using the Root Mean Square Error and the Continuous Ranked Probability Score on one-step ahead forecasts. The RMSE is computed as :
%\begin{equation} \label{rmse}
%    RMSE =  \bigg[ \frac{1}{T-R}\sum_{t=R}^{T-1}(\hat{y}_{t+1}-y_{t+1})^2 \bigg]^\frac{1}{2}
%\end{equation}
%where $R$ is the length of the rolling window and $\hat{y}_{t+1}$ are the one-step ahead predictions. CRPS evaluates density forecasts, based on the posterior predictive density, with the following formula :
%
%\begin{equation} \label{crps}
%\begin{split}
%    CRPS(F,y_{t+1}) &= \int_{-\infty}^{+\infty} (F(z)-\mathbbm{1}(y_{t+1}\leq z))^2dz\\
%    & = E_{F}|Y_{t+1}-y_{t+1}| - 0.5E_{F}|Y_{t+1}-Y'_{t+1}|
%\end{split}
%\end{equation}
%where $F$ is the cumulative distribution function associated with the posterior predictive density, $\mathbbm{1}(y_{t+1}\leq z)$ is an indicator function taking the value 1 if $y_{t+1}\leq z$ and 0 otherwise, and $Y_{t+1}$, $Y'_{t+1}$ are independent random draws from the posterior predictive density. For both metrics a lower score is a better score. 

{As stated in Section~\ref{sec_Data_Macro}, we use three sets of variables for small, medium, and large-scale VAR models. The small-scale VAR only includes three variables, the medium-scale VAR seven variables, and the large-scale VAR $15$ variables. Given the quarterly frequency of our data, we include $p=5$ lags for all of the models considered, and we fit the models using the MCMC algorithm with $6000$ iterations after discarding the first $1000$ iterations as burn-in. For forecasting, we use a rolling window size of $60$ quarters and run one-step ahead forecasts.}

Before evaluating forecasting performance, we briefly show the results of the inferred degrees of freedom using a rolling window estimation approach. We show the results for the three different datasets of the posterior mean of the degrees of freedom estimated by using our loss-based hyperprior. Figure~\ref{Fig_Macro_nu} shows the results of the estimated degrees of freedom jointly with the $95\%$ highest posterior density (HPD) and the degrees of freedom used in the {Horseshoe}-Wishart scenario with fixed $\nu = m+1$ (in red). %The left panel shows the results for the small-scale VAR, the center panel for the medium-scale VAR, and the right for the large-scale VAR.

From Figure~\ref{Fig_Macro_nu}, we observe strong evidence of changes in the estimated degrees of freedom across time. The left panel results indicate an increase in the degrees of freedom after 2000 and a fall around 2009, strongly linked to the Lehman Brothers failure. {These results are confirmed also in the medium-scale VAR model (center panel) with an increase also observed in the last period. These changes are less evident in the large-scale VAR model (right panel) but provide clear support for not fixing  $\nu$. These findings of changes in the degrees of freedom across time can be linked to the uncertainty literature \cite[see,][]{Bloom2014}, which shows countercyclical fluctuations during recessionary periods.}

\begin{figure}[h!]
	{\includegraphics[width=0.33\linewidth]{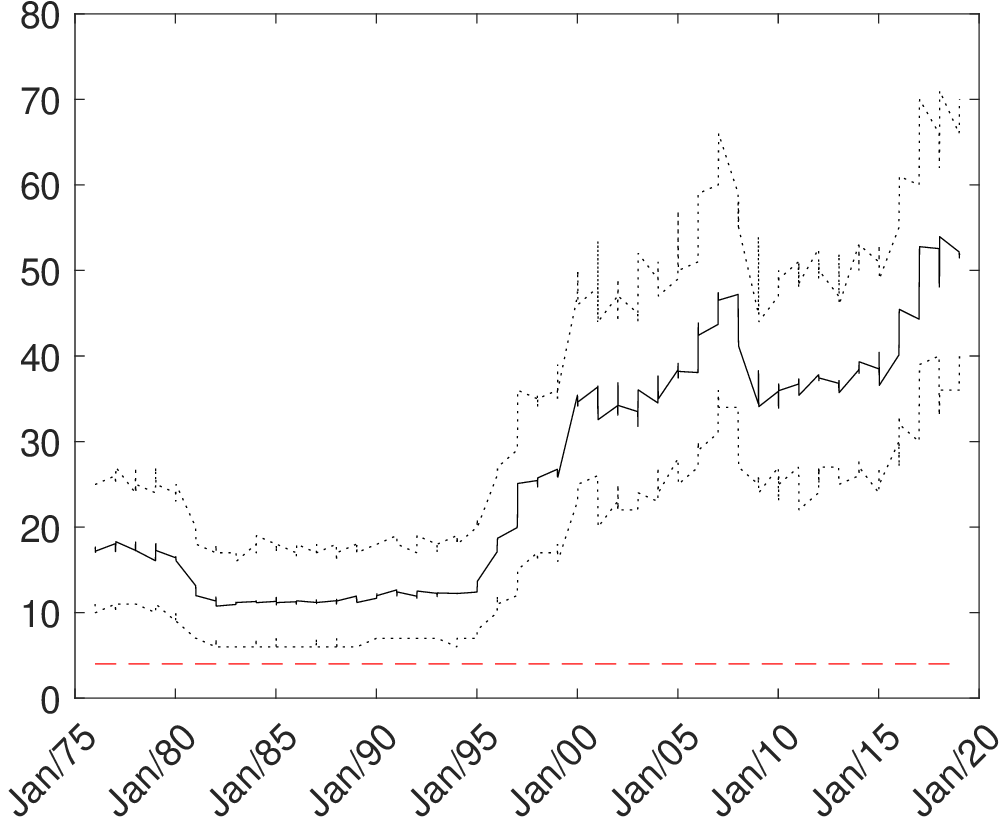}}\hfill%
	{\includegraphics[width=0.33\linewidth]{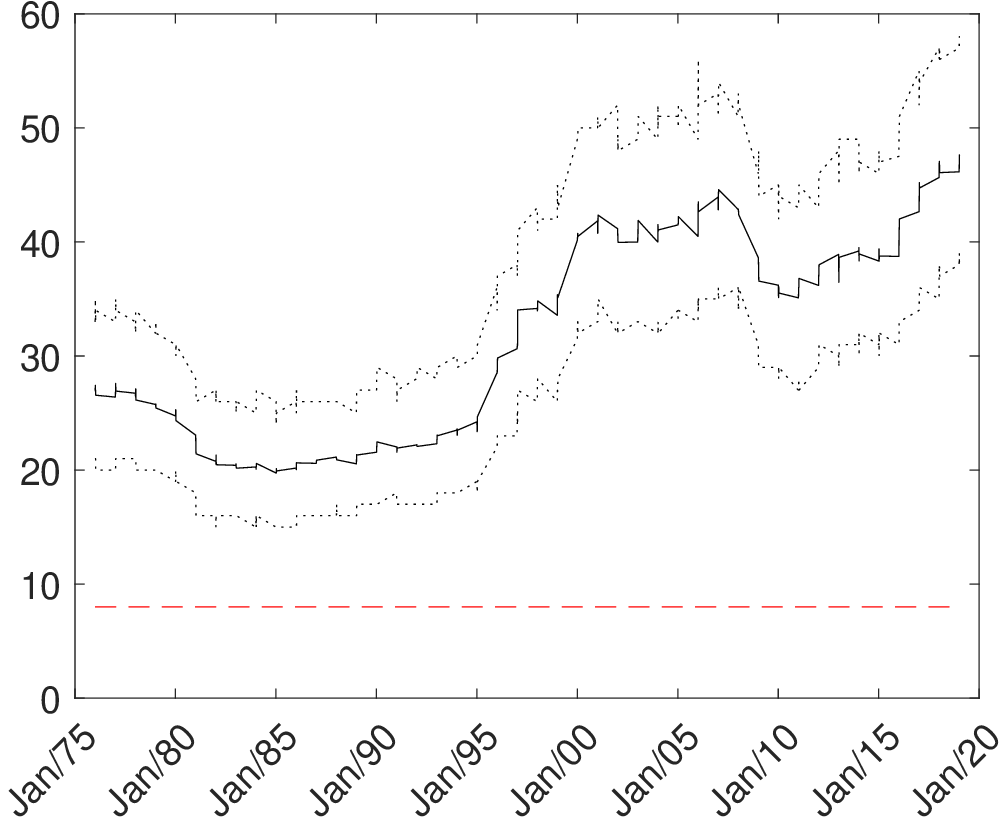}}\hfill%
	{\includegraphics[width=0.33\linewidth]{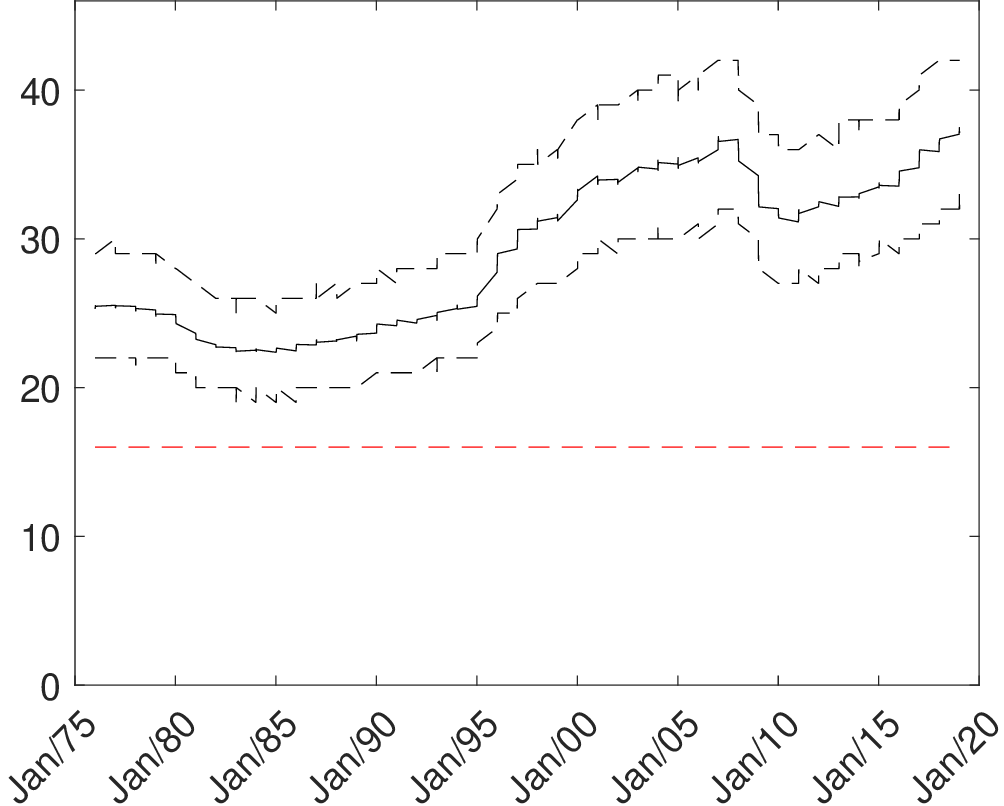}}
	\caption{Estimated degrees of freedom (solid line) for the loss-based hyperprior by using a rolling window of $60$ quarters with the $95\%$ Highest posterior density (dotted lines) and the case with {$\nu = m +1$, where $m \in \{3, 7, 15\}$} (red dashed line) for the macroeconomic data. The left panel is for the small-case; the center for the medium-scale and the right for the large-scale.}
	\label{Fig_Macro_nu}
\end{figure}

Table~\ref{Table_MacroSmall} shows the results for the small-case VAR for the RMSE (left panel) and the CRPS (right panel). {The column ``Fixed $\nu=0$'', which is the benchmark, reports the RMSE and the average CRPS, while the columns ``Fixed $\nu =m+1$'' and ``Hyperprior'', provide the ratio between each prior and the flat prior. When the ratio is less than $1$, it indicates that the model with loss-based hyperprior or with fixed $\nu = m+1$ outperforms the benchmark model with fixed $\nu = 0$.} In the point forecasting measure, the benchmark model is outperforming our hyperprior for the FEDFUNDS by a small amount. On the other hand, for the GDP, our hyperprior leads to an improvement of around $9\%$ with respect to the benchmark. This result is confirmed for the average CRPS, where our loss-based hyperprior outperforms the benchmark model by $15\%$ for real GDP growth and the GDP deflator, and by $3\%$ for FEDFUNDS. By using the Diebold-Mariano test, we show evidence of statistical significance for all variables when density forecasting is considered, while for point forecasting, the results are statistically significant for GDP and GDP deflator. {In addition, our loss-based prior shows improved density forecasting performance relative to the alternative fixed $\nu =m+1$ prior.}

%\begin{table}[h!]
%	\centering
%\begin{tabular}{l|lll|lll}
%\hline
%&\multicolumn{3}{c|}{RMSE} &
%\multicolumn{3}{c}{CRPS} \\ \hline
%& fixed $\nu$ & $\nu \sim \pi(\nu)$ & ratio & fixed $\nu$ & $\nu \sim \pi(\nu)$ & ratio \\ 
%\hline 
%GDPC1 & 0.010 & 0.010 & 1.006 & 0.045 & 0.039$^{\ast \ast \ast}$ & 0.866 \\ 
%GDPCTPI & 0.004 & 0.003 & 0.811 & 0.045 & 0.039$^{\ast \ast \ast}$ & 0.868 \\ 
%FEDFUNDS & 0.946 & 0.953 & 1.006 & 0.388 & 0.381$^{\ast \ast \ast}$ & 0.981 \\ 
%\hline 
%\end{tabular}
%\caption{RMSE and average CRPS for the small-case VAR for each prior. The first column refers to the Gibbs sampler with fixed $\nu$, the second to our loss-based hyperprior, and the third to the ratio between the two priors. $^{\ast \ast \ast}$, $^{\ast \ast}$ and $^\ast$ indicate ratios are significant different from $1$ at the $1\%$, $5\%$ and $10\%$ significance level according to the Diebold-Mariano test.}
%\label{Table_MacroSmall}
%\end{table}

\begin{table}[h!]
    \centering
    \begin{adjustbox}{width=1\textwidth,center=\textwidth}
    \begin{tabular}{c|c c c| c c c}
    \hline 
    & \multicolumn{3}{c}{RMSE} & \multicolumn{3}{|c}{average CRPS} \\
        Variable & Fixed & Fixed & Hyperprior & Fixed & Fixed & Hyperprior \\
         & $\nu =0$ & $\nu = m+1$ & & $\nu =0$ & $\nu = m+1$  \\
         \hline
        GDPC1 &  0.0092 &   0.9459 &   0.9116$^{\ast \ast \ast}$ & 0.0513  &  0.9653$^{\ast \ast \ast}$ &   0.8483$^{\ast \ast \ast}$ \\
    GDPCTPI & 0.0044  &  0.9484   & 0.9442$^{\ast \ast}$  &   0.0516  &  0.9616$^{\ast \ast \ast}$   & 0.8423$^{\ast \ast \ast}$ \\
    FEDFUNDS & 1.0206 &   0.9960  &  1.0012  &   0.4060  &  0.9880$^{\ast \ast \ast}$  &  0.9713$^{\ast \ast \ast}$ \\
%        GDPC1 & 0.0104 &  0.9902 & 0.9801 &  0.0472 & 0.9602$^{\ast \ast \ast}$ & 0.8373$^{\ast \ast \ast}$ \\
 %       GDPCTPI & 0.0040 & 1.0218 & 0.9558 &   0.0471 & 0.9627$^{\ast \ast \ast}$ & 0.8379$^{\ast \ast \ast}$ \\
  %      FEDFUNDS & 0.9469 & 1.0009 & 1.0056 & 0.3917 & 0.9916$^{\ast \ast \ast}$ & 0.9710$^{\ast \ast \ast}$ \\
         \hline 
    \end{tabular}
    \end{adjustbox}
\caption{RMSE {(Columns 2-4)} and average CRPS {(Columns 5-7)} for the small-case VAR for each prior. Column Fixed $\nu=0$ provides the RMSE and the average CRPS; Columns Fixed $\nu =m+1$ and hyperprior provide the ratio between the referred prior and the flat prior. $^{\ast \ast \ast}$, $^{\ast \ast}$ and $^\ast$ indicate ratios are significantly different from $1$ at the $1\%$, $5\%$ and $10\%$ significance level according to the Diebold-Mariano test.}
    \label{Table_MacroSmall}
\end{table}

For the 7-variable model, Table~\ref{Table_MacroMedium} presents the point and density measures, and similar results can be observed for both forecasting performance measures. {The main difference relates to the GDPCTPI, where the hyperprior approach outperforms the benchmark by about $11\%$ in point forecasting, whilst for the FEDFUNDS, there is little difference.  Moreover, we observe that the hyperprior model is outperforming the other models for the PCECC96 by about $11\%$, while for the GDP both the hyperprior and the model with fixed $\nu = m +1$ (with $m=7$) demonstrate around $7\%$ and $9\%$ improvement.} As in the small-scale VAR, the average CRPS shows better results with respect to the benchmark model across the $7$ variables.
%overturns the situation since the benchmark is never the best model across the $7$ variables. 
In particular, the GDP, GDP deflator, and the FEDFUNDS, the hyperprior model outperforms the benchmark by between $4\%$ and $16\%$. This is also confirmed for the other variables analysed and from the Diebold-Mariano test in particular for the density forecasting measures.

%\begin{table}[h!]
%	\centering
%\begin{tabular}{l|lll|lll}
%\hline
%&\multicolumn{3}{c|}{RMSE} &
%\multicolumn{3}{c}{CRPS} \\ \hline
%& fixed $\nu$ & $\nu \sim \pi(\nu)$ & ratio & fixed $\nu$ & $\nu \sim \pi(\nu)$ & ratio \\ 
%\hline 
%GDPC1 & 0.008 & 0.007$^\ast$ & 0.959 & 0.058 & 0.049$^{\ast \ast \ast}$ & 0.857 \\ 
%GDPCTPI & 0.004 & 0.004 & 0.971 & 0.058 & 0.050$^{\ast \ast \ast}$ & 0.865 \\ 
%FEDFUNDS & 0.941 & 0.947 & 1.006 & 0.391 & 0.381$^{\ast \ast \ast}$ & 0.974 \\ 
%PCECC96 & 0.006 & 0.006 & 1.018 & 0.058 & 0.050$^{\ast \ast \ast}$ & 0.860 \\ 
%GPDIC1 & 0.032 & 0.032 & 1.003 & 0.061 & 0.053$^{\ast \ast \ast}$ & 0.866 \\ 
%AWHMAN & 0.264 & 0.260$^{\ast \ast \ast}$ & 0.984 & 0.150 & 0.144$^{\ast \ast \ast}$ & 0.957 \\ 
%CES2000000008x & 0.007 & 0.007 & 0.909 & 0.059 & 0.051$^{\ast \ast \ast}$ & 0.859 \\ 
%\hline 
%\end{tabular}
%\caption{RMSE and average CRPS for the medium-case VAR for each prior. The first column refers to the Gibbs sampler with fixed $\nu$, the second to our loss-based hyperprior, and the third to the ratio between the two priors. $^{\ast \ast \ast}$, $^{\ast \ast}$ and $^\ast$ indicate ratios are significant different from $1$ at the $1\%$, $5\%$ and $10\%$ significance level according to the Diebold-Mariano test.}
%\label{Table_MacroMedium}
%\end{table}

\begin{table}[h!]
    \centering
    \begin{adjustbox}{width=1\textwidth,center=\textwidth}
    \begin{tabular}{c|c c c| c c c}
    \hline 
        & \multicolumn{3}{c}{RMSE} & \multicolumn{3}{|c}{average CRPS} \\
        Variable & Fixed & Fixed & Hyperprior & Fixed & Fixed & Hyperprior \\
         & $\nu =0$ & $\nu = m+1$ & & $\nu =0$ & $\nu = m+1$  \\
         \hline
     GDPC1 &       0.0095    &  0.9081$^{\ast \ast}$   &   0.9284$^{\ast \ast}$ &  0.0823  &  0.9304$^{\ast \ast \ast}$  &  0.8396$^{\ast \ast \ast}$ \\
   GDPCTPI &   0.0060   &   0.9155$^{\ast}$   &   0.8888$^{\ast \ast}$ &     0.0834  &  0.9252$^{\ast \ast \ast}$ &   0.8368$^{\ast \ast \ast}$ \\
 FEDFUNDS &      1.0048   &   1.0005  &    1.0002 &     0.4165 &   0.9821$^{\ast \ast \ast}$  &  0.9621$^{\ast \ast \ast}$ \\
 PCECC96 &     0.0082   &   0.9796   &   0.8926$^{\ast \ast}$ &     0.0837   & 0.9288$^{\ast \ast \ast}$ &   0.8335$^{\ast \ast \ast}$ \\
 GPDIC1 &     0.0342    &  0.9868    &  1.0029 &     0.0866   & 0.9300$^{\ast \ast \ast}$ &   0.8435$^{\ast \ast \ast}$ \\
AWHMAN &       0.2470    &  1.0003   &   0.9995 &     0.1545  &  0.9641$^{\ast \ast \ast}$  &  0.9308$^{\ast \ast \ast}$ \\
CES2000000008x &     0.0086  &  0.8936 &   0.9486 &     0.0859  &  0.9268$^{\ast \ast \ast}$  &  0.8336$^{\ast \ast \ast}$  \\
%GDPC1 & 0.0080 & 1.0374  & 0.9704 & 0.0620 & 0.9375$^{\ast \ast \ast}$ & 0.8028$^{\ast \ast \ast}$ \\
%GDPCTPI &  0.0044  &  1.0506 & 0.9839 & 0.0626 &  0.9323$^{\ast \ast \ast}$ & 0.7980$^{\ast \ast \ast}$ \\ 
%FEDFUNDS &   0.9371 & 1.0095 & 1.0116 & 0.3986  &  0.9831$^{\ast \ast \ast}$  & 0.9574$^{\ast \ast \ast}$\\ 
%PCECC96 &  0.0073 & 0.9787  &  0.9538 & 0.0630  &  0.9285$^{\ast \ast \ast}$  & 0.7964$^{\ast \ast \ast}$\\ 
%GPDIC1 &  0.0328 & 0.9916 & 0.9873 & 0.0655  &  0.9322$^{\ast \ast \ast}$  &  0.8058$^{\ast \ast \ast}$\\ 
%AWHMAN &   0.2679 & 0.9918$^{\ast}$ & 0.9757$^{\ast \ast \ast}$ & 0.1549  &  0.9732$^{\ast \ast \ast}$  & 0.9321$^{\ast \ast \ast}$ \\ 
%CES2000000008x & 0.0073 & 0.9884 & 0.8685$^{\ast \ast \ast}$ &  0.0640   &  0.9287$^{\ast \ast \ast}$  & 0.7986$^{\ast \ast \ast}$\\ 
         \hline 
    \end{tabular}
    \end{adjustbox}
\caption{RMSE {(Columns 2-4)} and average CRPS {(Columns 5-7)} for the medium-case VAR for each prior. Column Fixed $\nu=0$ provides the RMSE and the average CRPS; Columns Fixed $\nu =m+1$ and hyperprior provide the ratio between the referred prior and the flat prior. $^{\ast \ast \ast}$, $^{\ast \ast}$ and $^\ast$ indicate ratios are significantly different from $1$ at the $1\%$, $5\%$ and $10\%$ significance level according to the Diebold-Mariano test.}
    \label{Table_MacroMedium}
\end{table}

The results for the large-scale VAR with $15$ variables are presented in Table~\ref{Table_MacroLarge}. For point forecasting, the three main variables of interest show similar improvements to the one in the medium-scale VAR. In fact, the loss-based hyperprior improves with respect to the benchmark by about $12\%$ for the GDP and $28\%$ for its deflator, while for the {FEDFUNDS} the situation is similar. The average CRPS demonstrates that the improvement is stronger for every variable, particularly for the FEDFUNDS. If we look at all $15$ variables analysed, the loss-based hyperprior model always outperforms the benchmark in a density forecasting scenario, which is also highlighted by the Diebold-Mariano test.
{These results are confirmed when looking at the prior with fixed $\nu$ equal to $m+1$, where our loss-based prior outperforms the other prior in density forecasting.}

\begin{table}[h!]
    \centering
    \begin{adjustbox}{width=1\textwidth,center=\textwidth}
    \begin{tabular}{c|c c c| c c c}
    \hline 
        & \multicolumn{3}{c}{RMSE} & \multicolumn{3}{|c}{average CRPS} \\
        Variable & Fixed & Fixed & Hyperprior & Fixed & Fixed & Hyperprior \\
         & $\nu =0$ & $\nu = m+1$ & & $\nu =0$ & $\nu = m+1$  \\
         \hline
GDPC1 &  0.0093  & 0.9105  & 0.8833$^{\ast \ast}$  &  0.0905 &   0.8542$^{\ast \ast \ast}$  &  0.7783$^{\ast \ast \ast}$  \\
GDPCTPI & 0.0074  & 0.8672 & 0.7153$^{\ast \ast \ast}$ &    0.0913 &   0.8520$^{\ast \ast \ast}$  &  0.7747$^{\ast \ast \ast}$ \\
FEDFUNDS &  0.9441  & 1.0057 & 1.0065&    0.4138  &  0.9613$^{\ast \ast \ast}$  &  0.9454$^{\ast \ast \ast}$ \\
PCECC96 & 0.0085  & 0.9471 & 0.9442$^{\ast \ast}$ &    0.0920  &  0.8502$^{\ast \ast \ast}$  &  0.7781$^{\ast \ast \ast}$ \\ 
GPDIC1 & 0.0324  & 1.0003 & 1.0185 &    0.0945  &  0.8529$^{\ast \ast \ast}$  &  0.7827$^{\ast \ast \ast}$ \\  
AWHMAN & 0.2484 & 0.9978 & 0.9892$^{\ast}$ &   0.1616  &  0.9286$^{\ast \ast \ast}$  &  0.8957$^{\ast \ast \ast}$ \\
CES2000000008x &  0.0091 & 0.8400$^{\ast \ast \ast}$ & 0.8939$^{\ast \ast}$ &    0.0939  &  0.8493$^{\ast \ast \ast}$  &  0.7753$^{\ast \ast \ast}$ \\
PRFIx &  0.0407  & 0.9914 & 0.9930 &    0.0970  &  0.8520$^{\ast \ast \ast}$  &  0.7813$^{\ast \ast \ast}$ \\
INDPRO &  0.0146  & 0.8374$^{\ast \ast \ast}$ & 0.9306$^{\ast \ast \ast}$ &    0.0953  &  0.8501$^{\ast \ast \ast}$  &  0.7770$^{\ast \ast \ast}$ \\
CUMFNS &  1.0338  & 0.9915 & 0.9844 &    0.5651  &  0.9586$^{\ast \ast \ast}$  &  0.9435$^{\ast \ast \ast}$ \\ 
SRVPRD &  0.0081 & 0.9427 & 0.9696 &    0.0965  &  0.8458$^{\ast \ast \ast}$  &  0.7727$^{\ast \ast \ast}$ \\
PCECTPI & 0.0104  & 0.7792$^{\ast \ast \ast}$ & 0.6835$^{\ast \ast \ast}$ &    0.0972  &  0.8484$^{\ast \ast \ast}$  &  0.7703$^{\ast \ast \ast}$ \\ 
GPDICTPI & 0.0094  &  0.9068$^{\ast \ast}$ & 0.8185$^{\ast \ast \ast}$ &    0.0980  &  0.8435$^{\ast \ast \ast}$  &  0.7703$^{\ast \ast \ast}$ \\
CPIAUCSL & 0.0111 & 0.9165 & 0.8336$^{\ast \ast}$ &    0.0992  &  0.8452$^{\ast \ast \ast}$  &  0.7668$^{\ast \ast \ast}$ \\
SP500 & 0.0679 & 0.9913$^{\ast \ast}$ & 1.0188 &    0.1065  &  0.8517$^{\ast \ast \ast}$  &  0.7830$^{\ast \ast \ast}$ \\  
         \hline 
    \end{tabular}
    \end{adjustbox}
\caption{RMSE {(Columns 2-4)} and average CRPS {(Columns 5-7)} for the large-case VAR for each prior. Column Fixed $\nu=0$ provides the RMSE and the average CRPS; Columns Fixed $\nu =m+1$ and hyperprior provide the ratio between the referred prior and the flat prior. $^{\ast \ast \ast}$, $^{\ast \ast}$ and $^\ast$ indicate ratios are significantly different from $1$ at the $1\%$, $5\%$ and $10\%$ significance level according to the Diebold-Mariano test.}
    \label{Table_MacroLarge}
\end{table}

\section{Case study 2: Dengue data}\label{sec_dengue}

{The second case study analyses the Google Dengue Trend (GDT) dataset, which tracks the Dengue incidence based on internet search patterns and clusters weekly queries for key terms related to the disease. We use GDT data from January 2011 to December 2014 for Argentina, Bolivia, Brazil, India, Indonesia, Mexico, Philippines, Singapore, Thailand, and Venezuela, thus having $10$ response variables. Following \cite{Davis2016}, we examine a VAR model with two lags, and we run a forecasting exercise with a rolling window of $104$ weekly observations.}

{In Figure~\ref{Fig_Dengue_nu} we present the posterior mean of the degrees of freedom from the model using our hyperprior (solid black line), the $95\%$ HPD (dotted line), and the fixed values of $\nu$ equal to $11$. This has been evaluated using rolling window estimation. The results indicate that the estimated degrees of freedom are often considerably larger than the fixed value of $11$. We observe some changes in values at the beginning of the sample and then a decrease before it remains relatively stationary.}

 \begin{figure}[h!]
 \centering
	{\includegraphics[width=0.33\linewidth]{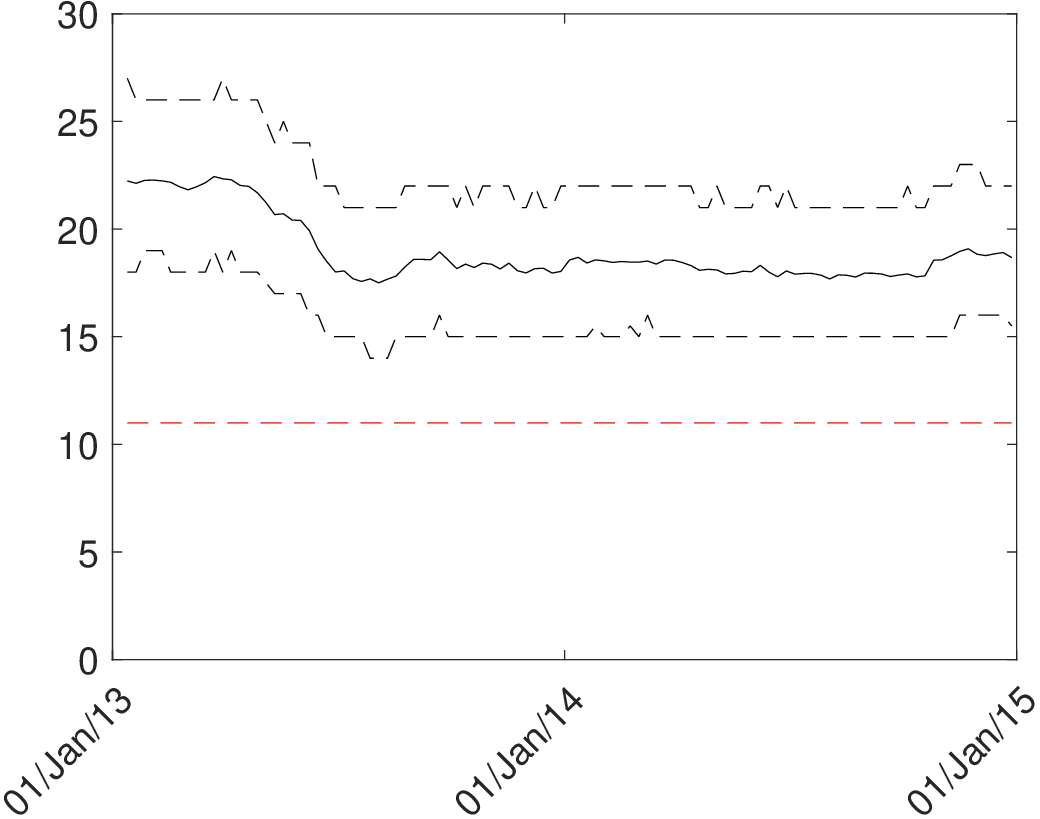}}\hfill%
%	\bigskip
%	{\includegraphics[width=0.33\linewidth]{Figures/RootMAD,T=100,M=20,dof=20}}\hfill%
%	{\includegraphics[width=0.33\linewidth]{Figures/RootMAD,T=100,M=20,dof=24}}\hfill%
%	{\includegraphics[width=0.33\linewidth]{Figures/RootMAD,T=100,M=20,dof=26}}
	\caption{Estimated degrees of freedom (solid line) for the loss-based hyperprior by using a rolling window of $104$ weekly data with the $95\%$ Highest posterior density (dotted lines) and the case with {$\nu = m +1$, where $m = 10$} (red dashed line) for the Dengue data.}
	\label{Fig_Dengue_nu}
\end{figure}

{Table~\ref{Table_Dengue} shows the forecasting results for each country. For the RMSE, the proposed loss-based hyperprior leads to small improvements with respect to the benchmark prior and the other fixed prior for Bolivia, India, and Mexico, while for Brazil and Philippines, the model with fixed prior equal to $\nu = m +1$ is performing better with respect to the benchmark and slightly better than the proposed loss-based prior. When we look at the average CRPS for every country, we notice strong improvements against the benchmark model and the fixed $\nu$ equal to $m+1$ prior in all countries. In particular, the proposed loss-based prior outperforms the benchmark by $5\%$ for India, Philippines, and Brazil, and by $3\%$ for Argentina, Indonesia, Thailand, and Venezuela.} These results provide evidence of strong significance from density forecasting for all countries as shown by the Diebold-Mariano test, while from a point forecast measure, the statistical significance of the proposed prior is shown for a few countries (such as Bolivia, India, and Brazil).

%\clearpage
%
%\begin{table}[h!]
%	\centering
%\begin{tabular}{l|lll|lll}
%\hline
%&\multicolumn{3}{c|}{RMSE} &
%\multicolumn{3}{c}{CRPS} \\ \hline
%& fixed $\nu$ & $\nu \sim \pi(\nu)$ & ratio & fixed $\nu$ & $\nu \sim \pi(\nu)$ & ratio \\  
%\hline 
%Argentina & 0.328 & 0.329 & 1.002 & 0.174 & 0.172$^{\ast \ast}$ & 0.992 \\ 
%Bolivia & 0.228 & 0.227$^{\ast \ast}$ & 0.993 & 0.145 & 0.144$^{\ast \ast \ast}$ & 0.989 \\ 
%Brazil & 0.279 & 0.277$^{\ast}$ & 0.991 & 0.152 & 0.149$^{\ast \ast \ast}$ & 0.979 \\ 
%India & 0.226 & 0.224$^{\ast \ast}$ & 0.992 & 0.147 & 0.144$^{\ast \ast \ast}$ & 0.977 \\ 
%Indonesia & 0.235 & 0.233 & 0.994 & 0.140 & 0.138$^{\ast \ast \ast}$ & 0.986 \\ 
%Mexico & 0.273 & 0.271 & 0.994 & 0.150 & 0.148$^{\ast \ast \ast}$ & 0.989 \\ 
%Philippines & 0.231 & 0.230 & 0.996 & 0.134 & 0.131$^{\ast \ast \ast}$ & 0.979 \\ 
%Singapore & 0.687 & 0.685 & 0.997 & 0.326 & 0.323$^{\ast \ast \ast}$ & 0.991 \\ 
%Thailand & 0.372 & 0.372$^{\ast}$ & 0.997 & 0.189 & 0.187$^{\ast \ast \ast}$ & 0.987 \\ 
%Venezuela & 0.291 & 0.289 & 0.992 & 0.161 & 0.160$^{\ast \ast}$ & 0.991 \\ 
%\hline 
%\end{tabular}
%\caption{RMSE and average CRPS for the Dengue Data for each prior. The first column refers to the Gibbs sampler with fixed $\nu$, the second to our loss-based hyperprior, and the third to the ratio between the two priors. $^{\ast \ast \ast}$, $^{\ast \ast}$ and $^\ast$ indicate ratios are significant different from $1$ at the $1\%$, $5\%$ and $10\%$ significance level according to the Diebold-Mariano test.}
%\label{Table_Dengue}
%\end{table}

\begin{table}[h!]
    \centering
    \begin{adjustbox}{width=1\textwidth,center=\textwidth}
    \begin{tabular}{c|c c c| c c c}
    \hline 
        & \multicolumn{3}{c}{RMSE} & \multicolumn{3}{|c}{average CRPS} \\
        Variable & Fixed & Fixed & Hyperprior & Fixed & Fixed & Hyperprior \\
         & $\nu =0$ & $\nu = m+1$ & & $\nu =0$ & $\nu = m+1$  \\
         \hline
   Argentina &  0.3475   &  1.0033   &  1.0026 & 0.1933   & 0.9864$^{\ast \ast \ast}$  &   0.9761$^{\ast \ast \ast}$  \\
   Bolivia &  0.2529   &  0.9982   &  0.9972$^{\ast}$ &    0.1682   & 0.9711$^{\ast \ast \ast}$   &  0.9601$^{\ast \ast \ast}$  \\
   Brazil & 0.2389   &  0.9894$^{\ast \ast}$   &  0.9899$^{\ast \ast \ast}$ &    0.1482  &  0.9618$^{\ast \ast \ast}$  &   0.9529$^{\ast \ast \ast}$  \\
   India &  0.2142   &  0.9912$^{\ast \ast}$   &  0.9905$^{\ast \ast}$ &     0.1607 &   0.9673$^{\ast \ast \ast}$   &  0.9494$^{\ast \ast \ast}$  \\
   Indonesia & 0.2696   &  1.0097   &  1.0079 &    0.1651  &  0.9787$^{\ast \ast \ast}$   &  0.9690$^{\ast \ast \ast}$  \\
   Mexico & 0.2425   &  0.9997   &  0.9986 &    0.1547  &  0.9749$^{\ast \ast \ast}$   &  0.9593$^{\ast \ast \ast}$  \\
   Philippines & 0.2513   &  0.9905$^{\ast \ast}$  &   0.9928 &     0.1532  &  0.9614$^{\ast \ast \ast}$   &  0.9488$^{\ast \ast \ast}$  \\
   Singapore & 0.7310   &  1.0068   &  1.0083 &     0.3501   & 0.9923$^{\ast}$  &  0.9912$^{\ast}$  \\
   Thailand & 0.3731    & 1.0033   &  1.0028 &     0.2010  &  0.9789$^{\ast \ast \ast}$   &  0.9710$^{\ast \ast \ast}$  \\
   Venezuela &  0.2563   &  1.0055   &  1.0019 &     0.1594  &  0.9831 $^{\ast \ast \ast}$   & 0.9692$^{\ast \ast \ast}$  \\

%Argentina & 0.3299  &  0.9958  &  0.9988 &   0.1774  &  0.9811$^{\ast \ast \ast}$  &  0.9783$^{\ast \ast \ast}$  \\
%Bolivia & 0.2293  &  0.9972  &  0.9916$^{\ast \ast \ast}$  &   0.1493  &  0.9767$^{\ast \ast \ast}$  &  0.9643$^{\ast \ast \ast}$ \\
%Brazil & 0.2796  &  0.9989  &  0.9953 &   0.1548  &  0.9806$^{\ast \ast \ast}$  &  0.9741$^{\ast \ast \ast}$ \\
%India & 0.2282  &  0.9880$^{\ast \ast \ast}$  &  0.9852$^{\ast \ast}$ &   0.1514  &  0.9689$^{\ast \ast \ast}$  &  0.9522$^{\ast \ast \ast}$ \\
%Indonesia & 0.2350  &  0.9971  &  1.0021 &   0.1438  &  0.9760$^{\ast \ast \ast}$  &  0.9672$^{\ast \ast \ast}$ \\
%Mexico & 0.2741  &  0.9936$^{\ast \ast}$  &  0.9902$^{\ast \ast}$ &   0.1523  &  0.9868$^{\ast \ast \ast}$  &  0.9755$^{\ast \ast \ast}$ \\
%Philippines & 0.2330  &  0.9932  &  0.9907$^{\ast \ast}$ &   0.1375  &  0.9710$^{\ast \ast \ast}$  &  0.9598$^{\ast \ast \ast}$ \\
%Singapore &  0.6871  &  0.9964  &  0.9976$^{\ast}$ &   0.3275  &  0.9900$^{\ast \ast \ast}$  &  0.9869$^{\ast \ast \ast}$ \\
%Thailand & 0.3746  &  0.9971  &  1.0006 &   0.1930  &  0.9831$^{\ast \ast \ast}$  &  0.9751$^{\ast \ast \ast}$ \\
%Venezuela & 0.2954  &  0.9817$^{\ast \ast \ast}$  &  0.9787$^{\ast \ast \ast}$ &   0.1639  &  0.9793$^{\ast \ast \ast}$  &  0.9713$^{\ast \ast \ast}$ \\
         \hline 
    \end{tabular}
    \end{adjustbox}
\caption{RMSE {(Columns 2-4)} and average CRPS {(Columns 5-7)} for the Dengue Data for each prior. Column Fixed $\nu=0$ provides the RMSE and the average CRPS; Columns Fixed $\nu =m+1$ and hyperprior provide the ratio between the referred prior and the flat prior. $^{\ast \ast \ast}$, $^{\ast \ast}$ and $^\ast$ indicate ratios are significantly different from $1$ at the $1\%$, $5\%$ and $10\%$ significance level according to the Diebold-Mariano test.}
    \label{Table_Dengue}
\end{table}

\section{Discussion}\label{sec_concl}
We have presented a novel method to perform forecasting in VAR models, where a hyperprior is set on the number of degrees of freedom of the covariance matrix for the {global-local-shrinkage}-Wishart prior. {The proposed loss-based prior takes into consideration only the intrinsic properties of the model, that is the sampling distribution and the priors. The method has been compared with what is currently used in the literature when no information about the parameter values of the {Horseshoe}-Wishart prior is available, where $\nu$ is typically fixed to $0$ or $m+1$.}

The analysis of simulated data has shown that, when the true value of $\nu$ is close to $m+1$, both approaches perform similarly. While, as one would expect, the farther the true $\nu$ is from $m+1$, the better the performance of the proposed prior. To illustrate in practice the advantage of having a loss-based prior on $\nu$, we have analysed its performance in terms of prediction on two datasets. One concerns macroeconomic variables from the FRED dataset, and the other is the analysis of infection variables from the GDT dataset. For both datasets, it appears that the proposed method outperforms the one currently used, in particular when the CRPS index is considered.

In support of our results, we have estimated the number of degrees of freedom using rolling windows. This analysis has shown that the data contains information for a value of $\nu$ always above the value of $m+1$, which justifies the use of the proposed method. In other words, as the data appear to have been ``generated'' by a Bayesian model with a relatively large number of degrees of freedom, by setting $\nu=m+1$, one impacts the predictive performance of the model. On the other hand, by assuming uncertainty of the value of $\nu$ by assigning an objective prior to it, the model is free to ``choose'' the most appropriate value of the parameter and, even considering the extra uncertainty that this implies, the results are better than the previous method.

{There are at least two possible research directions to further extend the work presented in this paper. Given the increasing importance of forecasting macroeconomic variables on policymakers' agendas, one may consider incorporating time-varying volatility, such as stochastic volatility, into the VAR framework when estimating the degrees of freedom through a loss-based prior \citep[see, e.g.,][]{Clark2015, Kastner2020}. Additionally, another interesting avenue is to explore the application of the proposed prior in analyzing matrix VAR in the form of the Wishart Autoregressive process of multivariate stochastic volatility \citep{Gourieroux2009}.}

\section*{Acknowledgments}
The authors are grateful to the two referees and the Associate Editor for the constructive comments on earlier versions of the paper which helped in improving the quality of this work. Luca Rossini acknowledges financial support from the Italian Ministry of University and Research (MUR) under the Department of Excellence 2023-2027 grant agreement ``Centre of Excellence in Economics and Data Science'' (CEEDS).

\bibliographystyle{apalike}
\bibliography{Biblio_Lossnu}

\appendix

\renewcommand{\thesection}{A.\arabic{section}}
\renewcommand{\theequation}{A.\arabic{equation}}
\renewcommand{\thefigure}{A.\arabic{figure}}
\renewcommand{\thetable}{A.\arabic{table}}
\setcounter{table}{0}
\setcounter{figure}{0}
\setcounter{equation}{0}

\section{Loss-based prior}\label{LBprior}
{In this section we illustrate the derivation of the loss-based prior in Equation \eqref{objpriorgeneral}.

Consider a probability distribution $f(x|\theta)$, where $\theta\in\Theta$ is an unknown discrete parameter.
The prior mass to be put on $\theta$ is derived by considering what is lost if the model $f(x|\theta)$ is removed and it is the true one. {The approach associates a {\em worth} to each parameter value measured by applying a result available in \cite{Berk1966} which states that, if a model is misspecified (i.e. if $\theta$ is removed and it is the true value) then the posterior distribution asymptotically accumulates at the value $\theta^\prime$ such that the Kullback--Leibler divergence \citep{KL1951} $KL(f(\cdot|\theta)\|f(\cdot|\theta^\prime))$ is minimised.} That is, if the true model is removed, the estimation process will asymptotically indicate as the correct model the nearest one, in terms of the Kullback--Leibler divergence; i.e., the model which is the most similar to the true one \citep{BS1994}. To link the {\em worth} of each parameter value to the prior probability, we use the {\em self-information} loss function. This particular type of loss function assigns a loss to a probability statement and, say we have defined prior $\pi(\theta)$, its form is $-\log\pi(\theta)$. More information about the self-information loss function can be found, for example, in \cite{MF1998}.

To formally derive the prior for $\theta$, we can proceed in terms of utilities, instead of losses; this approach allows for a clearer exposition and does not impact the logic behind the prior derivation. Let us then write utility {$u_1(\theta)=\log\pi(\theta)$}. We then let the minimum divergence from $\theta$ be represented by utility $u_2(\theta)$. We naturally want $u_1(\theta)$ and $u_2(\theta)$ to be matching utility functions; though as it stands {$-\infty<u_1\leq0$} and $0\leq u_2<\infty$, and we want $u_1=-\infty$ when $u_2=0$. The scales are matched by taking exponential transformations, so $\exp(u_1)$ and $\exp(u_2)-1$ are on the same scale and we obtain
\begin{align*}
    \pi(\theta) = e^{u_1(\theta)}\propto e^{u_2(\theta)}-1,
\end{align*}
yielding the loss-based prior for $\theta$ in Eq.~\eqref{objpriorgeneral}.}

\section{Further simulation results}
\label{App_A}

\subsection{Case $T=100$}

In this section, we provide further simulation results. We report the root mean absolute deviation (RMAD) for the case with $T=100$. In particular, Figure~\ref{Fig_Sim_M5_Supp} shows the RMAD for the covariance matrix for the five-dimensional case, when the data are generated from a Wishart distribution with degrees of freedom equal to $5$ (left), $10$ (center), and $15$ (right). Table~\ref{tab:IRF_Sim_M5_Supp} provides the RMAD for the impulse response function at four different horizons $h=1,3,5$ and $7$. 

\begin{figure}[h!]
% 	{\includegraphics[width=0.33\linewidth]{Figures_Revision/RootMAD,T=100,M=5,dof=5_HS_Rev}}\hfill%
%	{\includegraphics[width=0.33\linewidth]{Figures_Revision/RootMAD,T=100,M=5,dof=10_HS_Rev}}\hfill%
%	{\includegraphics[width=0.33\linewidth]{Figures_Revision/RootMAD,T=100,M=5,dof=15_HS_Rev}}
 	{\includegraphics[width=0.33\linewidth]{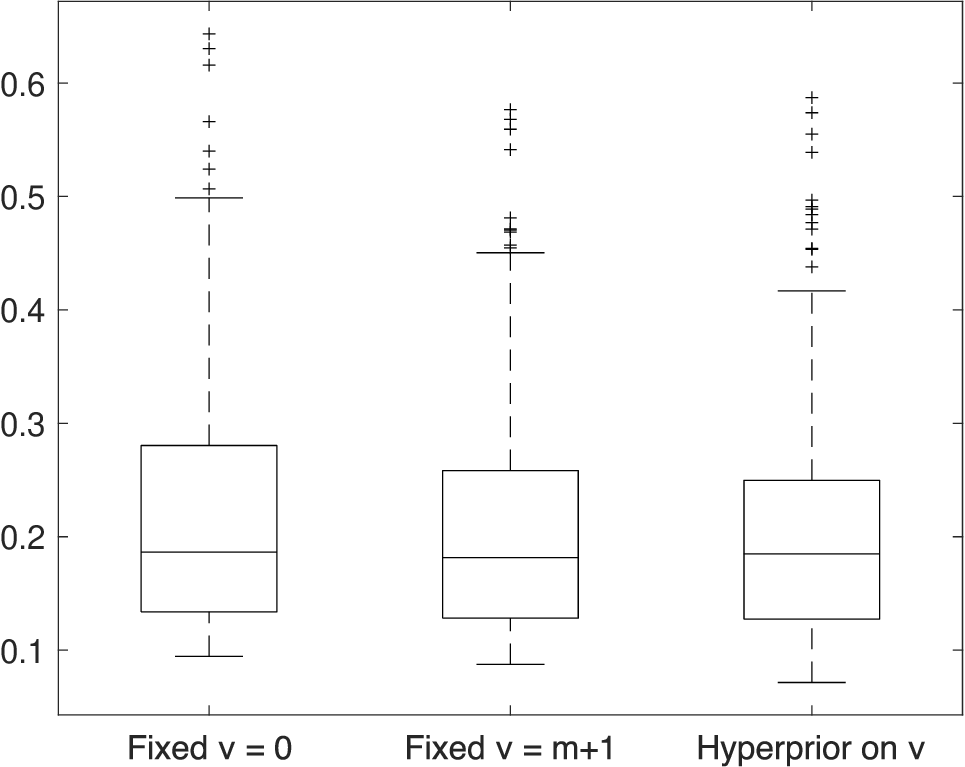}}\hfill%
	{\includegraphics[width=0.33\linewidth]{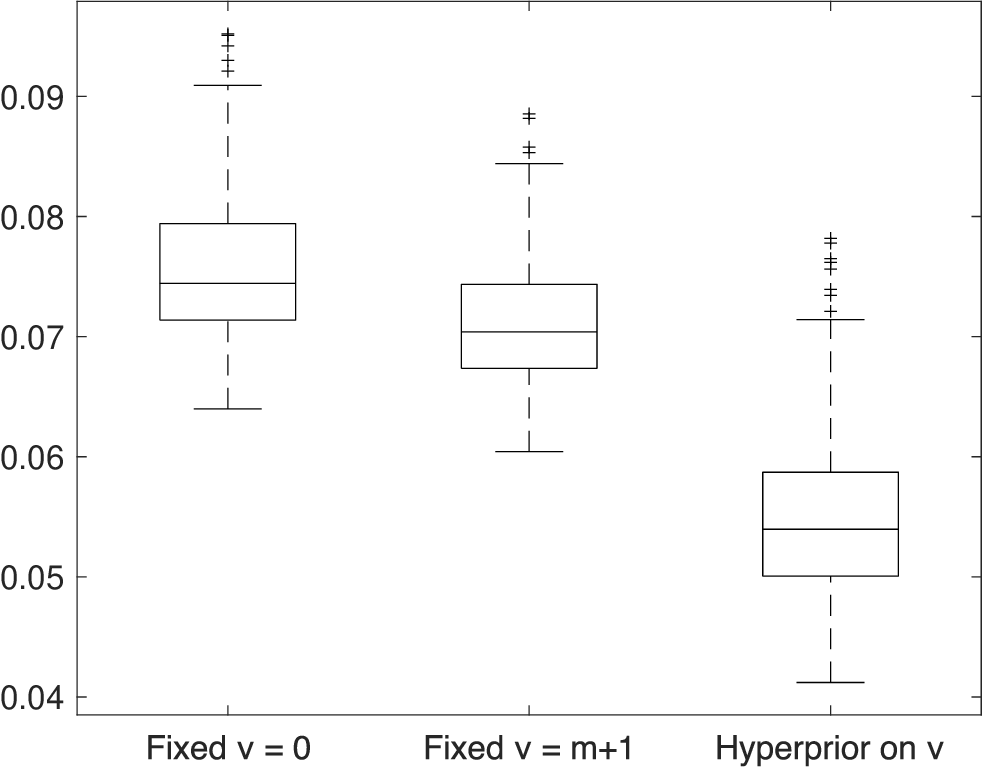}}\hfill%
	{\includegraphics[width=0.33\linewidth]{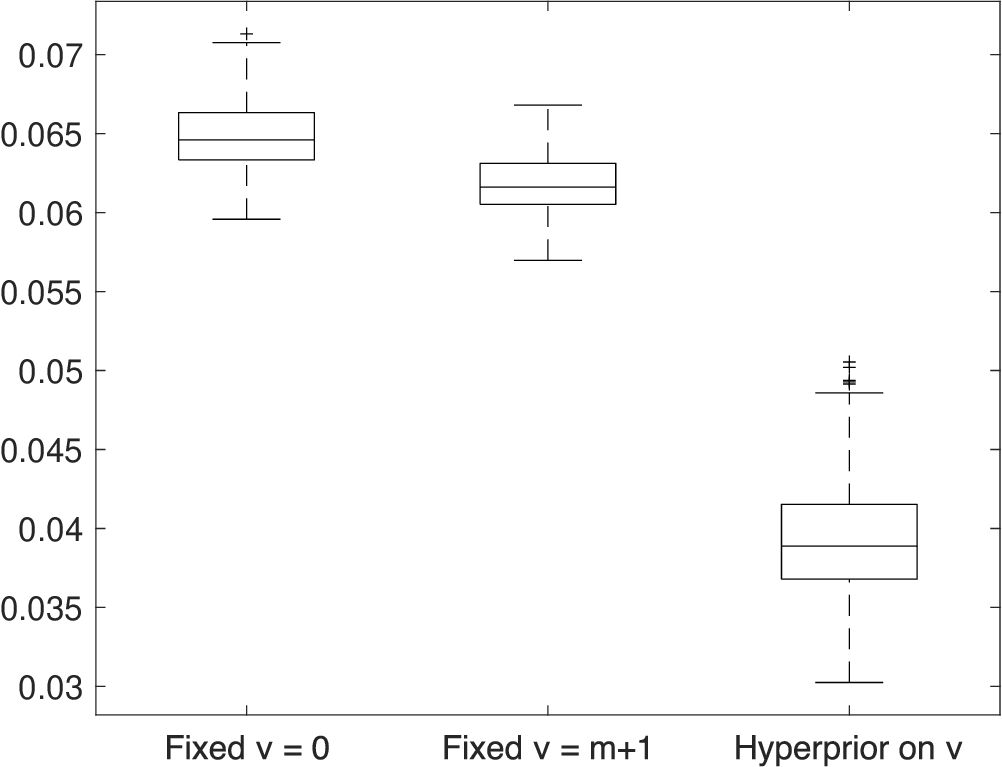}}
 \caption{
 {Monte Carlo Simulation - RMAD of the covariance matrices. These distributions are obtained by simulating $250$ VAR(1) with dimension $m=5$ and sample size $T=100$. Results are reported for data generated from a Wishart distribution with $\nu = 5$ (left), $\nu=10$ (center), and $\nu = 15$ (right).}}
% Monte Carlo Simulation - Root Mean Absolute Deviations of the covariance matrices of dimension $m=5$. These empirical distributions are obtained by simulating $250$ VAR(1) of sample size $T=100$. Results are reported separately for data generated from a Wishart with $\nu= 5$ (left),  $\nu =10$ (center), and $\nu = 15$ (right).}
	\label{Fig_Sim_M5_Supp}
\end{figure}

\begin{table}[h!]
    \centering
    \begin{adjustbox}{width=1\textwidth,center=\textwidth}
    \begin{tabular}{c|c c c| c c c| c c c}
    \hline 
    & \multicolumn{3}{c}{$\nu = 5$} & \multicolumn{3}{c}{$\nu = 10$} & \multicolumn{3}{c}{$\nu = 15$} \\
        Horizon & Fixed & Fixed & Hyperprior & Fixed & Fixed & Hyperprior & Fixed & Fixed & Hyperprior  \\
         & $\nu =0$ & $\nu = m+1$ & & $\nu =0$ & $\nu = m+1$ & & $\nu =0$ & $\nu = m+1$ \\
         \hline
     1 &        0.4110  &  0.9960  &  0.9950 &        0.2830  &  0.9980  &  0.9890 &    0.2720 &   0.9990    0.9900 \\
   3 & 0.4860  &  1.0010  &  1.0010 &     0.4150   & 0.9880  &  0.9820 &     0.3160  &  0.9940  &  0.9330 \\
   5 & 0.4240  &  0.9960 &   0.9940 &     0.3310  &  0.9860  &  0.9180 &     0.4380  &  0.9800  &  0.8050 \\
  7 &   0.3030  &  0.9970  &  0.9820 &     0.4400  &  0.9720  &  0.8410 &     0.7390  &  0.9670  &  0.6940 \\
         \hline 
    \end{tabular}
    \end{adjustbox}
    \caption{{Monte Carlo Simulation - RMAD of the IRF for four horizons $h=1, 3, 5$ and $7$ by simulating $250$ VAR(1) with dimension $m=5$ and sample size $T=100$. Column Fixed $\nu=0$ provides the RMAD of the IRF, while Columns Fixed $\nu = m+1$ and Hyperprior provide the ratio between the referred priors and the flat prior.}}  
 %   Monte Carlo Simulation - Root Mean Absolute Deviation of the impulse response functions (IRF) of dimension $m=5$ for four different horizons $h=1, 3, 5, 7$ by simulating $250$ VAR(1) of sample size $T=100$. Column Fixed $\nu=0$ provides the RMAD of the IRF, while Columns Fixed $\nu = m+1$ and Hyperprior provide the ratio between the referred priors and the flat prior.}
    \label{tab:IRF_Sim_M5_Supp}
\end{table}

\FloatBarrier

In Figure~\ref{Fig_Sim_M10_Supp}, we show results for a ten-dimensional case for the matrix of covariance and with data generated with $10$ (left), $15$ (center), and $20$ (right) degrees of freedom. Table~\ref{tab:IRF_Sim_M10_Supp} shows the IRF for the same dataset for four different horizons.

\begin{figure}[h!]
%	{\includegraphics[width=0.33\linewidth]{Figures_Revision/RootMAD,T=100,M=10,dof=10_HS_Rev}}\hfill%
%	{\includegraphics[width=0.33\linewidth]{Figures_Revision/RootMAD,T=100,M=10,dof=15_HS_Rev}}\hfill%
%	{\includegraphics[width=0.33\linewidth]{Figures_Revision/RootMAD,T=100,M=10,dof=20_HS_Rev}}
	{\includegraphics[width=0.33\linewidth]{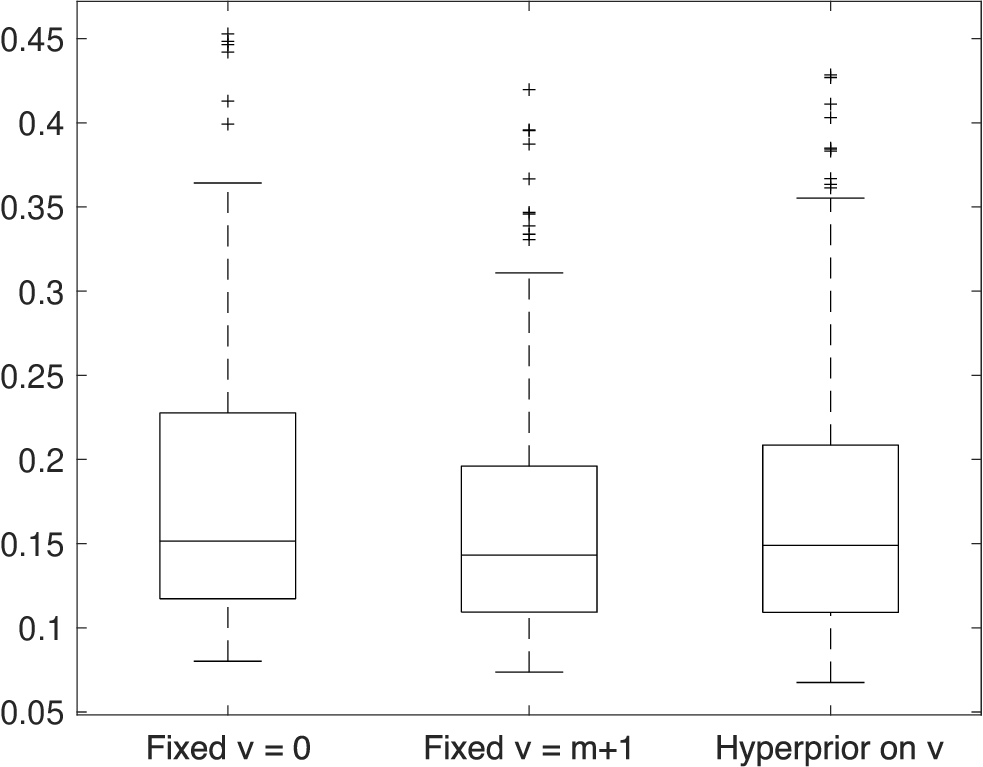}}\hfill%
	{\includegraphics[width=0.33\linewidth]{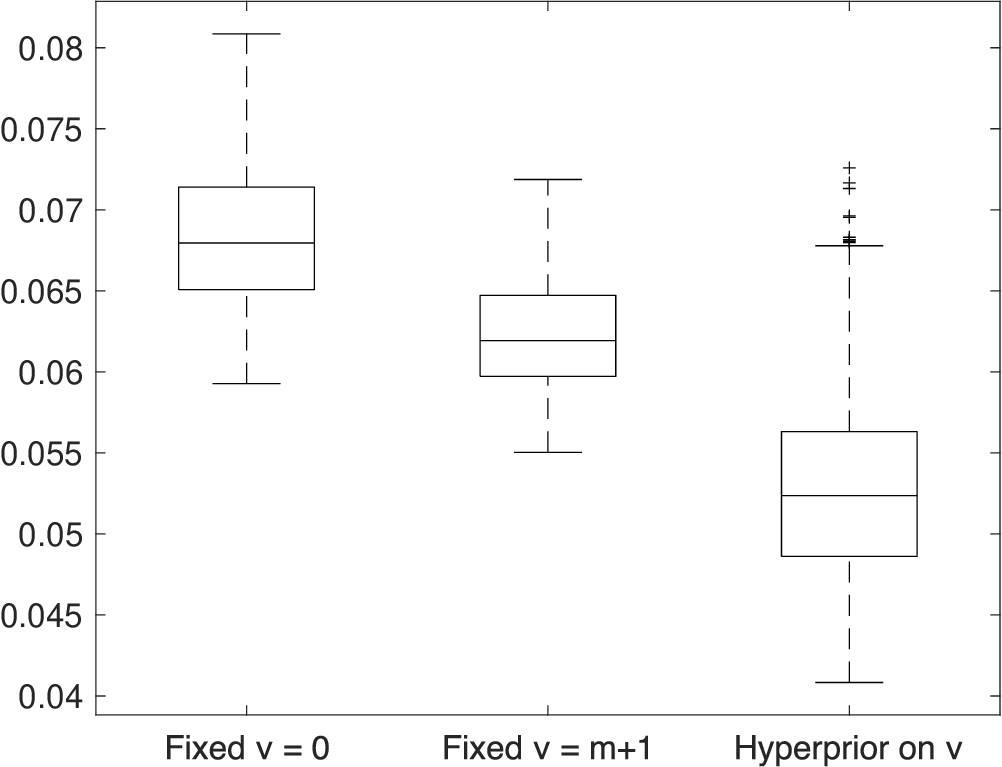}}\hfill%
	{\includegraphics[width=0.33\linewidth]{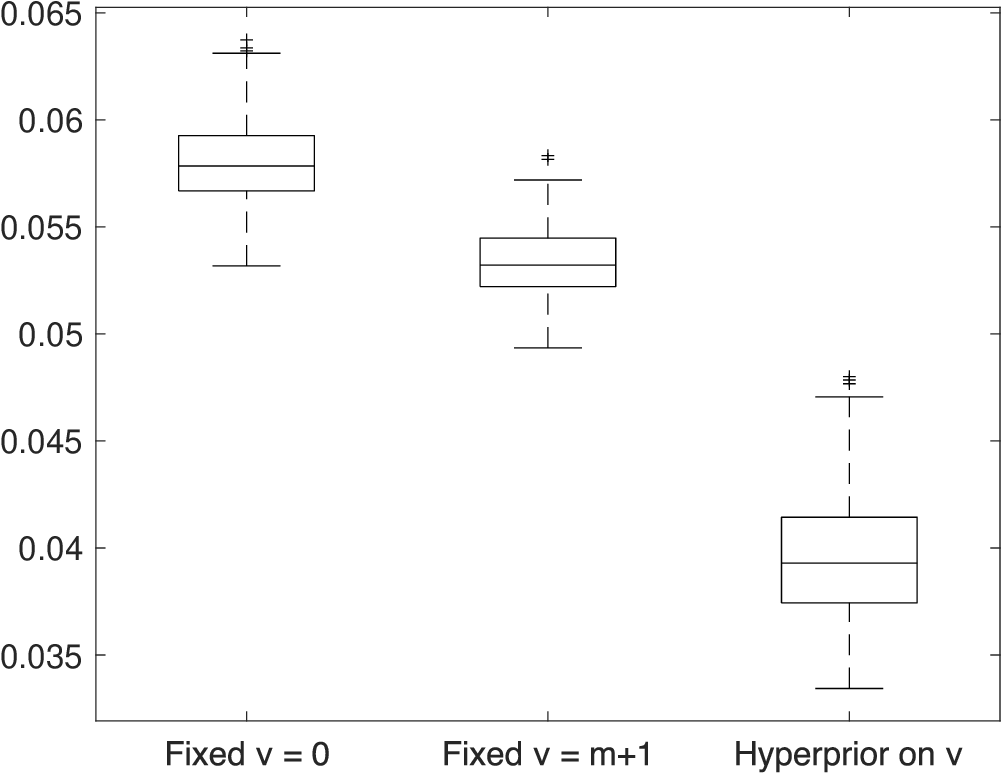}}
	\caption{{Monte Carlo Simulation - RMAD of the covariance matrices. These distributions are obtained by simulating $250$ VAR(1) with dimension $m=10$ and sample size $T=100$. Results are reported for data generated from a Wishart distribution with $\nu = 10$ (left), $\nu=15$ (center), and $\nu = 20$ (right).}}
% Monte Carlo Simulation - Root Mean Absolute Deviations of the covariance matrices of dimension $m=10$. These empirical distributions are obtained by simulating $250$ VAR(1) of sample size $T=100$. Results are reported separately for data generated from a Wishart with $\nu = 10$ (left),  $\nu=15$ (center), and $\nu = 20$ (right).}
	\label{Fig_Sim_M10_Supp}
\end{figure}

\begin{table}[h!]
    \centering
    \begin{adjustbox}{width=1\textwidth,center=\textwidth}
    \begin{tabular}{c|c c c| c c c| c c c}
    \hline 
    & \multicolumn{3}{c}{$\nu = 10$} & \multicolumn{3}{c}{$\nu = 15$} & \multicolumn{3}{c}{$\nu = 20$} \\
        Horizon & Fixed & Fixed & Hyperprior & Fixed & Fixed & Hyperprior & Fixed & Fixed & Hyperprior  \\
         & $\nu =0$ & $\nu = m+1$ & & $\nu =0$ & $\nu = m+1$ & & $\nu =0$ & $\nu = m+1$ \\
         \hline
     1 &        0.5400  &  0.9960   & 0.9930   &  0.3460 &   0.9970  &  0.9900 &        0.3090  &  0.9990  &  0.9920 \\
   3 & 0.5880  &  1.0060  &  1.0070  &   0.3420 &   0.9930 &   0.9760 &    0.3420  &  0.9890  &  0.9450 \\
   5 &  0.6310  &  0.9890  &  0.9820  &   0.5240 &   0.9590  &  0.8570 &    0.6860  &  0.9580 &  0.7860 \\
  7 &  0.9960  &  0.9600  &  0.9360 & 1.0960   & 0.9330  &  0.7680   &1.7110  &  0.9380  &  0.6890  \\
         \hline 
    \end{tabular}
    \end{adjustbox}
    \caption{{Monte Carlo Simulation - RMAD of the IRF for four horizons $h=1, 3, 5$ and $7$ by simulating $250$ VAR(1) with dimension $m=10$ and sample size $T=100$. Column Fixed $\nu=0$ provides the RMAD of the IRF, while Columns Fixed $\nu = m+1$ and Hyperprior provide the ratio between the referred priors and the flat prior.}}
 %   Monte Carlo Simulation - Root Mean Absolute Deviation of the impulse response functions (IRF) of dimension $m=10$ for four different horizons $h=1, 3, 5, 7$ by simulating $250$ VAR(1) of sample size $T=100$. Column Fixed $\nu=0$ provides the RMAD of the IRF, while Columns Fixed $\nu = m+1$ and Hyperprior provide the ratio between the referred priors and the flat prior.}
    \label{tab:IRF_Sim_M10_Supp}
\end{table}

In conclusion, Figure~\ref{Fig_Sim_M20_Supp} shows the results for the twenty-dimensional case, where the data are generated from a Wishart with $20$ (left), $24$ (center), and $26$ (right) degrees of freedom. Table~\ref{tab:IRF_Sim_M20_Supp} shows the RMAD for the impulse response functions at four horizons.

\begin{figure}[h!]
% 	{\includegraphics[width=0.33\linewidth]{Figures_Revision/RootMAD,T=100,M=20,dof=20_HS_Rev}}\hfill%
%	{\includegraphics[width=0.33\linewidth]{Figures_Revision/RootMAD,T=100,M=20,dof=24_HS_Rev}}\hfill%
%	{\includegraphics[width=0.33\linewidth]{Figures_Revision/RootMAD,T=100,M=20,dof=26_HS_Rev}}
 	{\includegraphics[width=0.33\linewidth]{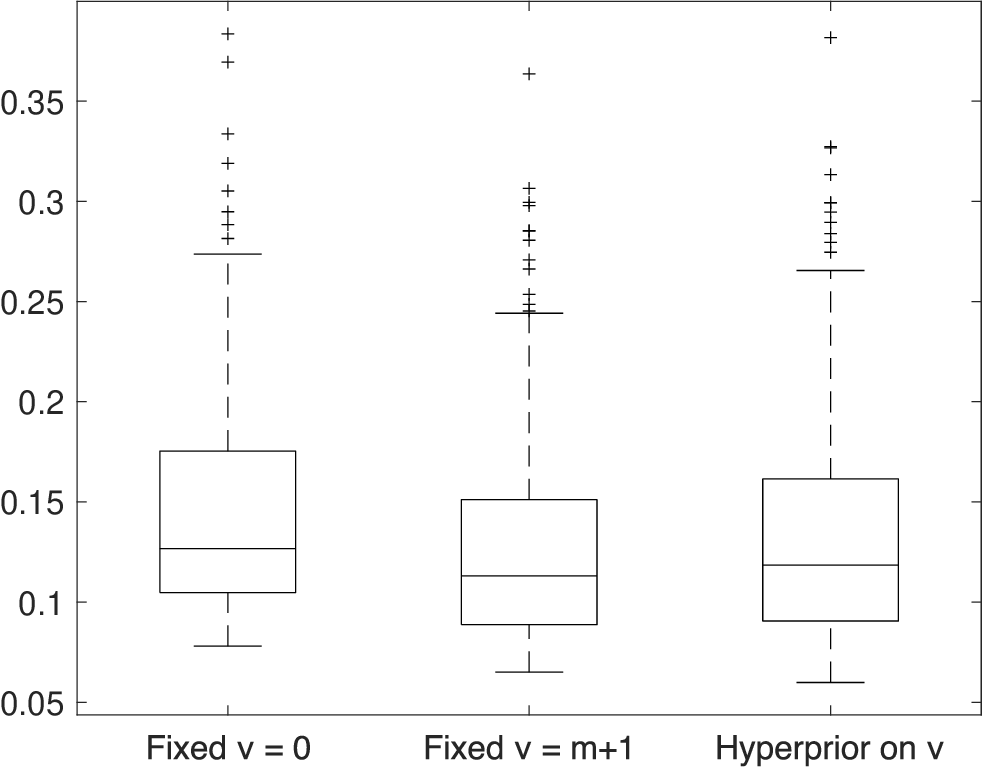}}\hfill%
	{\includegraphics[width=0.33\linewidth]{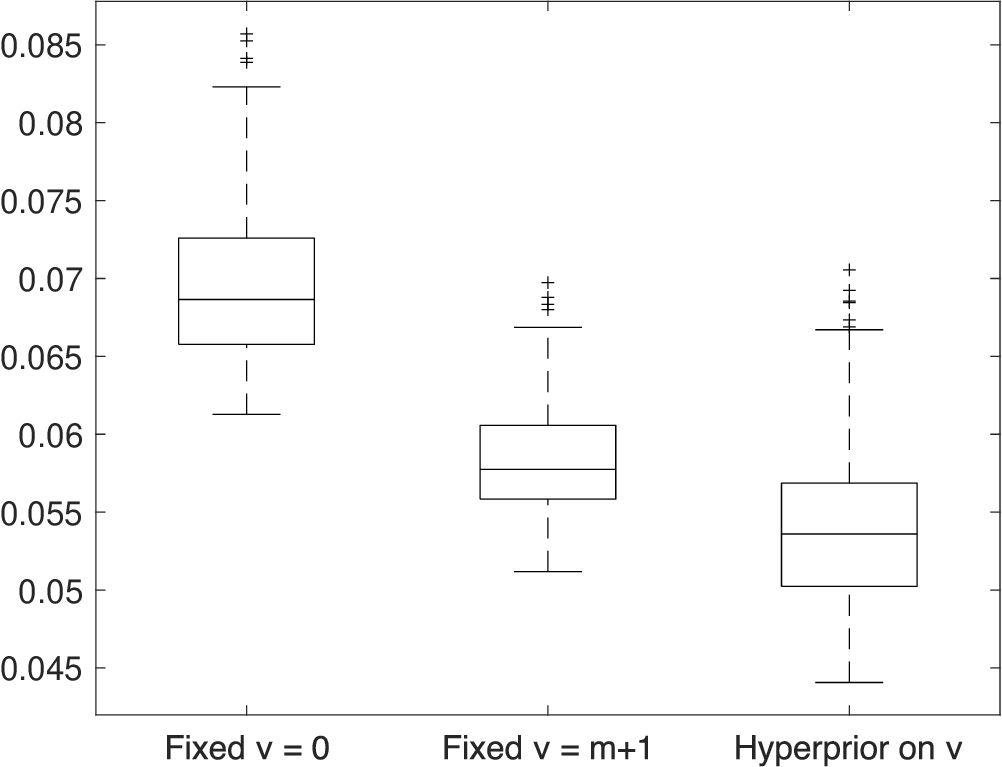}}\hfill%
	{\includegraphics[width=0.33\linewidth]{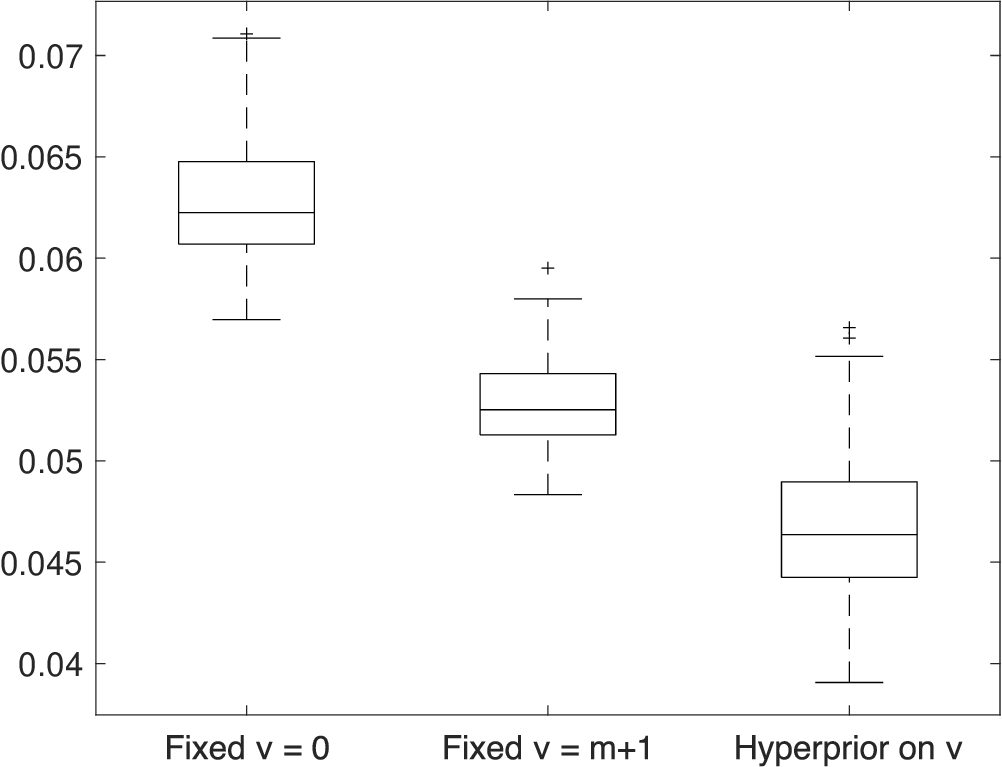}}
	\caption{{Monte Carlo Simulation - RMAD of the covariance matrices. These distributions are obtained by simulating $250$ VAR(1) with dimension $m=20$ and sample size $T=100$. Results are reported for data generated from a Wishart distribution with $\nu = 20$ (left), $\nu=24$ (center), and $\nu = 26$ (right).}}
% Monte Carlo Simulation - Root Mean Absolute Deviations of the covariance matrices of dimension $m=20$. These empirical distributions are obtained by simulating $250$ VAR(1) of sample size $T=100$. Results are reported separately for data generated from a Wishart with $\nu = 20$ (left)  $\nu=24$ (center), and $\nu = 26$ (right).}
	\label{Fig_Sim_M20_Supp}
\end{figure}

\begin{table}[h!]
    \centering
    \begin{adjustbox}{width=1\textwidth,center=\textwidth}
    \begin{tabular}{c|c c c| c c c| c c c}
    \hline 
    & \multicolumn{3}{c}{$\nu = 20$} & \multicolumn{3}{c}{$\nu = 24$} & \multicolumn{3}{c}{$\nu = 26$} \\
        Horizon & Fixed & Fixed & Hyperprior & Fixed & Fixed & Hyperprior & Fixed & Fixed & Hyperprior  \\
         & $\nu =0$ & $\nu = m+1$ & & $\nu =0$ & $\nu = m+1$ & & $\nu =0$ & $\nu = m+1$ \\
         \hline
     1 &  0.6200   &  0.9960    & 0.9950 &    0.4010   &  0.9970   &  0.9950 &    0.3680 &   0.9980 &   0.9950 \\
    3 & 0.6440   &  1.0210  &   1.0250 &     0.3910   &  0.9890    & 0.9820 &     0.3750 &   0.9860  &  0.9710 \\
    5 & 1.1930   &  0.9600   &  0.9510 &     1.0610   &  0.9280   &  0.8680 &     1.1370  &  0.9320  &  0.8580 \\
    7 & 3.8090   &  0.9230   &  0.9010 &     3.7890   &  0.8930   &  0.8070 &     4.2470  &  0.8990  &  0.7910 \\
    \hline 
    \end{tabular}
    \end{adjustbox}
    \caption{{Monte Carlo Simulation - RMAD of the IRF for four horizons $h=1, 3, 5$ and $7$ by simulating $250$ VAR(1) with dimension $m=20$ and sample size $T=100$. Column Fixed $\nu=0$ provides the RMAD of the IRF, while Columns Fixed $\nu = m+1$ and Hyperprior provide the ratio between the referred priors and the flat prior.}}  
 %   Monte Carlo Simulation - Root Mean Absolute Deviation of the impulse response functions (IRF) of dimension $m=20$ for four different horizons $h=1, 3, 5, 7$ by simulating $250$ VAR(1) of sample size $T=100$. Column Fixed $\nu=0$ provides the RMAD of the IRF, while Columns Fixed $\nu = m+1$ and Hyperprior provide the ratio between the referred priors and the flat prior.}
    \label{tab:IRF_Sim_M20_Supp}
\end{table}

\FloatBarrier

As stated in the paper, the results show improvements in the use of our loss-based prior with respect to a fixed $\nu$ prior when the data are generated with a higher than the dimension degrees of freedom.

\FloatBarrier

\subsection{Case $T=240$}

{As a third simulated experiment, we report the RMAD for the case with $T=240$. Figure~\ref{Fig_Sim_M3_Supp} shows the RMAD for the covariance matrix for the three-dimensional case when the data are generated from a Wishart distribution with degrees of freedom equal to $3$ (left), $5$ (center), and $7$ (right). Table~\ref{tab:IRF_Sim_M3_Supp} shows the RMAD for the impulse response variable at four different horizons $h=1,3,5$, and $7$.}

\begin{figure}[h!]
% 	{\includegraphics[width=0.33\linewidth]{Figures_Revision/RootMAD,T=240,M=3,dof=3_HS_Rev.eps}}\hfill%
%	{\includegraphics[width=0.33\linewidth]{Figures_Revision/RootMAD,T=240,M=3,dof=5_HS_Rev.eps}}\hfill%
%	{\includegraphics[width=0.33\linewidth]{Figures_Revision/RootMAD,T=240,M=3,dof=7_HS_Rev.eps}}
 	{\includegraphics[width=0.33\linewidth]{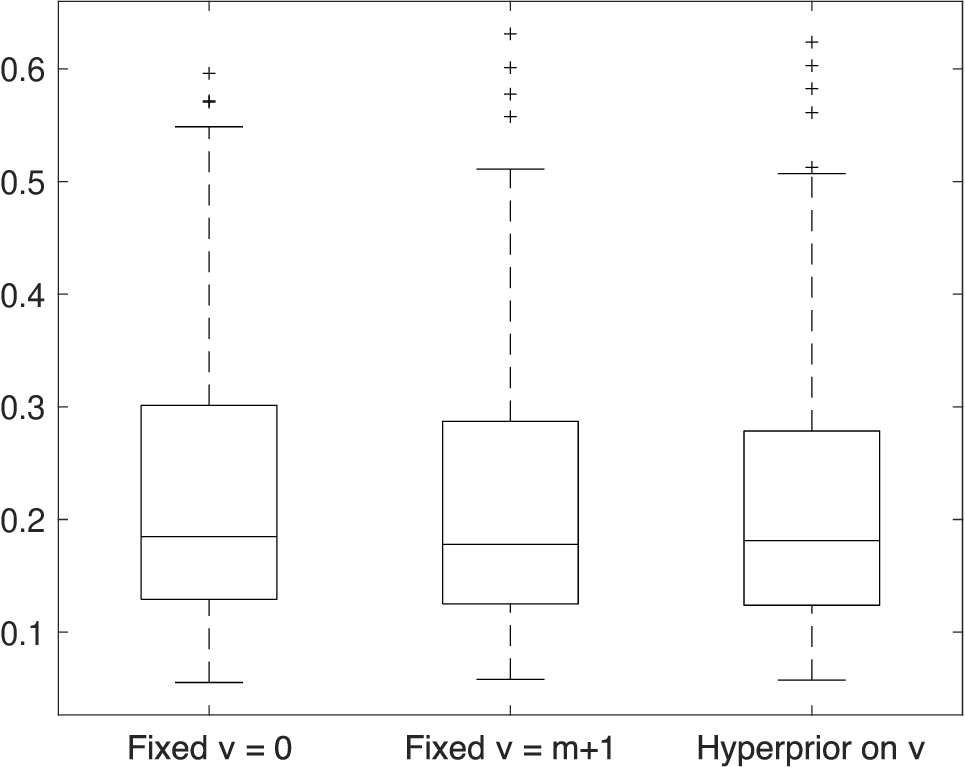}}\hfill%
	{\includegraphics[width=0.33\linewidth]{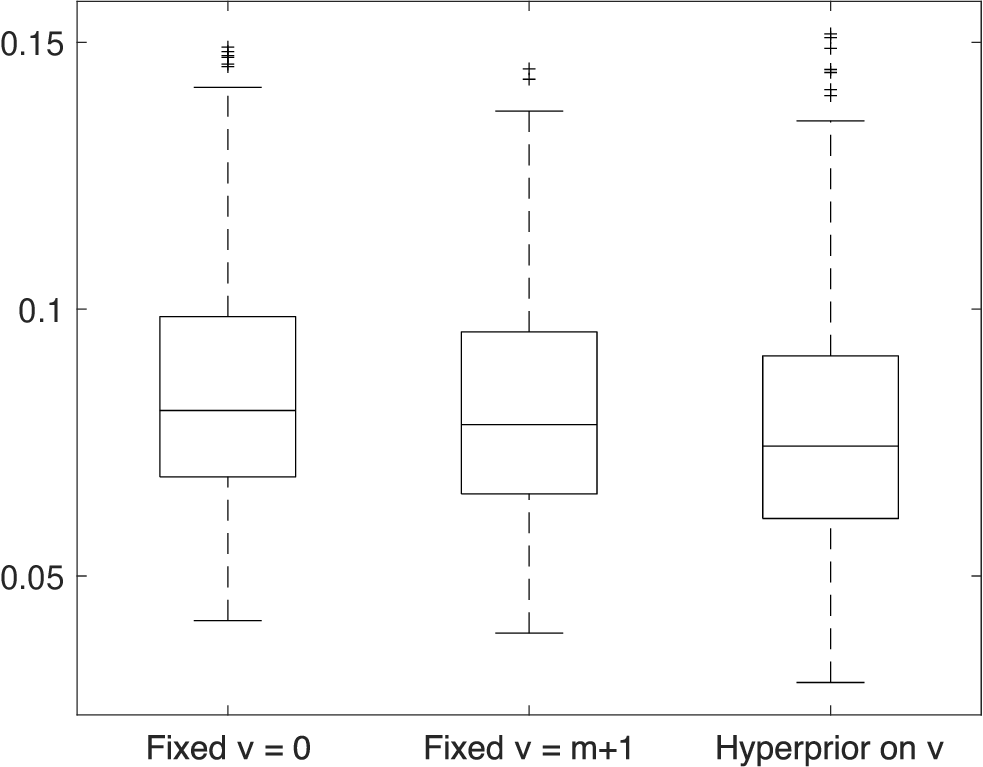}}\hfill%
	{\includegraphics[width=0.33\linewidth]{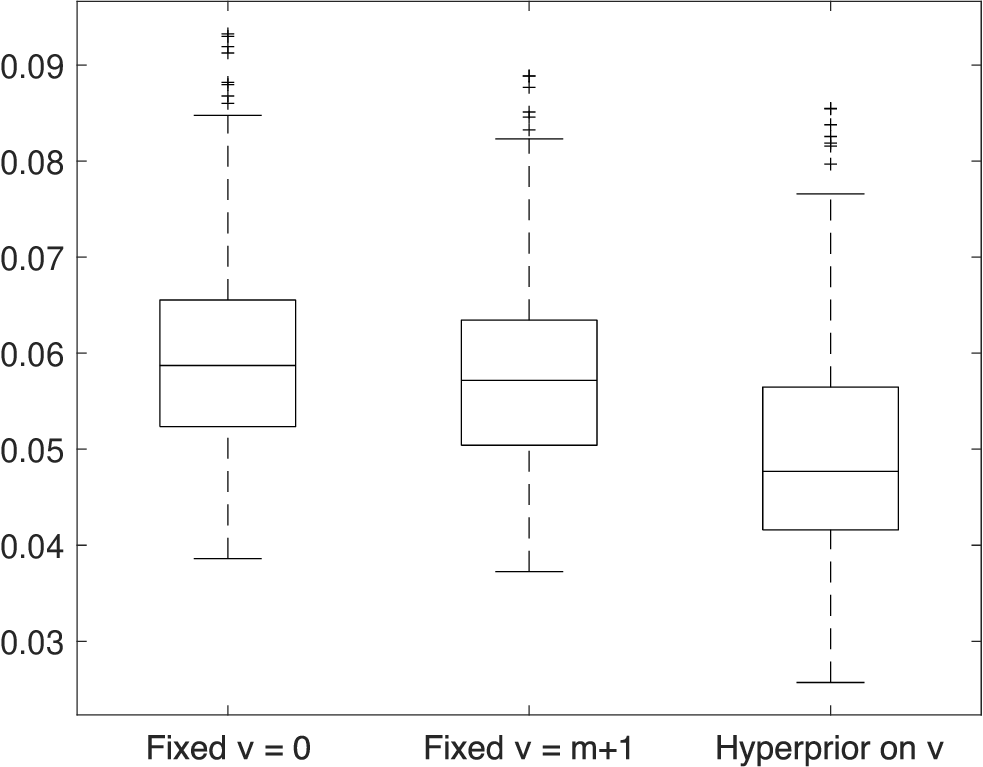}}
	\caption{{Monte Carlo Simulation - RMAD of the covariance matrices. These distributions are obtained by simulating $250$ VAR(1) with dimension $m=3$ and sample size $T=240$. Results are reported for data generated from a Wishart distribution with $\nu = 3$ (left), $\nu=5$ (center), and $\nu = 7$ (right).}}
% Monte Carlo Simulation - Root Mean Absolute Deviations of the covariance matrices of dimension $m=3$. These empirical distributions are obtained by simulating $250$ VAR(1) of sample size $T=240$. Results are reported separately for data generated from a Wishart with $\nu= 3$ (left),  $\nu =5$ (center), and $\nu = 7$ (right).}
	\label{Fig_Sim_M3_Supp}
\end{figure}

\begin{table}[h!]
    \centering
    \begin{adjustbox}{width=1\textwidth,center=\textwidth}
    \begin{tabular}{c|c c c| c c c| c c c}
    \hline 
    & \multicolumn{3}{c}{$\nu = 3$} & \multicolumn{3}{c}{$\nu = 5$} & \multicolumn{3}{c}{$\nu = 7$} \\
        Horizon & Fixed & Fixed & Hyperprior & Fixed & Fixed & Hyperprior & Fixed & Fixed & Hyperprior  \\
         & $\nu =0$ & $\nu = m+1$ & & $\nu =0$ & $\nu = m+1$ & & $\nu =0$ & $\nu = m+1$ \\
         \hline
     1 &    0.2150   &   0.9990   &   0.9990 &    0.1750  &  0.9990 &   0.9970 &   0.1750   & 0.9990  &  0.9960 \\
     3 & 0.2990   &   1.0000   &   1.0010 &      0.2280   & 1.0000 &   1.0000   &  0.2080  &  1.0000  &  0.9970 \\
     5 & 0.2400   &   1.0000   &   1.0010 &      0.1930 &   0.9990 &   0.9980  &   0.1850 &   0.9990  &  0.9920 \\
    7  & 0.2020   &   0.9990  &    1.0000 &     0.1740  &  0.9980   & 0.9950  &   0.1710   & 0.9980   & 0.9820 \\
    \hline 
    \end{tabular}
    \end{adjustbox}
    \caption{{Monte Carlo Simulation - RMAD of the IRF for four horizons $h=1, 3, 5$ and $7$ by simulating $250$ VAR(1) with dimension $m=3$ and sample size $T=240$. Column Fixed $\nu=0$ provides the RMAD of the IRF, while Columns Fixed $\nu = m+1$ and Hyperprior provide the ratio between the referred priors and the flat prior.}}
 %   Monte Carlo Simulation - Root Mean Absolute Deviation of the impulse response functions (IRF) of dimension $m=3$ for four different horizons $h=1, 3, 5, 7$ by simulating $250$ VAR(1) of sample size $T=240$. Column Fixed $\nu=0$ provides the RMAD of the IRF, while Columns Fixed $\nu = m+1$ and Hyperprior provide the ratio between the referred priors and the flat prior.}
    \label{tab:IRF_Sim_M3_Supp}
\end{table}

\FloatBarrier

{In Figure~\ref{Fig_Sim_M7_Supp}, we study the seven-dimensional case for the matrix of covariances and data generated with $7$ (left), $10$ (center), and $13$ (right) degrees of freedom. The RMAD of the impulse response functions for four different horizons are reported in Table~\ref{tab:IRF_Sim_M7_Supp}.}

\begin{figure}[h!]
%	{\includegraphics[width=0.33\linewidth]{Figures_Revision/RootMAD,T=240,M=7,dof=7_HS_Rev.eps}}\hfill%
%	{\includegraphics[width=0.33\linewidth]{Figures_Revision/RootMAD,T=240,M=7,dof=10_HS_Rev.eps}}\hfill%
%	{\includegraphics[width=0.33\linewidth]{Figures_Revision/RootMAD,T=240,M=7,dof=13_HS_Rev}}
	{\includegraphics[width=0.33\linewidth]{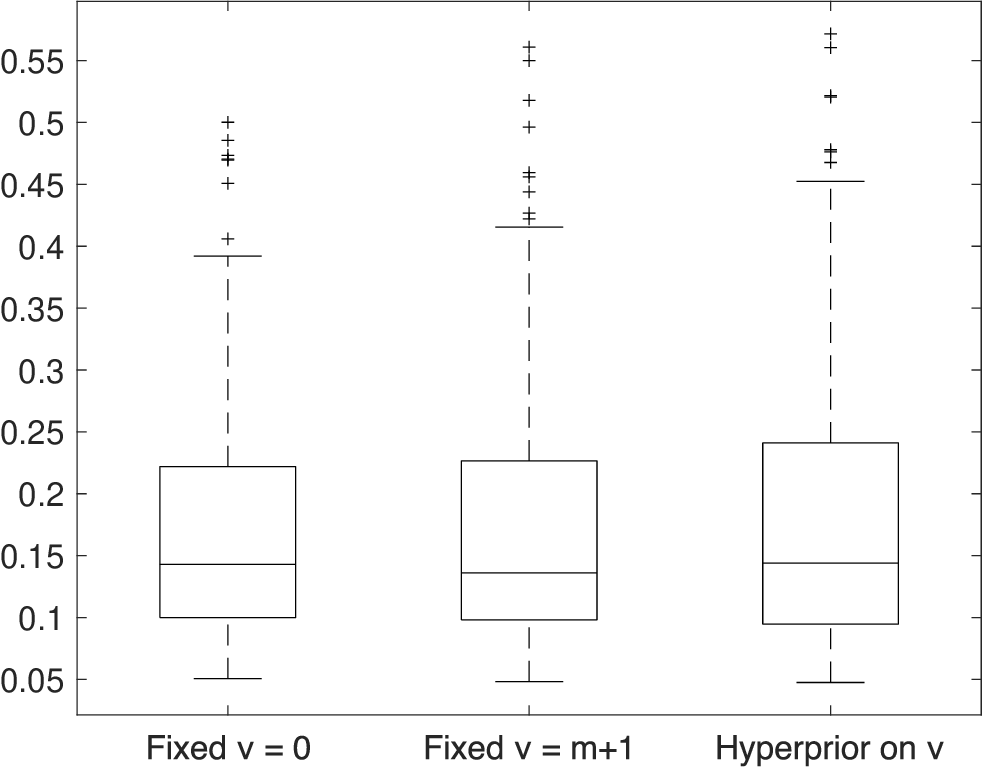}}\hfill%
	{\includegraphics[width=0.33\linewidth]{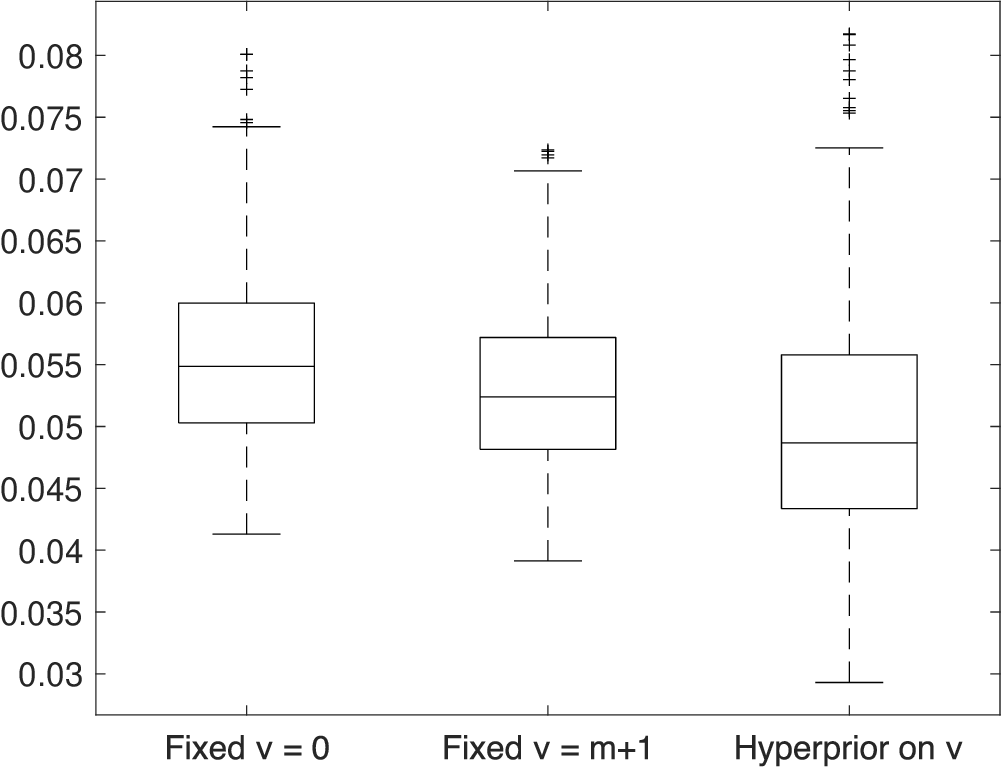}}\hfill%
	{\includegraphics[width=0.33\linewidth]{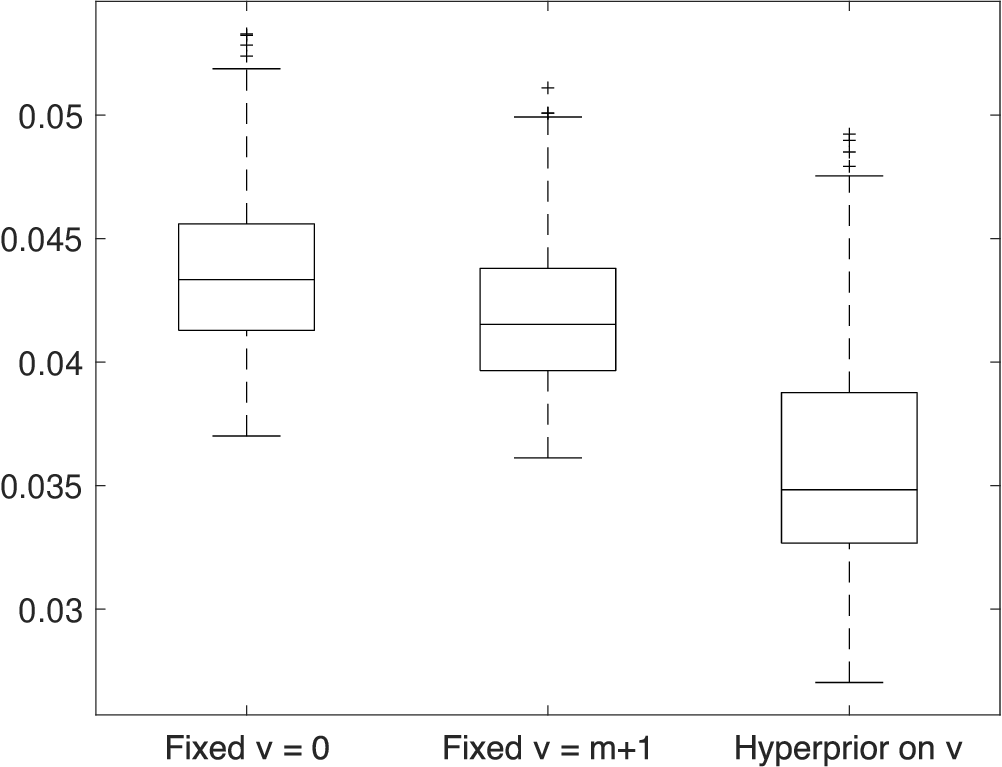}}	
 \caption{{Monte Carlo Simulation - RMAD of the covariance matrices. These distributions are obtained by simulating $250$ VAR(1) with dimension $m=7$ and sample size $T=240$. Results are reported for data generated from a Wishart distribution with $\nu = 7$ (left), $\nu=10$ (center), and $\nu = 13$ (right).}}
	\label{Fig_Sim_M7_Supp}
\end{figure}

\begin{table}[h!]
    \centering
    \begin{adjustbox}{width=1\textwidth,center=\textwidth}
    \begin{tabular}{c|c c c| c c c| c c c}
    \hline 
    & \multicolumn{3}{c}{$\nu = 7$} & \multicolumn{3}{c}{$\nu = 10$} & \multicolumn{3}{c}{$\nu = 13$} \\
        Horizon & Fixed & Fixed & Hyperprior & Fixed & Fixed & Hyperprior & Fixed & Fixed & Hyperprior  \\
         & $\nu =0$ & $\nu = m+1$ & & $\nu =0$ & $\nu = m+1$ & & $\nu =0$ & $\nu = m+1$ \\
         \hline
     1 &    0.4740 &   0.9980  &  0.9960&    0.2730   &   0.9980   &   0.9900 &    0.2620 &   0.9990 &   0.9880 \\
     3 &     0.5820 &   1.0010 &   1.0000&    0.2750  &    0.9990   &   0.9970 &     0.2520 &   0.9990  &  0.9910 \\
     5 &     0.5160  &  0.9990 &   0.9970&    0.2450  &    0.9970   &   0.9870 &     0.2380  &  0.9950  &  0.9660 \\
     7 &    0.5210  &  0.9950 &   0.9950&     0.2420   &   0.9930   &   0.9660 &     0.2570  &  0.9890  &  0.9240 \\
    \hline 
    \end{tabular}
    \end{adjustbox}
    \caption{{Monte Carlo Simulation - RMAD of the IRF for four horizons $h=1, 3, 5$ and $7$ by simulating $250$ VAR(1) with dimension $m=7$ and sample size $T=240$. Column Fixed $\nu=0$ provides the RMAD of the IRF, while Columns Fixed $\nu = m+1$ and Hyperprior provide the ratio between the referred priors and the flat prior.}}
    \label{tab:IRF_Sim_M7_Supp}
\end{table}

\FloatBarrier

{In conclusion, Figure~\ref{Fig_Sim_M15_Supp} shows the results for the fifteen-dimensional case, where the data are generated from a Wishart with $15$ (left); $20$ (center), and $25$ (right) degrees of freedom. Table~\ref{tab:IRF_Sim_M15_Supp} shows the RMAD for the impulse response function for four different horizons.}

\begin{figure}[h!]
%	{\includegraphics[width=0.33\linewidth]{Figures_Revision/RootMAD,T=240,M=15,dof=15_HS_Rev.eps}}\hfill%
%	{\includegraphics[width=0.33\linewidth]{Figures_Revision/RootMAD,T=240,M=15,dof=20_HS_Rev.eps}}\hfill%
%	{\includegraphics[width=0.33\linewidth]{Figures_Revision/RootMAD,T=240,M=15,dof=25_HS_Rev.eps}}
	{\includegraphics[width=0.33\linewidth]{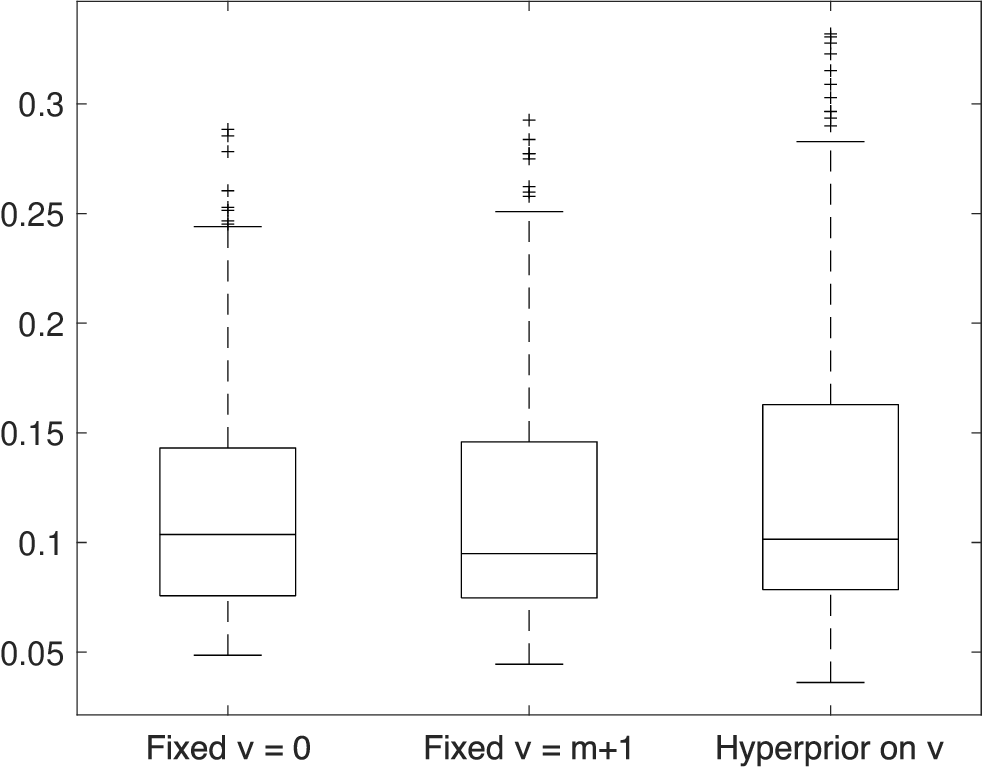}}\hfill%
	{\includegraphics[width=0.33\linewidth]{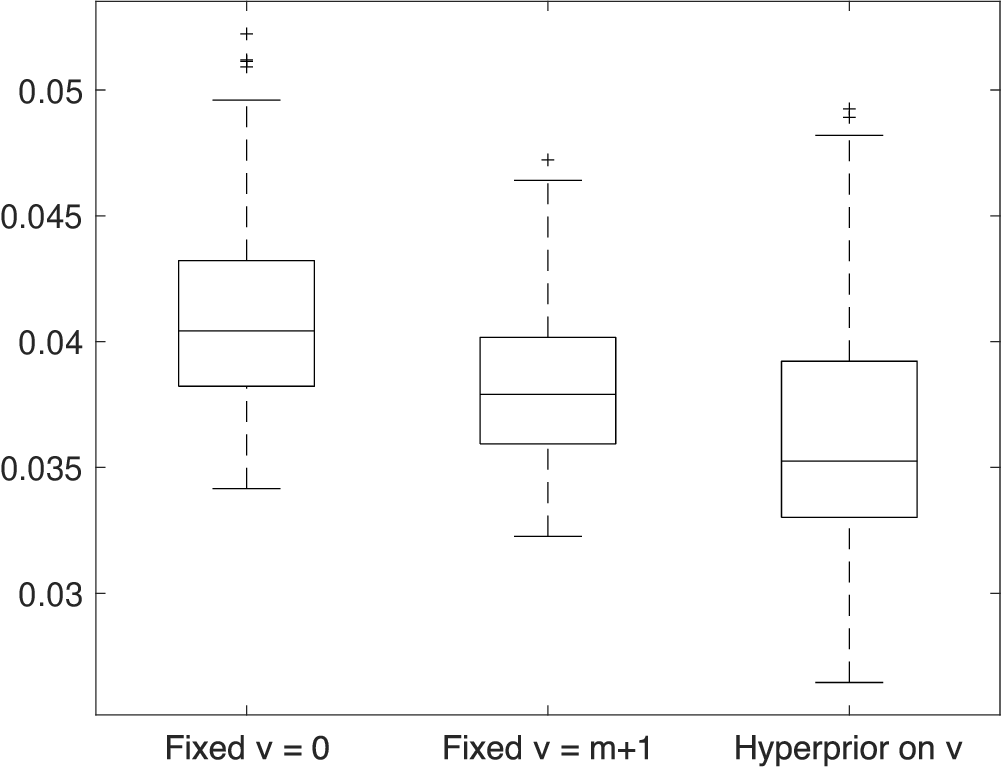}}\hfill%
	{\includegraphics[width=0.33\linewidth]{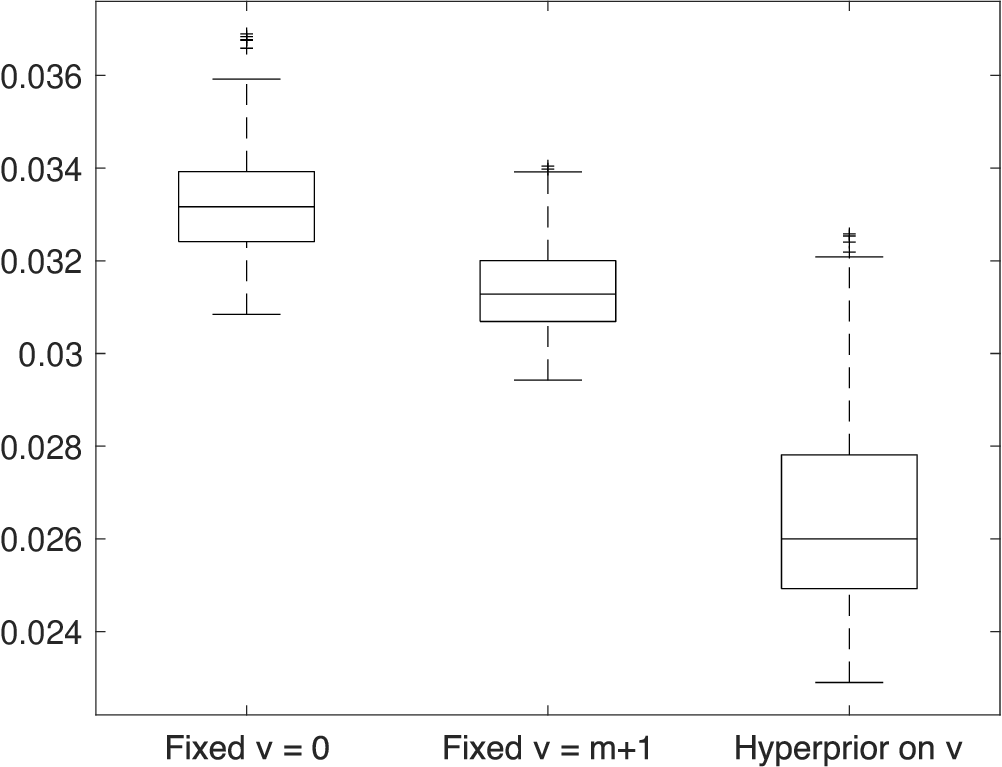}}
 \caption{{Monte Carlo Simulation - RMAD of the covariance matrices. These distributions are obtained by simulating $250$ VAR(1) with dimension $m=15$ and sample size $T=240$. Results are reported for data generated from a Wishart distribution with $\nu = 15$ (left), $\nu=20$ (center), and $\nu = 25$ (right).}}
 %Monte Carlo Simulation - Root Mean Absolute Deviations of the covariance matrices of dimension $m=15$. These empirical distributions are obtained by simulating $250$ VAR(1) of sample size $T=240$. Results are reported separately for data generated from a Wishart with $\nu = 15$ (left)  $\nu=20$ (center), and $\nu = 25$ (right).}
	\label{Fig_Sim_M15_Supp}
\end{figure}

\begin{table}[h!]
    \centering
    \begin{adjustbox}{width=1\textwidth,center=\textwidth}
    \begin{tabular}{c|c c c| c c c| c c c}
    \hline 
    & \multicolumn{3}{c}{$\nu = 15$} & \multicolumn{3}{c}{$\nu = 20$} & \multicolumn{3}{c}{$\nu = 25$} \\
        Horizon & Fixed & Fixed & Hyperprior & Fixed & Fixed & Hyperprior & Fixed & Fixed & Hyperprior  \\
         & $\nu =0$ & $\nu = m+1$ & & $\nu =0$ & $\nu = m+1$ & & $\nu =0$ & $\nu = m+1$ \\
         \hline
     1 &      0.5670   &   0.9970   &   0.9940 &     0.3370   & 0.9970  &  0.9880&    0.3020   & 0.9980  &  0.9880 \\
     3 &     0.6660   &   1.0040   &   1.0070 &     0.3100   & 0.9990  &  0.9930&     0.2850  &  0.9950   & 0.9730 \\
     5 &     0.6910   &   0.9940   &   0.9930 &     0.3540   & 0.9810  &  0.9170&     0.4180  &  0.9720   & 0.8350 \\
     7 &    1.0690   &   0.9750   &   0.9600 &     0.5880   & 0.9580  &  0.8200&     0.8550  &  0.9520   & 0.7270 \\
    \hline 
    \end{tabular}
    \end{adjustbox}
    \caption{{Monte Carlo Simulation - RMAD of the IRF for four horizons $h=1, 3, 5$ and $7$ by simulating $250$ VAR(1) with dimension $m=15$ and sample size $T=240$. Column Fixed $\nu=0$ provides the RMAD of the IRF, while Columns Fixed $\nu = m+1$ and Hyperprior provide the ratio between the referred priors and the flat prior.}}
    %Monte Carlo Simulation - Root Mean Absolute Deviation of the impulse response functions (IRF) of dimension $m=15$ for four different horizons $h=1, 3, 5, 7$ by simulating $250$ VAR(1) of sample size $T=240$. Column Fixed $\nu=0$ provides the RMAD of the IRF, while Columns Fixed $\nu = m+1$ and Hyperprior provide the ratio between the referred priors and the flat prior.}
    \label{tab:IRF_Sim_M15_Supp}
\end{table}

\section{Convergence Diagnostics}
{In this section of the appendix, we describe the converge analysis we have performed for the different simulation experiments.

The convergence analysis has been done using the R Coda Package \citep{Coda06}. In particular, we provide the chain of the degrees of freedom, the Geweke convergence test, the Gelman–Rubin’s plot, and test statistics. Figure~\ref{fig:Coda_Analysis} provides the chain of the degrees of freedom when $T=30$ and $m$ is equal to $5$ (left), $10$ (center), and $20$ (right). As shown in the Figure, we are able to reach convergence of the chain over the number of iterations.
Table~\ref{tab:Geweke_Diagn} reports the results of Geweke’s convergence test \citep{Geweke:1992gm} for the different simulation experiments. The test statistics
show no convergence issues. To perform the Gelman–Rubin test of convergence, we have run multiple chains with sparse starting points. Figure~\ref{fig:Gelman_Diagn} plots the shrinking factor for different simulation experiments when the burn-in iterations are not discarded and in all cases, we do not see any indication of failed convergence under the proposed loss-based prior.
}

\begin{figure}[h!]
	{\includegraphics[width=0.33\linewidth]{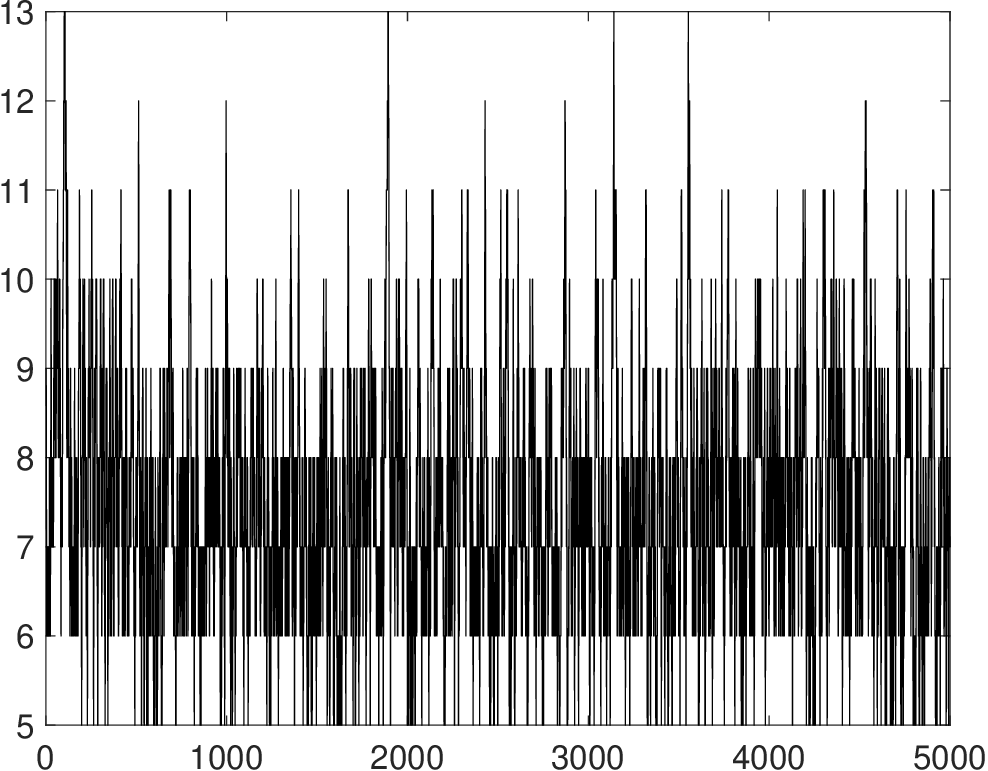}}\hfill%
	{\includegraphics[width=0.33\linewidth]{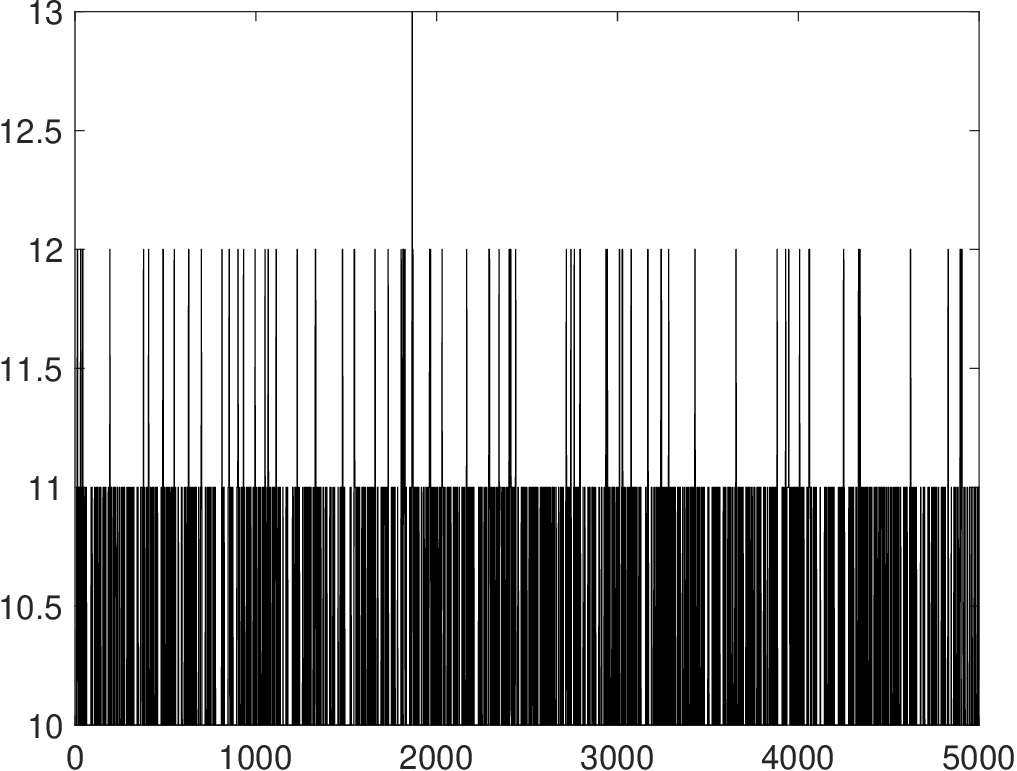}}\hfill%
	{\includegraphics[width=0.33\linewidth]{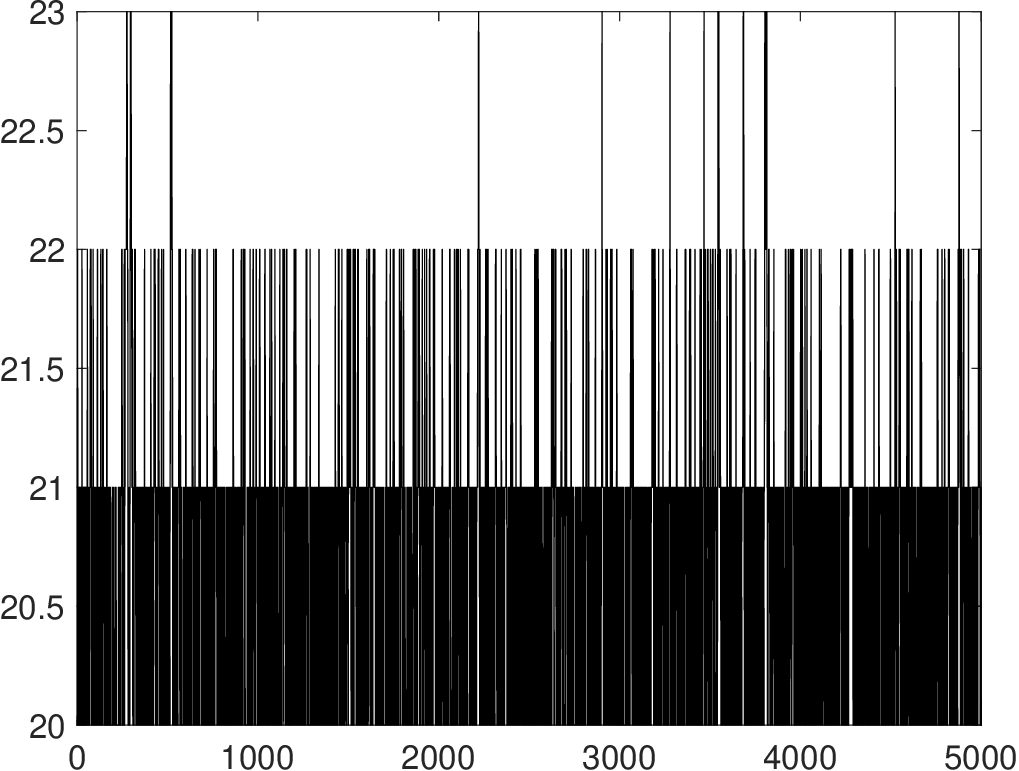}}
 \caption{Posterior Chain of the estimated degrees of freedom when $T=30$ and $m$ is equal to $5$ (left); $10$ (center), and $20$ (right).}
	\label{fig:Coda_Analysis}
\end{figure}

\begin{table}[h!]
    \centering
    \begin{tabular}{cc}
    \hline
    Case: Simulated data & Test \\
    \hline
    $m = 5$ & 0.6713 \\ % 5 chain
    $m = 10$ & 0.9306 \\ % 3 chain
    $m = 20$ &  -0.8532  \\ % 3 chain
    \hline
    \end{tabular}
    \caption{Geweke's test statistics for the posterior chain of the degrees of freedom for different simulation experiments when $T=30$.}
    \label{tab:Geweke_Diagn}
\end{table}

\begin{figure}[h!]
	{\includegraphics[width=0.33\linewidth]{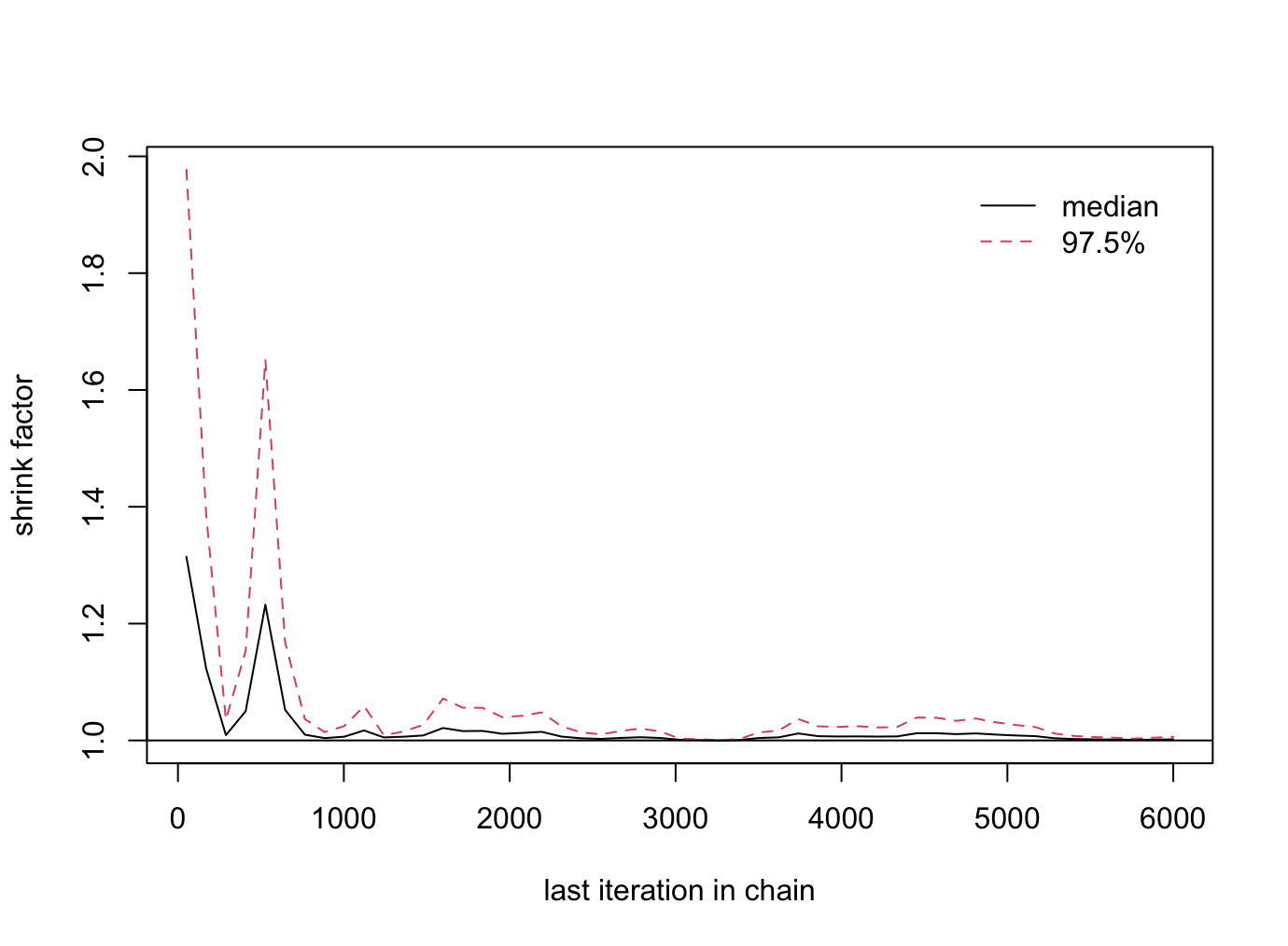}}\hfill%
	{\includegraphics[width=0.33\linewidth]{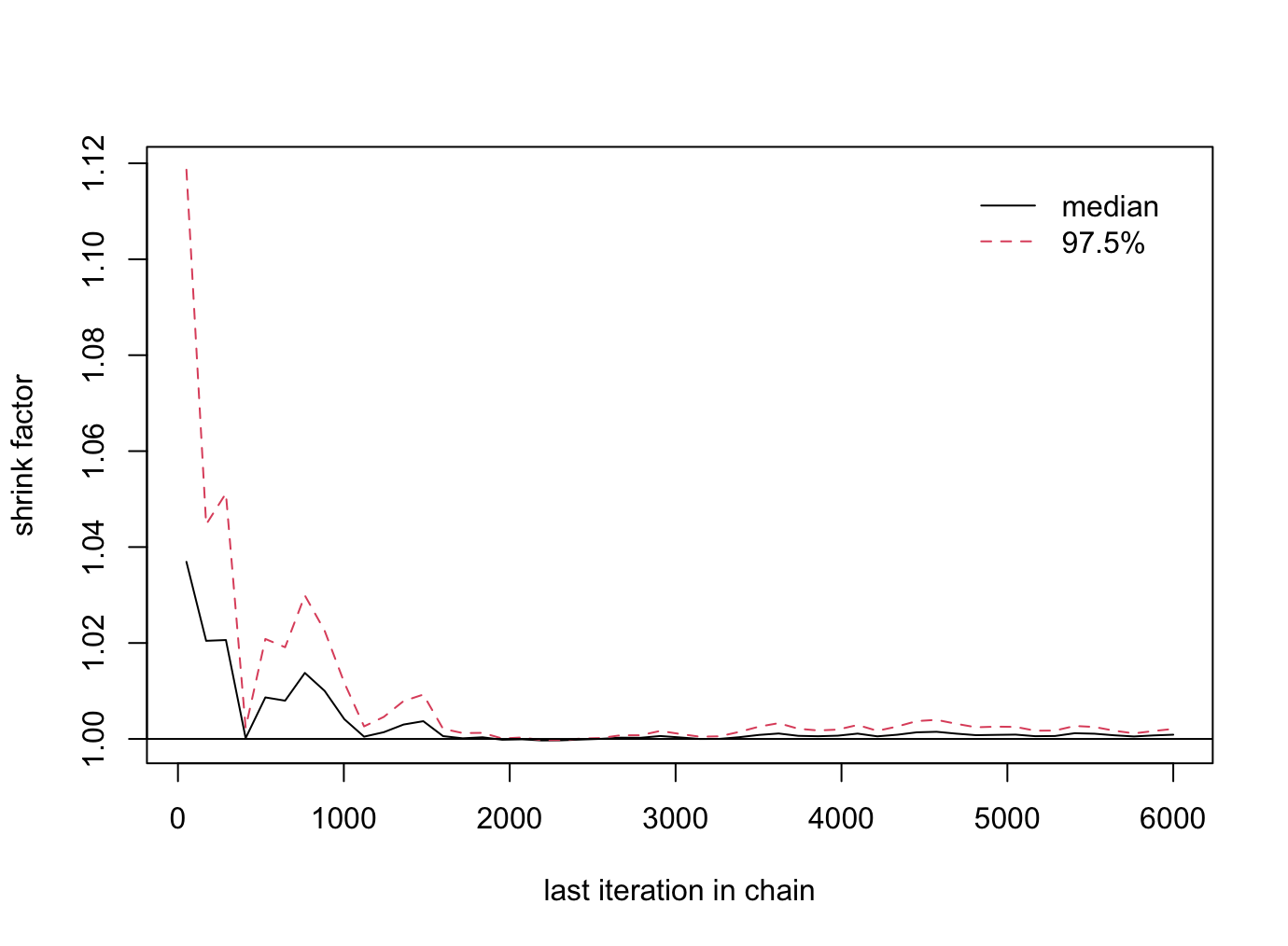}}\hfill%
	{\includegraphics[width=0.33\linewidth]{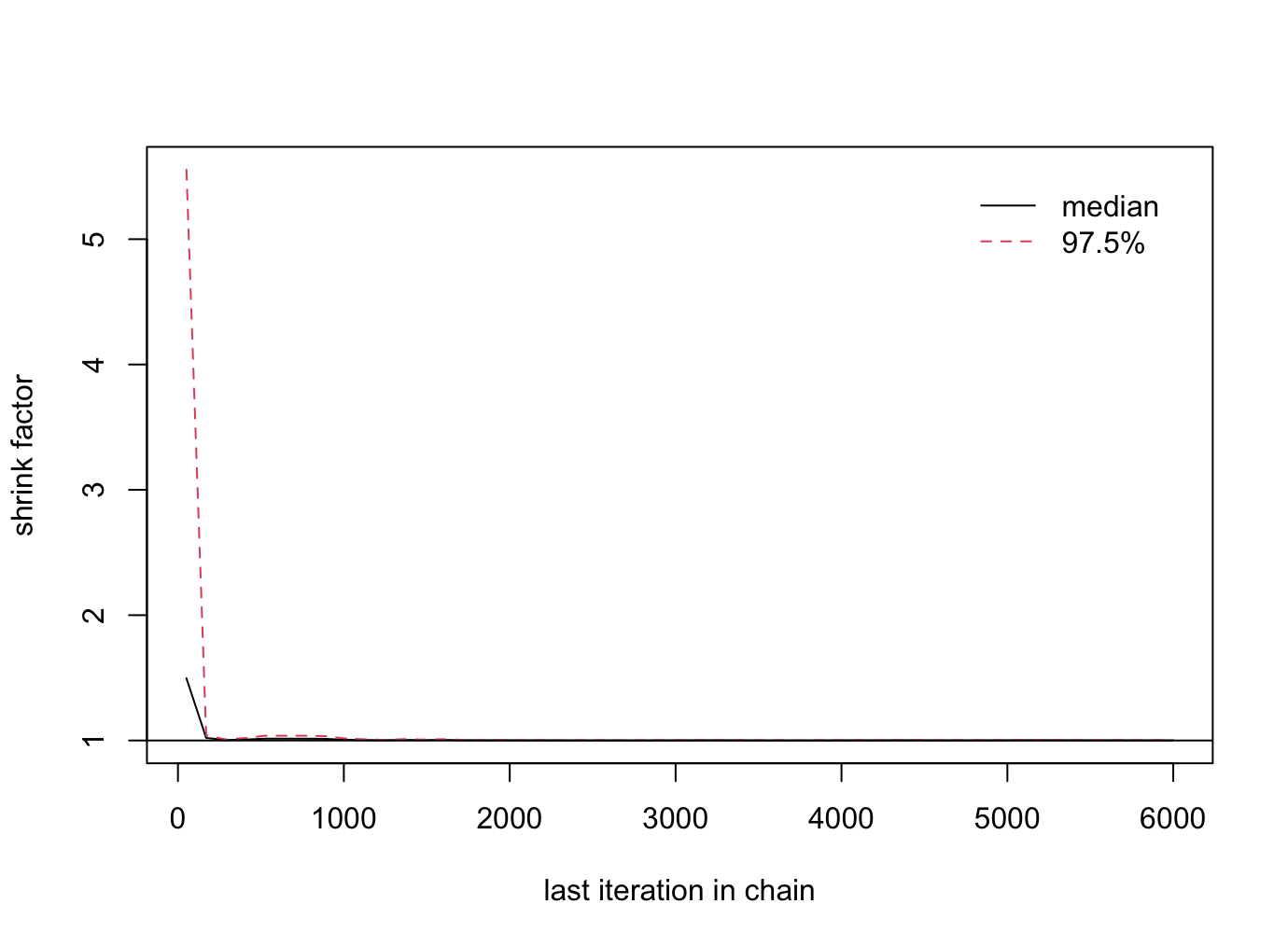}}
 \caption{Gelman–Rubin plot for the estimated degrees of freedom when $T=30$ and $m$ is equal to $5$ (left); $10$ (center), and $20$ (right) when no burn-in iterations are discarded.}
	\label{fig:Gelman_Diagn}
\end{figure}

\end{document}